\documentclass[11pt]{article}
\usepackage{amsfonts}
\usepackage[usenames,dvipsnames]{color}
\usepackage{graphicx, amsmath,latexsym,times,color}
\usepackage{amssymb}
\usepackage{amsbsy}
\usepackage{booktabs}
\usepackage{color, ulem}
\usepackage{enumerate,verbatim}
\usepackage[sort&compress]{natbib}
\usepackage[sf,small,center, compact]{titlesec}
\usepackage{float}
\usepackage{hyperref}
\usepackage{subfigure}
\usepackage{graphicx}
\usepackage{ulem}
\usepackage{url}
\usepackage[lined,algonl,boxed]{algorithm2e}
\usepackage{threeparttable}
\usepackage{xr}
\usepackage{multirow}
\usepackage{wrapfig}
\usepackage{lscape}
\usepackage{rotating}
\usepackage{epstopdf}
\usepackage{subfigure}
\usepackage{mwe}
\usepackage{graphbox} 
\usepackage{array,epsfig}
\usepackage{amsthm}
\usepackage{graphics}
\usepackage{multicol}
\usepackage{cancel}
\usepackage{secdot}

\definecolor{red}{rgb}{0,0,0}
\definecolor{green}{rgb}{0,0,0}
\definecolor{blue}{rgb}{0,0,0}

\DeclareMathOperator*{\argmin}{arg\,min}
\DeclareMathOperator*{\argmax}{arg\,max}
\def\bs{\boldsymbol}

\def\bs{\boldsymbol}

\def\as{\mathrm{a.s.}}
\def\S{\mathcal{S}}

\def\V{\mathcal{V}}

\def\A{\mathcal{A}}
\def\oITR{\pi^{\mathrm{opt}}}
\newcommand{\opt}{\mathrm{opt}} 
\newcommand{\Q}{\mathrm{Q}} 

\floatstyle{ruled}
\newtheorem{theorem}{Theorem}
\newtheorem{lemma}{Lemma}

\newtheorem{remark}{Remark}

\begin{document}

\renewcommand{\baselinestretch}{1.2}
\markboth{\hfill{\footnotesize\rm Xinyi Li and Michael R. Kosorok}\hfill}
{\hfill {\footnotesize\rm Functional ITRs with Imaging Features} \hfill}
\renewcommand{\thefootnote}{}
$\ $\par \fontsize{10.95}{14pt plus.8pt minus .6pt}\selectfont
\vspace{0.8pc} \centerline{\large\bf Functional Individualized Treatment Regimes with Imaging Features}
\vspace{.4cm} \centerline{ 
Xinyi Li$^{a}$ and  
Michael R. Kosorok$^{b}$
\footnote{\emph{Address for correspondence}: 
Xinyi Li, School of Mathematical and Statistical Sciences, Clemson University, Clemson, SC 29634, USA. Email: lixinyi@clemson.edu}} \vspace{.4cm} 
\centerline{\it 
$^{a}$Clemson University, 
$^{b}$University of North Carolina at Chapel Hill} 
\vspace{.55cm}
\fontsize{9}{11.5pt plus.8pt minus .6pt}\selectfont


\begin{quotation}
\noindent {\it Abstract:}
Precision medicine seeks to discover an optimal personalized treatment plan and thereby provide informed and principled decision support, based on the characteristics of individual patients. With recent advancements in medical imaging, it is crucial to incorporate patient-specific imaging features in the study of individualized treatment regimes. We propose a novel, data-driven method to construct interpretable image features which can be incorporated, along with other features, to guide optimal treatment regimes. The proposed method treats imaging information as a realization of a stochastic process, and employs smoothing techniques in estimation. We show that the proposed estimators are consistent under mild conditions. The proposed method is applied to a dataset provided by the Alzheimer's Disease Neuroimaging Initiative.

\vspace{9pt}
\end{quotation}

\fontsize{10.95}{14pt plus.8pt minus .6pt}\selectfont
\thispagestyle{empty}

\label{SEC:introduction}
\setcounter{equation}{0}
\noindent \textbf{1. Introduction} \vskip 0.1in
\renewcommand{\thefigure}{1.\arabic{figure}} \setcounter{figure}{0}

\noindent Precision medicine is an approach to medical treatment that takes into account individual variability in various features, such as genomic characteristics, environment factors, medical history, lifestyle, and current health status. In contrast to traditional treatment strategies, which tend to be uniform and one-size-fits-all, precision medicine aims to provide optimal treatment that is tailored to the unique characteristics of each individual. 
To achieve this goal, individualized treatment rules are developed to formalize the decision-making process that translates patient information into recommended treatment. The identification and application of individualized treatment regimes (ITRs) are crucial areas of investigation in precision medicine research, particularly in the study of chronic diseases or disorders that necessitate diverse medical interventions.
Notably, ITRs have been successfully applied in various clinical contexts, including but not limited to Alzheimer's Disease \citep{Isaacson:19}, diabetes \citep{Luckett:etal:20}, cancers \citep{Thall:etal:07, Zhao:etal:11}, and attention deficit hyperactivity disorder \citep{Pelham:Fabiano:08}.

The field of medical imaging has experienced significant advancements that have had a notable impact on disease and health studies. In addition to characterizing abnormalities and providing morphologic disease information, imaging can also define the imaging phenotype of the disease, using an increasingly diverse set of modern imaging tools and imaging biomarkers \citep{Herold:etal:16}.
The integration of imaging features in precision medicine research has the potential to significantly enhance patient care by enabling effective early detection and diagnosis tools, as well as guiding optimal treatments and lifestyle interventions.
For instance, Alzheimer's disease-related pathophysiologic changes can be visualized up to 15 years before the onset of clinical dementia \citep{Bateman:etal:12}. 
Thus, biomedical imaging and imaging-guided interventions, which provide multiparametric morphologic and functional information,
constitute crucial elements in the infrastructure of precision medicine.
Consequently, it is essential to develop strategies that incorporate imaging data in conjunction with other abundant information in precision medicine research.

In this work, we propose a novel and flexible approach, Functional Individualized Treatment Regimes with Imaging features (FITRI), to adapt semiparametric learning and functional data analysis (FDA) frameworks to precision medicine with abundant features, including medical images, genetic features, environment factors, etc. 
Semiparametric modeling and FDA offer powerful tools to deal with image features, allowing for a flexible range of possible model specifications. 
Specifically, we apply FDA for low-dimensional representation of high-resolution and high-dimensional imaging features. In other words, the imaging features are represented as functional terms in our analysis. We first develop a multi-dimensional functional principal component (FPC) basis expansion tool to approximate the functional imaging features, and then conduct a Q-learning \citep{Kosorok:Laber:19,Clifton:Laber:20} type analysis to estimate the optimal treatment regimes.
We also investigate the theoretical properties of the proposed bases and estimators, and develop computationally efficient algorithms accordingly. 

The proposed work presents a significant contribution to the field of precision medicine research in several ways. First and foremost, the FITRI approach bridges the gap in the current literature in precision medicine studies that incorporate imaging features, which is not a trivial task due to the large volume and complexity of image data.
The development and analysis of FITRI require efficient statistical methods and decision-support algorithms, which will be discussed later. And the proposed estimates provide an efficient data-driven summary statistic for the information contained in the imaging features.
Moreover, the proposed work expands the existing methods that use functional data to construct optimal ITRs. 
Prior studies, including
\cite{McKeague:Qian:14}, \cite{Ciarleglio:etal:15}, \cite{Ciarleglio:etal:16}, \cite{Ciarleglio:etal:18}, \cite{Laber:Staicu:18} and \cite{Park:etal:21}, have developed methods for various FDA models for dense or sparse functional data. All of these works focus on functional data with a one-dimensional index continuum. 
In contrast, the proposed basis tool demonstrates both theoretical superiority and computational efficiency in dealing with imaging data, which can be widely used in the FDA approach dealing with image data \citep{Nathoo:Kong:Zhu:19,Zhu:Li:Zhao:22}.
Such a smoothing tool is lacking in the current literature on functional data indexed by two or higher-dimensional continuums. Not only do the proposed multi-dimensional FPC bases work for index continuum dimensions greater than 1, they also provide practical, data-driven interpretation for the dominant information in the imaging features.  

The rest of the paper is organized as follows. 
Section 2 presents the construction methods for the proposed multi-dimensional FPC basis for multivariate functional data, with a particular focus on the two-dimensional (2D) case.  
In Section 3, we apply the proposed 2D-FPC basis to the functional model and describe the corresponding estimator.
Section 4 provides a presentation of the theoretical findings, while Section 5 elaborates on the implementation aspects.
In Section 6, we furnish empirical investigations to showcase the performance of the proposed estimator. Additionally, Section 7 provides an illustrative example of real-world data in the context of Alzheimer's Disease.
The conclusion of the paper is presented in Section 8, which discusses open issues and future work. 
The supplemental materials provide theoretical details.

\vskip 0.2in \noindent \textbf{2. Bases approximation for imaging features} \vskip 0.1in
\renewcommand{\thetable}{2.\arabic{table}} \setcounter{table}{0} 
\renewcommand{\thefigure}{2.\arabic{figure}} \setcounter{figure}{0}
\renewcommand{\theequation}{2.\arabic{equation}} \setcounter{equation}{0} 
\label{SEC:basis}

We assume the observed data are $\{\mathbf{X}_i, Z_i(\bs{s}), A_i, Y_i\}_{i=1}^n$, which comprise $n$ independent, identically distributed (iid) copies of a trajectory $\{\mathbf{X}, Z(\bs{s}), A, Y\}$, where $\mathbf{X}\in\mathcal{X}\subseteq\mathbb{R}^q$ denotes patient information, $\{Z(\bs{s}): \bs{s}\in \mathcal{V}\}$ is a stochastic process indexed by $\mathcal{V}\subset\mathbb{R}^m$ with $m\geq 1$, $A\in\A=\{0,1\}$ is a binary treatment, and $Y$ is real-valued outcome for which higher values are more desirable. 
For the $i$th subject, the real-valued imaging measure $Z_i(\bs{s}_j)$ is observed at point $\bs{s}_j\in\mathcal{V}$ for $j=1,\ldots,N_s$. 
For example, when $m=2$, $\bs{s}_j=(s_{j1},s_{j2})^{\top}$ could be the coordinates of a voxel center in a two-dimensional 
brain image.
Let $\mu$ be a nonnegative finite measure on $\mathcal{V}$,
and define $L_2(\mathcal{V}, \mu)=\{f: \|f\|_{\mu,2}<\infty\}$, where
	$\|f\|^2_{\mu,2}=\int_{\mathcal{V}} |f(\bs{s})|^2 \mathrm{d}\mu(\bs{s})$.
Note that $L_2(\mathcal{V},\mu)$ is a closed linear space, specifically, a Hilbert space. 
We assume that $Z_1,\ldots,Z_n$ are iid with distribution $f$ such that $\|Z\|_{\mu,2}<\infty$ almost surely.

We employ a Q-learning type analysis. Specifically, define the Q-function as
\begin{equation}
    Q\{\mathbf{x}, z(\bs{s}), a\} = \mathrm{E}\{Y | \mathbf{X} = \mathbf{x}, Z(\bs{s}) = z(\bs{s}), A = a\}. 
\label{EQN:QFunction}
\end{equation}
Assume that we want to parametrize all linear functionals of $Z$, which can be characterized as 
\begin{equation}
\label{EQN: eqn1}
	\int_{\mathcal{V}} Z(\bs{s})g(\bs{s})\mathrm{d}\mu(\bs{s})
\end{equation}
for some $g\in L_2(\mathcal{V},\mu)$.
To be specific, we assume a semiparametric working model:
\begin{align}
    \mathrm{E}\{Y | \mathbf{X} = \mathbf{x}, Z(\bs{s}) =& ~ z(\bs{s}), A = a\} \nonumber \\
	=& ~ \bs{\alpha}_1^{\top}\mathbf{X} + \int_{\V} \beta_1(\bs{s})Z(\bs{s})\mathrm{d}\bs{s} + A\left\{\bs{\alpha}_2^{\top}\mathbf{X} + \int_{\V} \beta_2(\bs{s})Z(\bs{s})\mathrm{d}\bs{s}\right\},
\label{MOD:1}
\end{align}
where $\bs{\alpha}=(\bs{\alpha}_1^{\top}, \bs{\alpha}_2^{\top})^{\top}$ are linear coefficients, and $\beta_1(\bs{s})$ and $\beta_2(\bs{s})$ are coefficient functions/maps.

Efficient tools are essential for addressing coefficient maps in  estimating the Q-function and, subsequently, determining the optimal treatment regime. 
The analysis of imaging data presents challenges such as high dimensionality and resolution, complex geometric structures, including complex boundaries, and spatial heterogeneity.
A variety of nonparametric methods have been developed for analyzing images; see for example, tensor-product-based kernel smoothing \citep{Zhu:Fan:Kong:14}, thin plate spline smoothing \citep{Ramsay:02}, soap film smoothing \citep{Wood:Bravington:Hedley:08}, and multivariate spline over triangulation \citep[MST;][]{Lai:Schumaker:07}.
MST has demonstrated superiority in analyzing multi-dimensional (MD) imaging data, as evidenced by its application in bivariate spline analysis of 2D images \citep{Lai:Wang:13} and trivariate spline analysis of 3D images \citep{Li:etal:2022}.
The construction of MST involves two main steps: constructing triangulation to approximate the entire domain and constructing multivariate splines based on the triangulation. 
In dealing with MD data, the complex and irregular shape poses a primary challenge compared to 1D data, which conventional nonparametric methods struggle to estimate accurately along the boundary, commonly known as the ``leakage'' issue \citep{Ramsay:02, Wang:Ranalli:07}.
However, the triangulation, a set of MD simplices whose union approximates the domain $\mathcal{V}$ well, provides an efficient tool to address the ``leakage'' issue. Based on the constructed triangulation, the multivariate spline can be constructed with explicit formulas, allowing efficient approximation of the MD images.
One can refer to \cite{Lai:Schumaker:07} for additional technical details.

The use of MST presents numerous advantages in the analysis of MD imaging data. However, when applied to the model (\ref{MOD:1}), theoretical and computational limitations emerge. Specifically, to establish the asymptotic consistency of the proposed coefficients estimates and, consequently, the estimated optimal regime, the stability condition for the bases inside the integral is required. Unfortunately, such a condition is lacking in the literature for MST. Moreover, the construction of multivariate splines typically involves hundreds or thousands of bases, rendering the whole system underdetermined, given the limited sample size. To overcome these limitations, we propose an MD-FPC basis based on MST, in a parallel fashion to the construction of 1D FPC bases. By constructing a Reproducing Kernel Hilbert Space (RKHS) based on the multivariate spline space, the low-dimensional representation of MD-FPC can approximate the multivariate spline space, and thus provide a good approximation to the original space $L_2(\mathcal{V},\mu)$. This approach offers theoretical and computational advantages and allows for the handling of the complicated MD image domain.

\vskip 0.1in \noindent \textbf{2.1. Construction of MD-FPC basis} \vskip .10in

In the following, while we develop the methodology for general $L_2$ spaces, including $\mathcal{V}\subset\mathbb{R}^m$, for $m\geq 1$; we also will specialize some of the results for $\mathcal{V}\subset\mathbb{R}^2$, that is, for the 2D-FPC basis.

We first need the following assumptions:
\begin{itemize}
	\item[(A1)] ({\it Space}) Suppose we have a Hilbert space $\mathcal{H}_0\equiv L_2(\mathcal{V},\mu)$, and assume $Z\in \mathcal{H}_0$ with probability one.
	\item[(A2)] ({\it Covariance function}) Assume $\mathrm{E}\|Z\|^2_{\mu,2}<\infty$. Let $V_0(\bs{s},\bs{s}^{\prime})=\mathrm{E}\{Z(\bs{s})-\mathrm{E}Z(\bs{s})\}\{Z(\bs{s}^{\prime})-\mathrm{E}Z(\bs{s}^{\prime})\}$, and suppose $V_0(\bs{s},\bs{s}^{\prime})=\sum_{k=1}^{\infty}\lambda_k\phi_k(\bs{s})\phi_k(\bs{s}^{\prime})$, where $\{\phi_k, k\geq 1\}$ is an orthonormal basis contained in $\mathcal{H}_0$, with corresponding eigenvalues $0 \leq \cdots < \lambda_2 < \lambda_1 < \infty$. 
	\item[(A3)] 
	({\it Eigenvalues}) $\sum_{k=1}^{\infty}\lambda_k<\infty$.
\end{itemize}

\begin{remark}
By Assumptions (A1) and (A2), we have the Karhunen-Lo\'{e}ve expansion $Z(\bs{s})=\mathrm{E}Z(\bs{s})+\sum_{k=1}^{\infty}\xi_k\phi_k(\bs{s})$, where $\xi_k=\int_{\mathcal{V}}\phi_k(\bs{s})\{Z(\bs{s})-\mathrm{E}Z(\bs{s})\}\mathrm{d}\mu(\bs{s})$, $\mathrm{E}\xi_k=0$, $\mathrm{Var}(\xi_k)=\lambda_k$, and $\mathrm{E}(\xi_k\xi_{k^{\prime}})=0$, for $1\leq k<k^{\prime}$.
Let $\mathcal{H}_1$ be the closed linear span of the basis functions $\{\phi_k, k\geq 1\}$, then consequently $\mathcal{H}_1\subset\mathcal{H}_0$ is also a Hilbert space, and Assumption (A3) will then imply that $\Pr(Z-\mathrm{E}Z\in \mathcal{H}_1)=1$.
\end{remark}

Let $\mathcal{B}_1=\{b(\bs{s})\in\mathcal{H}_0: \|b\|_{\mu,2}\leq1\}$, $\mathcal{F}_1=\{\int_{\mathcal{V}}b(\bs{s})\{Z(\bs{s})-\mathrm{E}Z(\bs{s})\}\mathrm{d}\mu(\bs{s}): b(\bs{s})\in\mathcal{B}_1\}$, and 
$
	\mathcal{F}_2
	=\mathcal{F}_1\times\mathcal{F}_1
	=\{\int_{\mathcal{V}\times\mathcal{V}} a(\bs{s})b(\bs{s}^{\prime})\{Z(\bs{s})-\mathrm{E}Z(\bs{s})\}\{Z(\bs{s}^{\prime})-\mathrm{E}Z(\bs{s}^{\prime})\}\mathrm{d}\mu(\bs{s})\mathrm{d}\mu(\bs{s}^{\prime}): a(\bs{s}),b(\bs{s}^{\prime})\in\mathcal{B}_1\}
$.
Theorem \ref{THM:GC} in supplemental materials shows that both
$\mathcal{F}_1$ and $\mathcal{F}_2$ are Glivenko-Cantelli classes of functions of $Z$.
Let $\mathcal{S}$ be the closed linear space spanned by MST in $\mathcal{H}_0$.
Let $G_N$ be the operator that projects onto $\mathcal{S}$, where the index $N$ is the number of spline bases that grow with the sample size $n$.
We need the following additional assumption:
\begin{itemize}
	\item[(A4)] 
({\it Convergence of spline space}) Assume that for any $g\in \mathcal{H}_1$ such that $\|g\|_{\mu,2}\leq 1$, we have that $\|G_Ng-g\|_{\mu,2}\rightarrow 0$ as $n\rightarrow\infty$.
\end{itemize}

\begin{remark}
Assumption (A4) assumes the convergence of the spline space, which is a regular conclusion in the MST literature; see, for example, \cite{Lai:Wang:13} and \cite{Li:etal:2022}.
\end{remark}

Denote the empirical variance function as $V_n(\bs{s},\bs{s}^{\prime})=n^{-1}\sum_{i=1}^n\{Z_i(\bs{s})-\bar{Z}_n(\bs{s})\}\{Z_i(\bs{s}^{\prime})-\bar{Z}_n(\bs{s}^{\prime})\}$, where $\bar{Z}_n=n^{-1}\sum_{i=1}^n Z_i$. 
We first obtain the sequence of theoretical bases based on $V_n(\bs{s},\bs{s}^{\prime})$. 
Define
\begin{equation}
    \widehat{\phi}_{n1}(\bs{s})
	=\argmax_{f\in\mathcal{B}_1\cap \mathcal{S}}\int_{\mathcal{V}\times\mathcal{V}}f(\bs{s})f(\bs{s}^{\prime})V_n(\bs{s},\bs{s}^{\prime})\mathrm{d}\mu(\bs{s})\mathrm{d}\mu(\bs{s}^{\prime}).
\label{EQN:phi_n1_hat}
\end{equation}
As we show in Theorem \ref{THM:fn1hat} in Section A in supplemental materials,
$\widehat{\phi}_{n1}$ converges to
$\phi_{1}$ up to sign almost surely, as $n\to\infty$. In addition, we define 
\begin{eqnarray*}
    \widehat{\mathcal{H}}_{n1}
	&=&\{\text{closed linear span of }\{\widehat{\phi}_{n1}(\bs{s})\} \text{ in } \mathcal{H}_0\}, 
 \\
	\widehat{\phi}_{n2} (\bs{s})
	&=& \argmax_{f\in\mathcal{B}_1\cap \mathcal{S}\cap\widehat{\mathcal{H}}_{n1}^{\perp}} 
	\int_{\mathcal{V}\times\mathcal{V}}f(\bs{s})f(\bs{s}^{\prime})V_n(\bs{s},\bs{s}^{\prime})\mathrm{d}\mu(\bs{s})\mathrm{d}\mu(\bs{s}^{\prime}),
\end{eqnarray*}
where $\widehat{\mathcal{H}}_{n1}^{\perp}$ denotes the closed orthocomplements of $\widehat{\mathcal{H}}_{n1}$ in $L_2(\mathcal{V},\mu)$.
Similarly, $\widehat{\phi}_{n2}$ converges to $\phi_{2}$ up to sign almost surely, as $n\to\infty$. Repeating the construction recursively, we have for $k=1,\ldots,K$,
\begin{align}
\label{DEF:HK}
	\widehat{\mathcal{H}}_{nk}
	&=\{\text{closed linear span of }\{\widehat{\phi}_{n1},\ldots,\widehat{\phi}_{nk}\} \text{ in }  \mathcal{H}_0\},  
 \nonumber \\
	\mathcal{H}_{0k}
	&=\{\text{closed linear span of }\{\phi_1,\ldots,\phi_k\} \text{ in }  \mathcal{H}_0\},  
 \\
	\widehat{\phi}_{n(k+1)} (\bs{s})
	&= \argmax_{f\in\mathcal{B}_1\cap \mathcal{S}\cap\widehat{\mathcal{H}}_{nk}^{\perp}} 
	\int_{\mathcal{V}\times\mathcal{V}}f(\bs{s})f(\bs{s}^{\prime})V_n(\bs{s},\bs{s}^{\prime})\mathrm{d}\mu(\bs{s})\mathrm{d}\mu(\bs{s}^{\prime}). \nonumber
\end{align}
As shown in Theorem \ref{THM:fHconv} in Section 4.1, $\{\widehat{\phi}_{n1},\ldots,\widehat{\phi}_{nK}\}$ form an orthonormal system on $\mathcal{S}\cap\mathcal{H}_0$, and $\{\widehat{\phi}_{nk}\}_k$ converge to $\{\phi_{k}\}_k$ up to sign almost surely, as $n\to\infty$.

In practice, in the 2D-FPC setting, we obtain the measurements of the image data $\{Z_i(\bs{s_j})\}_{i=1,j=1}^{n,N_s}$ on a finite grid or random set of points, instead of the whole continuum. That means, $V_n(\bs{s},\bs{s}^{\prime})$ in (\ref{EQN:phi_n1_hat}) and (\ref{DEF:HK}) is not computable for practical data settings; consequently, $\{\widehat{\phi}_{nk}\}_k$ is not directly accessible. Therefore, we propose to employ bivariate splines over triangulation (BST) as an initial smoothing tool to construct the 2D-FPC basis function.
In 2D-FPC, we are specializing to the setting where $\mathcal{V}\subset\mathbb{R}^2$.

Specifically, let $\mathbf{B}(\bs{s})=\{B_{J}(\bs{s})\in\mathcal{S}\}_{J\in\mathcal{J}}$ be the BST basis of $\mathcal{S}$, and let $\mathbf{H}=\int_{\mathcal{V}} \mathbf{B}(\bs{s}) \mathbf{B}^{\top}(\bs{s})\mathrm{d}\mu(\bs{s})$ with dimension $|\mathcal{J}|\times |\mathcal{J}|$, where $|\mathcal{J}|$ is the cardinality of the bivariate basis index set $\mathcal{J}$. 
Define a ``pre-smoothed'' variance-covariance matrix by bivariate spline as
\[
    \mathbf{K}_n
    =\mathbf{H}^{-1/2}
    \left\{\int_{\mathcal{V}\times\mathcal{V}}\mathbf{B}(\bs{s})V_n(\bs{s},\bs{s}^{\prime})\mathbf{B}(\bs{s}^{\prime})^{\top}\mathrm{d}\mu(\bs{s})\mathrm{d}\mu(\bs{s}^{\prime})\right\} \mathbf{H}^{-1/2}.
\]
Let $\{\widehat{\bs{\varphi}}_{nk}\}_{k=1}^{p_n}$ be the eigenvectors of $\mathbf{K}_n$, ordered from largest eigenvalue to smallest, where $p_n=|\mathcal{J}|\wedge n$. Then for $k\geq1$, we define
\begin{equation}
    \widehat{\phi}^{\S}_{nk}(\bs{s})=\widehat{\bs{\varphi}}_{nk}^{\top}\mathbf{H}^{-1/2}\mathbf{B}(\bs{s}), ~~~~
    \widehat{\lambda}_{nk}=\int_{\mathcal{V}\times\mathcal{V}}\widehat{\phi}_{nk}(\bs{s})\widehat{\phi}_{nk}(\bs{s}^{\prime})V_n(\bs{s},\bs{s}^{\prime})\mathrm{d}\mu(\bs{s})\mathrm{d}\mu(\bs{s}^{\prime}).
\label{EQN:lambdahat}
\end{equation}

As shown in Theorem \ref{THM:fHconv} (iii) in Section 4.1, $\{\widehat{\phi}_{nk}(\bs{s})\}_k$ and $\{\widehat{\phi}^{\S}_{nk}(\bs{s})\}_k$ are asymptotically equal up to sign almost surely, as $n\to\infty$.
Also, by Theorem \ref{THM:fHconv} (iv) in Section 4.1, we have that $\{\widehat{\lambda}_{nk}\}_k$ converge to $\{\lambda_k\}_k$ almost surely.
Note that the $\widehat{\lambda}_{nk}$ values are also the eigenvalues of $\mathbf{K}_n$. 
Therefore, we obtain that the 
$\{\widehat{\phi}^{\S}_{nk}(\bs{s})\}_k$ functions are the constructed orthonormal bases.

\vskip 0.2in \noindent \textbf{3. Functional learning with imaging features for estimating optimal ITRs} \vskip 0.1in
\renewcommand{\thetable}{3.\arabic{table}} \setcounter{table}{0} 
\renewcommand{\thefigure}{3.\arabic{figure}} \setcounter{figure}{0}
\renewcommand{\theequation}{3.\arabic{equation}} \setcounter{equation}{0} 
\label{SEC:Method}

With all the preparation in Section 2, we are ready to implement functional learning with imaging features for estimating the optimal ITRs.

An individualized treatment rule is a function $\pi: \mathcal{X}\times \V\to\A$ so that under $\pi$, a patient presenting with $\mathbf{X}=\mathbf{x}$ and $Z(\bs{s})=z(\bs{s})$ will be assigned to treatment $\pi(\mathbf{x}, z(\bs{s}))$. We employ the potential outcome framework to define an optimal ITR. Let $Y^{\ast}(a)$ be the potential outcome under treatment $a\in\A$, and under any regime $\pi$, $Y^{\ast}(\pi)=\sum_{a\in\A}Y^{\ast}(a)1\{\pi(\mathbf{X}, Z(\bs{s}))=a\}$. 
For a decision rule $\pi$, let $V(\pi) = \mathrm{E}Y^{\ast}(\pi)$ be the value of $\pi$.
An optimal regime for $Y$, denoted as $\oITR$, satisfies $V(\oITR)\geq V(\pi)$ 
for any other regime $\pi$.

Throughout the paper, we use capital letters, like $\mathbf{X}$, to denote random variables, and lower cases like $\mathbf{x}$ for instances of corresponding random variables. To identify the optimal regimes in terms of the data-generating model, we make the following assumptions:

\begin{itemize}
	\item[(A5)] 
	(Consistency) $Y=Y^{\ast}(A)$.
	\item[(A6)] 
	(Positivity) There exists some constant $c$ such that $\Pr\{A=a | \mathbf{X}, Z(\bs{s})\} \geq c>0$ for each $a\in\A$, $\mathbf{X}\in\mathcal{X}$ and $Z(\bs{s})\in\V$ almost surely.
	\item[(A7)] 
	(Ignorability) $\{Y^{\ast}(-1), Y^{\ast}(1)\} \perp A | \mathbf{X}, Z(\bs{s})$.
\end{itemize}

\begin{remark}
    Assumptions (A5)--(A7) are standard in the context of precision medicine \citep{Robins:04,Hernan:Robins:10}. Assumption (A6) is verifiable in the study. Assumption (A7) implies there are no unmeasured confounders, which generally holds by construction in randomized studies but is not in general empirically verifiable in observational studies \citep{Rosenbaum:84,Rosenbaum:Rubin:83}. 
\end{remark}

Under Assumptions (A5)--(A7), and recalling the definition of the Q-function in (\ref{EQN:QFunction}), the optimal ITR is
$\pi^{\opt}\{\mathbf{x}, z(\bs{s})\} = \argmax_{a\in\A} Q\{\mathbf{x}, z(\bs{s}), a\}$.
Let $\widehat{Q}\{\mathbf{x}, z(\bs{s}), a\}$ be the functional-regression-based estimators fit using the data.
Then the optimal ITR can be estimated through $\widehat{Q}\{\mathbf{x}, z(\bs{s}), a\}$ via
\[
	\widehat{\pi}^{\Q}\{\mathbf{x}, z(\bs{s})\} = \argmax_{a\in\A} \widehat{Q}\{\mathbf{x}, z(\bs{s}), a\}.
\]

Switch back now to functional regression modeling. Recall our model in (\ref{MOD:1}), or equivalently,
assume 
\[
	Y=\bs{\alpha}_1^{\top}\mathbf{X} + A\bs{\alpha}_2^{\top}\mathbf{X} + \int_{\mathcal{V}}\beta_1(\bs{s})Z(\bs{s})\mathrm{d}\mu(\bs{s}) + A\int_{\mathcal{V}}\beta_2(\bs{s})Z(\bs{s})\mathrm{d}\mu(\bs{s}) + \varepsilon,
\]
where $\mathbf{X}$ includes a $\{1\}$ term (intercept) and is in $\mathbb{R}^q$, $\bs{\alpha}_1$, $\bs{\alpha}_2\in\mathbb{R}^q$, $A$ is binary treatment indicator (can assume $A\in\{0,1\}$), $Z(\bs{s})$ is a stochastic process on $\mathcal{H}_0$,
and $\beta_1(\bs{s})$,$\beta_2(\bs{s})\in\mathcal{H}_0$.
Assume $\mathrm{E}\|\mathbf{X}\|^2_2<\infty$, $\mathrm{E}(\varepsilon|\mathbf{X},A,Z)=0$ almost surely, $\mathrm{E}\{\mathrm{Var}(\varepsilon|\mathbf{X},A,Z)\}<\infty$ and $\mathrm{Var}(\varepsilon)=\sigma^2$.

Denote $\bs{\theta}=(\bs{\alpha}_1^{\top},\bs{\alpha}_2^{\top},\beta_1,\beta_2)^{\top}$.
Let the generic covariate vector $\mathbf{W}=(\mathbf{X}^{\top},A\mathbf{X}^{\top},Z,AZ)^{\top}$,  and define
	$\bs{\theta}(\mathbf{W})=\bs{\alpha}_1^{\top}\mathbf{X} + \bs{\alpha}_2^{\top}A\mathbf{X} + \int_{\mathcal{V}}\beta_1(\bs{s})Z(\bs{s})\mathrm{d}\mu(\bs{s}) + \int_{\mathcal{V}}\beta_2(\bs{s})AZ(\bs{s})\mathrm{d}\mu(\bs{s})$ .
Also, for subject $i$, the observed covariate is $\mathbf{W}_i=(\mathbf{X}_i^{\top},A_i\mathbf{X}_i^{\top},Z_i,A_iZ_i)^{\top}$.
Note that $\bs{\theta}$ is a linear operator, instead of a pure linear coefficient; and $\mathbf{W}$ and $\mathbf{W}_i$ are ``functional matrices'', namely a hybrid of both matrices and functional covariates. We need the following Assumption (A8) for the linear operator $\bs{\theta}$.

\begin{itemize}
	\item[(A8)] 
({\it Linear operator}) $\bs{\theta}(\mathbf{W}_i)=0$ $\as$ $\Rightarrow \|\bs{\theta}\|_{\ast}=0$, where $\|\bs{\theta}\|_{\ast}=\|\bs{\alpha}_1\|_2+\|\bs{\alpha}_2\|_2+\|\beta_1\|_{\mu,2}+\|\beta_2\|_{\mu,2}$. 
\end{itemize}

\begin{remark}
Assumption (A8) can be implied by $\mathrm{E}\mathbf{W}_i\mathbf{W}_i^{\top}$  being positive definite, since $\mathrm{E}\{\bs{\theta}(\mathbf{W}_i)\}^2=\bs{\theta}^{\top}\mathrm{E}\mathbf{W}_i\mathbf{W}_i^{\top}\bs{\theta}=0$ if and only if $\bs{\theta}=\bs{0}$.
\end{remark}

We employ the 2D-FPC basis proposed in Section 2 to approximate the covariates $Z$ and $AZ$. To be specific, denote the 2D-FPC bases for $Z(\bs{s})$ and $AZ(\bs{s})$ as $\{\widehat{\phi}_{1,nk}(\bs{s})\}_{k}$ and $\{\widehat{\phi}_{2,nk}(\bs{s})\}_{k}$, respectively, then
\begin{equation}
\label{EQN:ZAZ}
    Z(\bs{s}) \approx \sum_{k=1}^{K_1}\widehat{U}_{1,k}\widehat{\phi}_{1,nk}(\bs{s}),~~~~
    AZ(\bs{s}) \approx \sum_{k=1}^{K_2}\widehat{U}_{2,k}\widehat{\phi}_{2,nk}(\bs{s}),
\end{equation}
where $\widehat{U}_{1,k}=\int_{\mathcal{V}}\widehat{\phi}_{1,nk}(\bs{s})Z(\bs{s}) \mathrm{d}\mu(\bs{s})$, $\widehat{U}_{2,k}=\int_{\mathcal{V}}\widehat{\phi}_{2,nk}(\bs{s})AZ(\bs{s}) \mathrm{d}\mu(\bs{s})$,
and $K_1$ and $K_2$ are basis numbers.
Note that we exclude the expectation terms in Equation (\ref{EQN:ZAZ}), which are absorbed into the intercept terms and have no effect on the estimation procedure. Consequently, $\forall k$, $\widehat{U}_{1,k}$ and $\widehat{U}_{2,k}$ does not necessarily have mean zero.
Fix $1\leq K_1,K_2 <\infty$. Note that $K_1$ and $K_2$ can be zero, in which case the corresponding terms vanish in the functional regression model. To avoid such trivialities, we will require $K_1,K_2\geq1$. As for the selection of $K_1$ and $K_2$, we discuss this in detail in Section 5.1.
Accordingly, the coefficient functions $\beta_1(\bs{s})$ and $\beta_2(\bs{s})$ have the representations $\beta_{\ell}(\bs{s}) = \sum_{k=1}^{K_{\ell}}\gamma_{\ell, k}\widehat{\phi}_{\ell, nk}(\bs{s})$ for $\ell=1,2$, where $\gamma_{\ell, k}=\int_{\mathcal{V}}\widehat{\phi}_{\ell, nk}(\bs{s})\beta_{\ell}(\bs{s})\mathrm{d}\mu(\bs{s})$.

For $1\leq K_{\ell}\leq p_n$, let $\widehat{\mathbf{U}}_{\ell}(K_{\ell})=(\widehat{U}_{\ell,1},\ldots,\widehat{U}_{\ell,K_{\ell}})^{\top}$ and $\widehat{\mathbf{U}}_{\ell,i}(K_{\ell})=(\widehat{U}_{\ell,i1},\ldots,\widehat{U}_{\ell,iK_{\ell}})^{\top}$, $\ell=1,2$, where for $1\leq i\leq n$,
\begin{eqnarray*}
    \widehat{U}_{1,ik}&=&\int_{\mathcal{V}}\widehat{\phi}_{1,nk}(\bs{s})Z_i(\bs{s})\mathrm{d}\mu(\bs{s}), ~~1\leq k\leq K_1, \\
    \widehat{U}_{2,ik}&=&\int_{\mathcal{V}}\widehat{\phi}_{2,nk}(\bs{s})A_iZ_i(\bs{s})\mathrm{d}\mu(\bs{s}), ~~ 1\leq k\leq K_2.
\end{eqnarray*}
Note that $\mathrm{Var}(\widehat{U}_{\ell,ik})=\widehat{\lambda}_{\ell,nk}$, where $\widehat{\lambda}_{\ell,nk}$ are defined in (\ref{EQN:lambdahat}) with $V_{\ell,n}(\bs{s},\bs{s}^{\prime})$, $\ell=1,2$, respectively, with
\begin{eqnarray}
    V_{1,n}(\bs{s},\bs{s}^{\prime})&=&\frac{1}{n}\sum_{i=1}^n\{Z_i(\bs{s})-\bar{Z}_n(\bs{s})\}\{Z_i(\bs{s}^{\prime})-\bar{Z}_n(\bs{s}^{\prime})\}, \nonumber \\
    V_{2,n}(\bs{s},\bs{s}^{\prime})&=&\frac{1}{n}\sum_{i=1}^n\{A_iZ_i(\bs{s})-(\overline{AZ})_n(\bs{s})\}\{A_iZ_i(\bs{s}^{\prime})-(\overline{AZ})_n(\bs{s}^{\prime})\}. \label{EQN:Vn}
\end{eqnarray}
In a parallel fashion to $\mathbf{W}$ and $\mathbf{W}_i$, for $1\leq i\leq n$, let $\widehat{\mathbf{W}}=(\mathbf{X}^{\top},A\mathbf{X}^{\top},\widehat{\mathbf{U}}_1^{\top}(K_1),A\widehat{\mathbf{U}}_2^{\top}(K_2))^{\top}$ and $\widehat{\mathbf{W}}_i=(\mathbf{X}_i^{\top},A_i\mathbf{X}_i^{\top},\widehat{\mathbf{U}}_{1,i}^{\top}(K_1),A_i\widehat{\mathbf{U}}_{2,i}^{\top}(K_2))^{\top}$.
Assume the full parameter space is 
	$\bs{\Theta}=\{\bs{\alpha}_1,\bs{\alpha}_2\in\mathbb{R}^q, \beta_1, \beta_2\in \mathcal{H}_0\}$. 
Recall the definitions of $\mathcal{H}_{0K}$ and $\widehat{\mathcal{H}}_{nK}$ in (\ref{DEF:HK}), which are the closed linear spaces spanned by bases $\{\widehat{\phi}_{nk}\}_{k=1}^K$ and $\{\phi_{k}\}_{k=1}^K$, respectively. Define the parameter spaces
\begin{eqnarray}
\label{DEF:ThetaK}
	\bs{\Theta}_{K_1,K_2} &=& \{\bs{\alpha}_1,\bs{\alpha}_2\in\mathbb{R}^q, \beta_1\in\mathcal{H}_{0K_1}, \beta_2\in\mathcal{H}_{0K_2}\}, \\
	\widehat{\bs{\Theta}}_{K_1,K_2} &=& \{\bs{\alpha}_1,\bs{\alpha}_2\in\mathbb{R}^q, \beta_1\in\widehat{\mathcal{H}}_{nK_1}, \beta_2\in\widehat{\mathcal{H}}_{nK_2}\}.\nonumber 
\end{eqnarray}
As shown in Theorem \ref{THM:fHconv} in Section 4.1, $\widehat{\bs{\Theta}}_{K_1,K_2}\to\bs{\Theta}_{K_1,K_2}$ almost surely, as $n\to\infty$. Since $\bs{\Theta}_{K_1,K_2}\subset\bs{\Theta}$, we also have $\widehat{\bs{\Theta}}_{K_1,K_2}\subset\bs{\Theta}$ almost surely, as $n\to\infty$.
Let $\bs{\gamma}_1=(\gamma_{1,1},\ldots,\gamma_{1,K_1})^{\top}$ and $\bs{\gamma}_2=(\gamma_{2,1},\ldots,\gamma_{2,K_2})^{\top}$. With a slight abuse of notation, denote $\widetilde{\bs{\theta}}=(\bs{\alpha}_1^{\top},\bs{\alpha}_2^{\top},\bs{\gamma}_1^{\top},\bs{\gamma}_2^{\top})^{\top}$, and the corresponding parameter space as $\widetilde{\bs{\Theta}}_{K_1,K_2}=\{\bs{\alpha}_1,\bs{\alpha}_2\in\mathbb{R}^q, \bs{\gamma}_1\in\mathbb{R}^{K_1}, \bs{\gamma}_2\in\mathbb{R}^{K_2}\}$.
Obviously, the closed linear span of $\widehat{\bs{\Theta}}_{K_1,K_2}$ is equal to the closed linear span of $\widetilde{\bs{\Theta}}_{K_1,K_2}$ almost surely, as $n\to\infty$; in the sense that, for each $\widehat{\bs{\theta}}\in\widehat{\bs{\Theta}}_{K_1,K_2}$, there exists one and only one $\widetilde{\bs{\theta}}\in\widetilde{\bs{\Theta}}_{K_1,K_2}$ such that 
$
	\widetilde{\bs{\theta}}(\widehat{\mathbf{W}})=\widehat{\bs{\theta}}(\mathbf{W}) 
$
almost surely. Actually, this is an isomorphism by Assumption (A8).

Combining the model (\ref{MOD:1}) and the basis expansion with 2D-FPC basis in (\ref{EQN:ZAZ}), we are able to obtain the estimates
\begin{equation}
    \widehat{\bs{\theta}}_n=\argmin_{\bs{\theta}\in\widehat{\bs{\Theta}}_{K_1,K_2}}n^{-1}\sum_{i=1}^n \{Y_i-\widehat{\bs{\theta}}(\mathbf{W}_i)\}^2.
\label{EQN:coeff}    
\end{equation}

Consequently, the optimal decision rule is estimated as
\begin{eqnarray}
    \widehat{\pi}
    =\argmax_A\widehat{Q}\{\widehat{\bs{\theta}}_n(\mathbf{W})\}
    &=&I\left(\left[\widehat{Q}\{\widehat{\bs{\theta}}_n(\mathbf{W})\}|(A=1) -
        \widehat{Q}\{\widehat{\bs{\theta}}_n(\mathbf{W})\}|(A=0)\right] > 0\right) \nonumber \\
    &=& I\left(\widehat{\alpha}_2\mathbf{X} + \int_{\mathcal{V}}\widehat{\beta}_2(\bs{s})Z(\bs{s})\mathrm{d}\mu(\bs{s}) > 0\right).
\label{EQN:pihat}
\end{eqnarray} 

Combining all the aforementioned steps, the complete algorithm, which we refer to as the semiparametric functional learning algorithm, is given below in Algorithm \ref{ALG:SFL}.

\normalem{
\begin{algorithm}[htbp]
	\SetKwInOut{Input}{Input}
	\SetKwInOut{Initialization}{Initialization}
	\SetKwInOut{Output}{Output}
	\caption{Semiparametric Functional Learning Algorithm.}
	\Input{$\{\mathbf{X}_i, Z_i(\bs{v}), A_i, Y_i\}_{i=1}^n$.}
	\Output{$\widehat{\bs{\theta}}_n$ and $\widehat{\pi}$.}
	\BlankLine
		\textbf{Step 1.} Based on $\{Z_i(\bs{s}_j)\}_{i=1,j=1}^{n,N_s}$, construct 2D-FPC basis functions $\{\widehat{\phi}^{\S}_{\ell,nk}(\bs{s})\}_{k=1}^{K_{\ell}}$ by (\ref{EQN:lambdahat}) with respect to covariance functions $V_{\ell,n}(\bs{s},\bs{s}^{\prime})$ defined in (\ref{EQN:Vn}), $\ell=1,2$. 
		
		\textbf{Step 2.} Obtain the estimates $\widehat{\bs{\theta}}_n$ by (\ref{EQN:coeff}). 

		\textbf{Step 3.} Obtain the estimated optimal regime $\widehat{\pi}$ through (\ref{EQN:pihat}).
\label{ALG:SFL}
\end{algorithm}}

\vskip 0.2in \noindent \textbf{4. Theoretical results} \vskip 0.1in
\renewcommand{\thetable}{4.\arabic{table}} \setcounter{table}{0} 
\renewcommand{\thefigure}{4.\arabic{figure}} \setcounter{figure}{0}
\renewcommand{\theequation}{4.\arabic{equation}} \setcounter{equation}{0} 
\label{SEC:Theory}

In this section, we establish the asymptotic properties of the bases $\{\widehat{\phi}_{nk}(\bs{s})\}_k$ and $\{\widehat{\phi}_{nk}^{\mathcal{S}}(\bs{s})\}_k$ proposed in Section 2, and the estimated $\widehat{\bs{\theta}}$ and estimated optimal regime $\widehat{\pi}$ proposed in Section 3. Our main results are summarized in Theorems \ref{THM:fHconv} and \ref{THM:Conv}.

\vskip 0.1in \noindent \textbf{4.1. Theoretical results for 2D-FPC basis} \vskip .10in
\label{SEC:Theory_2DFPCA}

The main results are stated below. Theorem \ref{THM:fHconv} states that as $n\to\infty$, the proposed bases $\{\widehat{\phi}_{nk}(\bs{s})\}_k$ converge to the theoretical ones up to sign almost surely, the projection of the closed linear span of the proposed bases $\{\widehat{\phi}_{nk}(\bs{s})\}_k$ converge to the theoretical ones almost surely, the proposed ``pre-smoothed'' bases $\{\widehat{\phi}_{nk}^{\mathcal{S}}(\bs{s})\}_k$ are asymptotically equal to $\{\widehat{\phi}_{nk}(\bs{s})\}_k$ up to sign, and the eigenvalues $\{\widehat{\lambda}_{nk}\}_k$ converge to their theoretical values almost surely.

\begin{theorem}
\label{THM:fHconv}
Under Assumption (A3), there exists a sequence $K_n\to\infty$ as $n\to\infty$, such that
\begin{enumerate}[(i)]
	\item for each $k$, there exists a sign sequence $\{S_{nk}\}$ such that $\max_{1\leq k\leq K_n} \|S_{nk}\widehat{\phi}_{nk}-\phi_k\|_{\mu,2}\xrightarrow{\as}0$; 
	\item for any $g_n\in \mathcal{H}_0$
    with $\limsup_{n\to\infty}\|g_n\|_{\mu,2}<\infty$, $\|\widehat{\mathcal{H}}_{nK_n}g_n - \mathcal{H}_{0K_n}g_n\|_{\mu,2}\xrightarrow{\as}0$ ~{ and }~
	$\|\widehat{\mathcal{H}}_{nK_n}^{\perp}g_n - \mathcal{H}_{0K_n}^{\perp}g_n\|_{\mu,2}\xrightarrow{\as}0$;
	\item $\{\widehat{\phi}_{nk}(\bs{s})\}_k$ and $\{\widehat{\phi}^{\S}_{nk}(\bs{s})\}_k$ are asymptotically equal up to sign;
	\item $\max_{k} |\widehat{\lambda}_{nk}-\lambda_k|
	\xrightarrow{\as} 0$.
\end{enumerate}
\end{theorem}
\begin{proof} [Sketch Proof of Theorem \ref{THM:fHconv}]
To show Theorem \ref{THM:fHconv}, we need to show the following conclusions in sequence:
\begin{enumerate} [(a)]
    \item $\mathcal{F}_1$ and $\mathcal{F}_2$ are Glivenko-Cantelli classes;
    \item Recall that $\widehat{\phi}_{n1}=\arg\max_{f\in\mathcal{B}_1\cap\mathcal{S}}\int_{\mathcal{V}\times\mathcal{V}}f(\bs{s})f(\bs{s}^{\prime})V_n(\bs{s},\bs{s}^{\prime})\mathrm{d}\mu(\bs{s})\mathrm{d}\mu(\bs{s}^{\prime})$. Then $\exists$ a sign sequence $\{S_{n1}: \forall n, S_{n1}\in\{-1,1\}\}$ such that $\|S_{n1}\widehat{\phi}_{n1}-\phi_1\|_{\mu,2} \xrightarrow{\as}0$, as $n\to\infty$.
    \item Suppose for some $1\leq K<\infty$, $\exists$ sign sequences $\{S_{n1}, S_{n2}, \ldots, S_{nK}: S_{nk}\in\{-1,1\}, k=1,\ldots,K\}_n$ such that $\max_{1\leq k\leq K} \|S_{nk}\widehat{\phi}_{nk}-\phi_k\|_{\mu,2}  \xrightarrow{\as}0$, as $n\to\infty$, where $\widehat{\phi}_{nk}\in\mathcal{B}_1\cap \mathcal{S}$, $\forall 1\leq k\leq K$ and $\{\widehat{\phi}_{n1},\ldots,\widehat{\phi}_{nK}\}$ are orthogonal in $\mathcal{H}_0$. Let $\{g_n\}\in\mathcal{H}_0$ be a sequence satisfying $\limsup_{n\to\infty}\|g_n\|_{\mu,2}<\infty$, then both 
    \[
	\|\widehat{\mathcal{H}}_{nK}g_n - \mathcal{H}_{0K}g_n\|_{\mu,2}\xrightarrow{\as}0 ~~~~\text{ and }~~~~
	\|\widehat{\mathcal{H}}_{nK}^{\perp}g_n - \mathcal{H}_{0K}^{\perp}g_n\|_{\mu,2}\xrightarrow{\as}0, ~~~~\text{ as } n\to\infty.
    \]
    \item Assume for some $1\leq K<\infty$, that $\{\widehat{\phi}_{n1},\ldots,\widehat{\phi}_{nK}\}$ form an orthonormal system on $\mathcal{S}\cap\mathcal{H}_0$ with $\max_{1\leq k\leq K} \|S_{nk}\widehat{\phi}_{nk}-\phi_k\|_{\mu,2} \xrightarrow{\as}0$, as $n\to\infty$, for some sign sequences $\{S_{n1},\ldots,S_{nK}:S_{nk}\in\{-1,1\}, k=1,\ldots,K\}_n$. Let $\widehat{\phi}_{n(K+1)} = \argmax_{f\in\mathcal{B}_1\cap \mathcal{S}\cap\widehat{\mathcal{H}}_{nK}^{\perp}} 	\int_{\mathcal{V}\times\mathcal{V}}f(\bs{s})f(\bs{s}^{\prime})V_n(\bs{s},\bs{s}^{\prime})\mathrm{d}\mu(\bs{s})\mathrm{d}\mu(\bs{s}^{\prime})$.
    Then $\exists$ a sign sequence $\{S_{n(K+1)}\in\{-1,1\}\}_n$ such that $\|S_{n(K+1)}\widehat{\phi}_{n(K+1)}-\phi_{(K+1)}\|_{\mu,2} \xrightarrow{\as}0$, as $n\to\infty$.
\end{enumerate}
The detailed proof and more technical details are given in Section A in supplemental materials. 
\end{proof}

\vskip .10in \noindent \textbf{4.2. Theoretical results for individualized treatment regimes estimation} \vskip .10in
\label{SEC:Theory_esti}

Let $\bs{\theta}_0=(\bs{\alpha}_{01},\bs{\alpha}_{02},\beta_{01},\beta_{02})\in\bs{\Theta}$ be the true parameter value.
We require further the following Assumptions (A9)--(A10) to develop the consistency results of the estimates:

\begin{itemize}
	\item[(A9)] 
(Covariates) $\mathrm{E}\|\mathbf{X}\|^2<\infty$. 
The eigenvalues of $\mathrm{E}(\mathbf{X}\mathbf{X}^{\top})$ are bounded away from 0 and infinity.
	\item[(A10)] 
(Coefficients) $\|\alpha_{01}\|<\infty$, $\|\alpha_{02}\|<\infty$, $\beta_{01},\beta_{02}\in \mathcal{H}_0$, $\|\beta_{01}\|^2_{\mu,2}<\infty$ and $\|\beta_{02}\|^2_{\mu,2}<\infty$.
\end{itemize}
	
In what follows, $P$ denotes taking the expectation over a single observation $(\mathbf{W},Y)$.
The following 
Theorem \ref{THM:Conv} (i) gives the consistency result for $\widehat{\bs{\theta}}_n$,
and Theorem \ref{THM:Conv} (ii), which is similar to results found in \cite{Qian:Murphy:11} but generalized to our setting, establishes asymptotic optimality of the treatment regime estimated from the forgoing regression model. 

\begin{theorem}
\label{THM:Conv}
Under Assumptions (A5) -- (A10), as $n\to\infty$, 
\begin{enumerate}[(i)]
    \item $P\left\{\widehat{\bs{\theta}}_n(\mathbf{W})-\bs{\theta}_0(\mathbf{W})\right\}^2 \xrightarrow{\as} 0$;
	\item $V(\pi^{opt}) - V(\widehat{\pi}) \xrightarrow{\as} 0$. 
\end{enumerate}
\end{theorem}
\begin{proof} [Sketch Proof of Theorem \ref{THM:Conv}]
The conclusion in (i) follows from the facts of the convergence of bases, the negligibility of the basis approximation error caused by the basis cutoff, and the fact that $\limsup\|\widehat{{\bs{\theta}}}_n\|_{\ast}$ is bounded. The conclusion in (ii) now follows from \cite{Qian:Murphy:11} directly. See the detailed proof and more technical details in Section B in supplemental materials.
\end{proof}

\vskip .20in \noindent \textbf{5. Implementation} \vskip .10in
\renewcommand{\thetable}{5.\arabic{table}} \setcounter{table}{0} 
\renewcommand{\thefigure}{5.\arabic{figure}} \setcounter{figure}{0}
\renewcommand{\theequation}{5.\arabic{equation}} \setcounter{equation}{0} 
\label{SEC:implem}

This section outlines the practical implementation of the proposed procedure. We begin by discussing the selection criteria for bases numbers, followed by the implementation details of BST.

\vskip .10in \noindent \textbf{5.1. Selection Criteria for Bases Numbers} \vskip .10in
\label{SUBSEC:PVE_PAVE}

Selecting the leading principal components (PC) is a popular topic in FDA, particularly for functional data with univariate indexes. To this end, two methods have gained favor in the literature: ranking PC based on the percentage of variance explained \citep[PVE;][]{Kong:Staicu:Maity:16}, and the percentage of association–variation explained \citep[PAVE;][]{Su:Di:Hsu:17}, due to their high testing power and computational efficiency. In this study, we adopt both PVE and PAVE criteria to select the leading PC bases, similar to the univariate functional linear model setting. Specifically, for a given threshold $\alpha$, we define the eigenvalues $\widehat{\lambda}_{nk}$, $k=1,\ldots,K_n$ as in (\ref{EQN:lambdahat}). Then
\begin{itemize}
	\item PVE selects the number of bases $K_n$ such that
	\[
		\frac{\sum_{k=1}^{K_n}\widehat{\lambda}_{nk}}{\sum_{k=1}^{\infty}\widehat{\lambda}_{nk}} \geq \alpha ~~\text{ and }~~
		\frac{\sum_{k=1}^{K_n-1}\widehat{\lambda}_{nk}}{\sum_{k=1}^{\infty}\widehat{\lambda}_{nk}} < \alpha;
	\]
	\item PAVE selects the number of bases $K_n$ such that
	\[
		\frac{\sum_{k=1}^{K_n}\widehat{\lambda}_{nk}\widehat{\gamma}_k^2}{\sum_{k=1}^{K^{\ast}_n}\widehat{\lambda}_{nk}\widehat{\gamma}_k^2} \geq \alpha ~~\text{ and }~~
		\frac{\sum_{k=1}^{K_n-1}\widehat{\lambda}_{nk}\widehat{\gamma}_k^2}{\sum_{k=1}^{K^{\ast}_n}\widehat{\lambda}_{nk}\widehat{\gamma}_k^2} < \alpha,
	\]
where $K^{\ast}_n$ is selected by pre-fitting the truncated model with a high threshold of PVE, i.e., somewhere in the range $[0.95, 0.99]$, and where $\{\widehat{\gamma}_k\}_{k=1}^{K^{\ast}_n}$ are the corresponding PC basis coefficients. 
\end{itemize}
The threshold $\alpha$ is usually chosen to be within $[0.95, 0.99]$ in practice. In our work, we set $\alpha=0.99$ for PVE, PAVE, and the threshold of PVE in the pre-fitting step for PAVE.

In the following, we adopt both PVE and PAVE as selection criteria for bases numbers, denoting the results obtained by the selection criteria PVE and PAVE with the superscripts ``PVE'' and ``PAVE'', respectively.

\vskip .10in \noindent \textbf{5.2. Construction of bivariate spline bases over triangulation} \vskip .10in

The proposed 2D-FPC bases are constructed using the BST as the initial basis. The construction of BST involves the input of parameters for triangulation and spline bases. To investigate the impact of BST on the performance of the proposed 2D-FPC bases, we conducted the simulation study in various settings, such as coarse or fine grids for triangulation, and different smoothness conditions for the splines, with larger degree $d = 5$ or smaller degree $d = 3$. The simulation results indicate that the proposed method is robust regarding the initial bases. 
We briefly discuss the implementation of BST in this section. 

\textit{Triangulation.} Optimal triangulation involves determining both the number and shape of the triangles. According to the literature \citep{Mu:Wang:Wang:18, Yu:etal:21,Li:Wang:Wang:21}, BST performs consistently well when an adequate number of triangles is used. In practice, it is suggested that various numbers of triangles with coarse or fine grids can be attempted to select the optimal number of triangles. Once the number of triangles is determined, triangulated meshes can be constructed using typical triangulation methods, such as Delaunay Triangulation \citep{Bern:Eppstein:95}, which is implemented in the \textsf{R} package \textsf{Triangulation} \citep{Triangulation}.

\textit{BST parameters.} BST with a higher degree $d$ is expected to provide a more accurate approximation and requires more computational power. Throughout, we use $(d,r)=(5,1)$ for BST, as it attains full approximation power asymptotically \citep{Lai:Schumaker:07}. After the triangulation is constructed, BST can be generated using the \textsf{R} package \textsf{BPST} \citep{BPST}.

\vskip .20in \noindent \textbf{6. Empirical study} \vskip .10in
\renewcommand{\thetable}{6.\arabic{table}} \setcounter{table}{0} 
\renewcommand{\thefigure}{6.\arabic{figure}} \setcounter{figure}{0}
\renewcommand{\theequation}{6.\arabic{equation}} \setcounter{equation}{0} 
\label{SEC:simulation}

In this section, we investigate the finite sample performance of the proposed method. 
The experimental data is generated from the underlying model:
\[
	Y_i = \bs{\alpha}_1^{\top}\mathbf{X}_i + \int_{\V} \beta_1(\bs{s})Z_i(\bs{s})\mathrm{d}\bs{s} + A_i\left\{\bs{\alpha}_2^{\top}\mathbf{X}_i + \int_{\V} \beta_2(\bs{s})Z_i(\bs{s})\mathrm{d}\bs{s}\right\} + \varepsilon_i,
	~~i=1,\ldots,n.
\]
For each subject $i$, we generate $\mathbf{X}_i\in\mathbb{R}^q$ independently from $\text{MVN}(\bs{0}_q, \mathbf{\Omega}_q(r))$, where $\{\mathbf{\Omega}_q(r)\}_{\ell,\ell^{\prime}}=r^{|\ell-\ell^{\prime}|}$, and the dependence structure in $\mathbf{X}_i$ is indexed by the autocorrelation parameter $r$; the error $\varepsilon_i$ is independently generated from $N(0,1)$. To simulate the within-image dependence, we generate the imaging data $Z_i(\bs{s})=\sum_{k=1}^2\zeta_{ik}Z_k^{\zeta}(\bs{s})$ at a grid of $N_s=40\times40$ pixels, where $(\zeta_{i1},\zeta_{i2})^{\top}\sim \text{MVN}(\bs{0}_2,\mathbf{I}_2)$, and $Z_1^{\zeta}(\cdot)$ and $Z_2^{\zeta}(\cdot)$ are quadratic and exponential functions, respectively, with forms $Z_1^{\zeta}(\bs{s})=20\{(s_1-0.5)^2+(s_2-0.5)^2\}$ and $Z_2^{\zeta}(\bs{s})=\exp[-15\{(s_1-0.5)^2+(s_2-0.5)^2\}]$. The contour plots of $Z_1^{\zeta}(\cdot)$ and $Z_2^{\zeta}(\cdot)$ are illustrated in Figure \ref{FIG:eg1-1}, left panel.
\begin{figure}[thbp]
\begin{center}
\begin{tabular}{ccc}
	\includegraphics[height=1.36in]{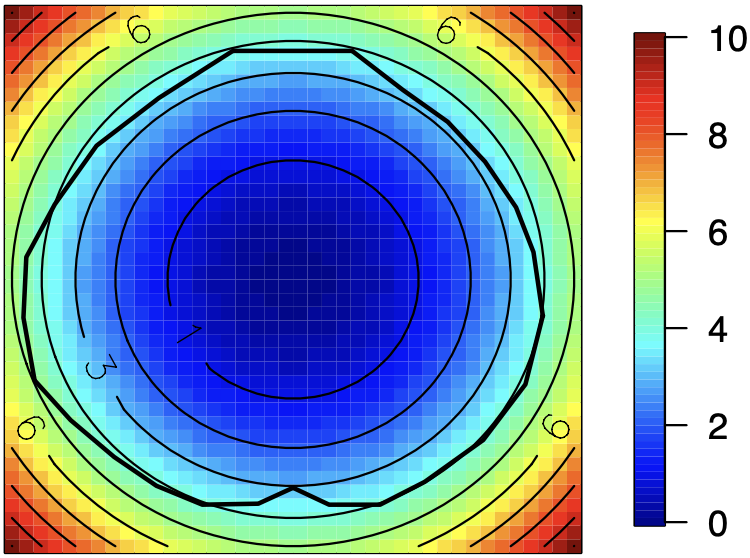} 
	& \includegraphics[height=1.36in]{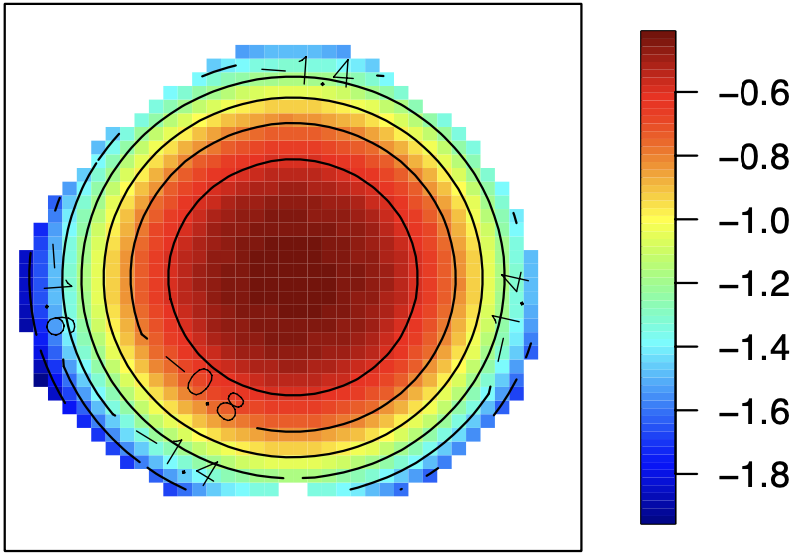}
	& \includegraphics[height=1.36in]{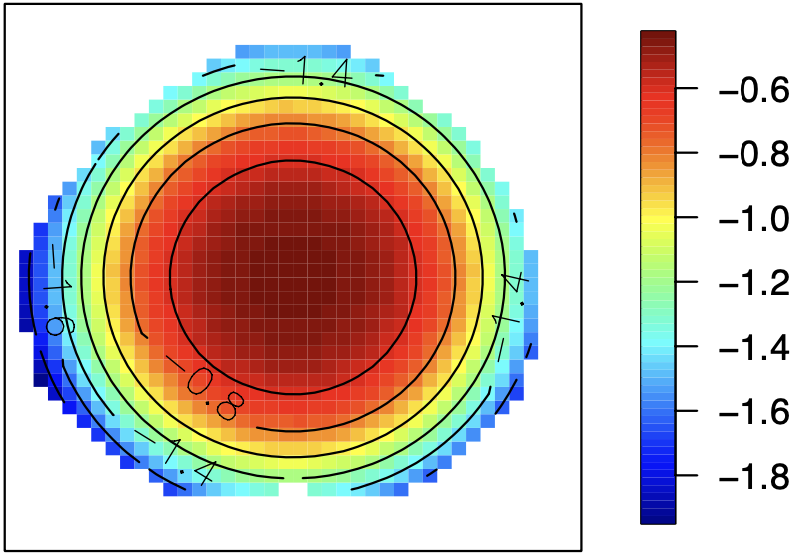} \vspace{-3pt}\\ 
	$Z_1^{\zeta}(\cdot)~~~~~~$ 
	& $\widehat{\phi}_{n1}(\cdot)$ for $Z(\cdot)~~~~~~~~~~~~$ 
	& $\widehat{\phi}_{n1}(\cdot)$ for $AZ(\cdot)~~~~~~~~~~~~$ \\
	\includegraphics[height=1.36in]{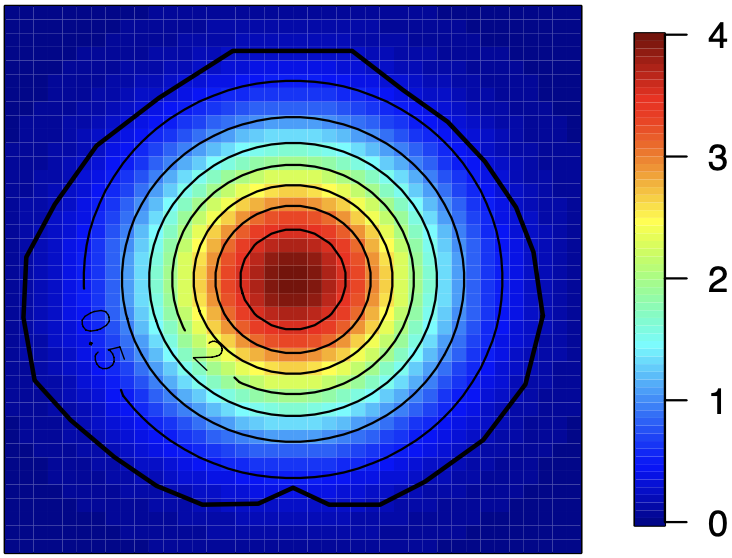} 
	& \includegraphics[height=1.36in]{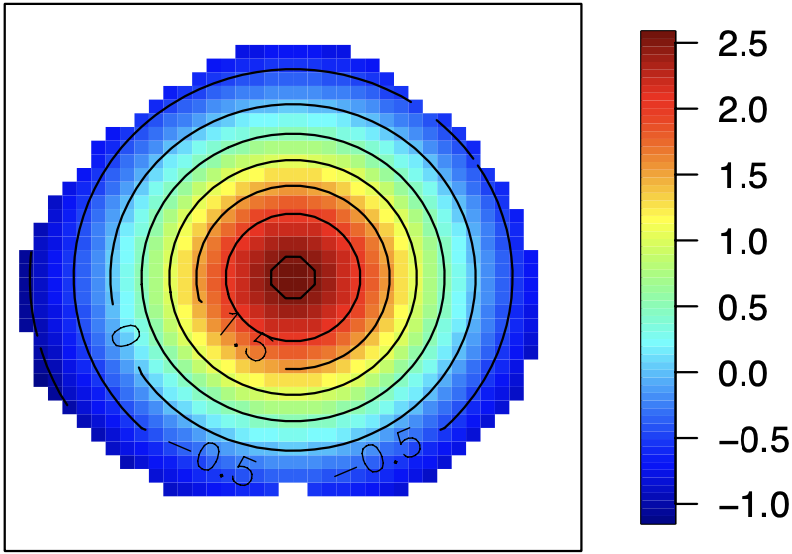}
	& \includegraphics[height=1.36in]{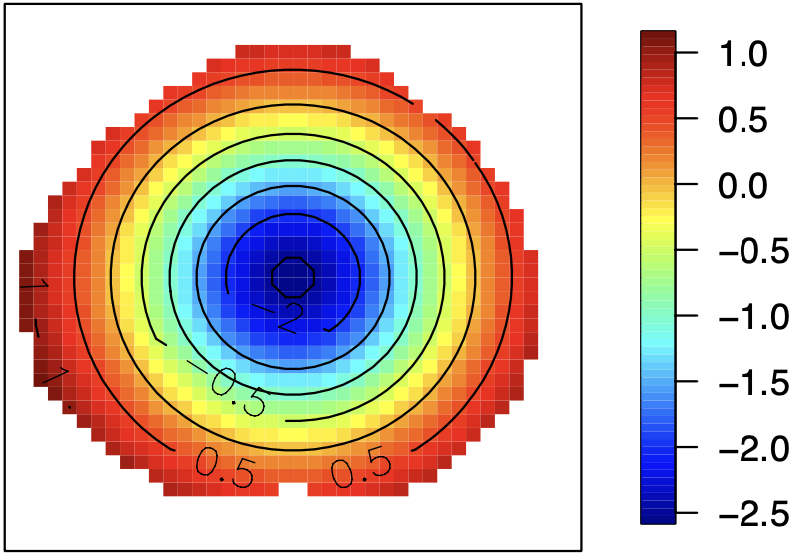} \vspace{-3pt}\\ 
	$Z_2^{\zeta}(\cdot)~~~~~~$ 
	& $\widehat{\phi}_{n2}(\cdot)$ for $Z(\cdot)~~~~~~~~~~~~$ 
	& $\widehat{\phi}_{n2}(\cdot)$ for $AZ(\cdot)~~~~~~~~~~~~$ \\
\end{tabular}
\end{center} \vspace*{-.6cm}
\caption{Contour plots of $Z_j^{\zeta}(\cdot)$ (left panel) and the corresponding bases functions $\widehat{\phi}_{nj}(\cdot)$ for $Z(\cdot)$  (middle panel) and $AZ(\cdot)$ (right panel), $j=1,2$.}
\label{FIG:eg1-1}
\end{figure}
Throughout, we take $q=5$, $\bs{\alpha}_1=\bs{\alpha}_2=\bs{1}_5^{\top}$, and the coefficient maps $\beta_1(\bs{s})=\beta_2(\bs{s})=1$, and we consider the autocorrelation $r=0$ and $r=0.5$, which corresponds to the independent and dependent covariate structures, respectively. 

We evaluate the methods on the accuracy of the estimated ITR with respect to the average marginal mean outcome $V(\widehat{\pi}^{\text{PVE}})$ and $V(\widehat{\pi}^{\text{PAVE}})$, for PVE and PAVE criteria, respectively. For comparison, we also compute the average optimal marginal mean outcome $V(\pi^{opt})$, which works as a benchmark and is only computable in simulations. We also evaluate the methods on the accuracy of coefficient estimation by computing the mean squared errors (MSEs). All the results are based on 100 Monte Carlo replications.

As illustrated in Figure \ref{FIG:eg1-1}, the generated bases functions are the normalized true function up to a sign difference, which validates the conclusion in Theorem \ref{THM:fHconv}. 
Table \ref{TAB:Ypi} presents the marginal mean outcome for the proposed methods. In all settings, the average marginal mean outcomes of PVE and PAVE criteria perform similarly, and both are close to those with optimal treatment regime, $\pi^{opt}$. 
Table \ref{TAB:mse} shows the MSEs for linear coefficients $\bs{\alpha}_1$ and $\bs{\alpha}_2$. With the increase of sample size, the MSEs of all estimates decrease significantly. And the estimates with bases selected by PVE and PAVE perform similarly, regardless of the settings.

\begin{table}[htbp]
\renewcommand{\arraystretch}{0.7}
\begin{center}
\caption{Estimated mean outcome for optimal $\pi$ and proposed method with PVE and PAVE criteria.}
\scalebox{1}{
\begin{tabular}{lcccccc}
    \hline\hline
    & \multicolumn{3}{c}{$r=0$} & \multicolumn{3}{c}{$r=0.5$} \\
	\cmidrule(lr){2-4}\cmidrule(lr){5-7}
    & \multicolumn{1}{c}{$n=100$} & \multicolumn{1}{c}{$n=200$} & \multicolumn{1}{c}{$n=500$} & \multicolumn{1}{c}{$n=100$} & \multicolumn{1}{c}{$n=200$} & \multicolumn{1}{c}{$n=500$} \\
    \hline
    $V(\widehat{\pi}^{opt})$ & 2.791 & 2.814 & 2.798 & 3.433 & 3.442 & 3.456 \\$V(\widehat{\pi}^{\text{PVE}})$ & 2.763 & 2.808 & 2.804 & 3.431 & 3.455 & 3.459 \\$V(\widehat{\pi}^{\text{PAVE}})$ & 2.763 & 2.808 & 2.804 & 3.431 & 3.455 & 3.459 \\
    \hline\hline
\end{tabular}}%
\label{TAB:Ypi}%
\end{center}
\end{table}%

\begin{table}[htbp]
\renewcommand{\arraystretch}{0.7}
\begin{center}
\caption{MSEs for linear coefficients $\bs{\alpha}_1$ and $\bs{\alpha}_2$.}
\label{TAB:mse}%
\begin{tabular}{ccccccccccccc}
    \hline\hline
    \multirow{2}[0]{*}{$r$} & \multirow{2}[0]{*}{$n$} & \multirow{2}[0]{*}{Criteria} & \multicolumn{5}{c}{MSE of $\bs{\alpha}_1$ ($\times 10^{-2}$)} & \multicolumn{5}{c}{MSE of $\bs{\alpha}_2$ ($\times 10^{-2}$)} \\
	\cmidrule(lr){4-8}\cmidrule(lr){9-13}
    & & & $\alpha_{11}$ & $\alpha_{12}$ & $\alpha_{13}$ & $\alpha_{14}$ & $\alpha_{15}$ & $\alpha_{21}$ & $\alpha_{22}$ & $\alpha_{23}$ & $\alpha_{24}$ & $\alpha_{25}$ \\
    \hline
    0     & 100   & PVE   & 2.53  & 1.29  & 1.30  & 1.19  & 1.42  & 6.85  & 1.08  & 1.35  & 1.31  & 1.47 \\
          &       & PAVE  & 2.53  & 1.29  & 1.30  & 1.19  & 1.42  & 6.85  & 1.08  & 1.35  & 1.31  & 1.47 \\
          & 200   & PVE   & 1.13  & 0.65  & 0.68  & 0.47  & 0.60  & 3.73  & 0.54  & 0.54  & 0.58  & 0.61 \\
          &       & PAVE  & 1.13  & 0.65  & 0.68  & 0.47  & 0.60  & 3.73  & 0.54  & 0.54  & 0.58  & 0.61 \\
          & 500   & PVE   & 0.46  & 0.19  & 0.20  & 0.22  & 0.22  & 1.48  & 0.23  & 0.21  & 0.17  & 0.20 \\
          &       & PAVE  & 0.46  & 0.19  & 0.20  & 0.22  & 0.22  & 1.48  & 0.23  & 0.21  & 0.17  & 0.20 \\
    0.5   & 100   & PVE   & 2.53  & 2.01  & 2.56  & 1.89  & 1.27  & 6.85  & 1.71  & 2.01  & 1.79  & 1.89 \\
          &       & PAVE  & 2.53  & 2.01  & 2.56  & 1.89  & 1.27  & 6.85  & 1.70  & 2.01  & 1.79  & 1.89 \\
          & 200   & PVE   & 1.13  & 0.67  & 1.03  & 1.03  & 1.02  & 3.73  & 0.67  & 0.78  & 1.08  & 0.78 \\
          &       & PAVE  & 1.13  & 0.67  & 1.03  & 1.03  & 1.02  & 3.73  & 0.67  & 0.78  & 1.08  & 0.78 \\
          & 500   & PVE   & 0.46  & 0.29  & 0.31  & 0.34  & 0.25  & 1.48  & 0.28  & 0.32  & 0.39  & 0.27 \\
          &       & PAVE  & 0.46  & 0.29  & 0.31  & 0.34  & 0.25  & 1.48  & 0.28  & 0.32  & 0.39  & 0.27 \\
    \hline\hline
\end{tabular}%
\end{center}
\end{table}%

\vskip .20in \noindent \textbf{7. Application to ADNI data} \vskip .10in
\renewcommand{\thetable}{7.\arabic{table}} \setcounter{table}{0} 
\renewcommand{\thefigure}{7.\arabic{figure}} \setcounter{figure}{0}
\renewcommand{\theequation}{7.\arabic{equation}} \setcounter{equation}{0} 
\label{SEC:ADNI}

The data that drives our research comes from the large neuroimaging datasets in the Alzheimer's Disease Neuroimaging Initiative (ADNI, \url{http://adni.loni.usc.edu}).
The longitudinal cohort study in ADNI, which has gone through three phases including ADNI1, ADNI GO, and ADNI2, is a comprehensive neuroimaging study that collected a variety of necessary phenotypic measures, including structural, functional, and molecular neuroimaging, biomarkers, clinical and neuropsychological variables, and genomic information \citep{Weiner:Veitch:15, Petersen:etal:10}. 
These data provide unprecedented resources for statistical methods development and scientific discovery.

We now analyze the records from 441 participants through the ADNI1 and ADNI GO phases. The data contains the following variables: 

\begin{itemize}
    \item Mini-mental state examination (MMSE) scores: response variable, ranging from 15 to 30, where lower values indicate a more severe AD status.
    \item Fludeoxyglucose positron emission tomography (PET) scans: neuroimaging representing brain metabolism activity level and can be used to make early diagnoses of AD, with $79\times95$ pixels,  with the measurements ranging from 0.013 to 2.149. The left panel in Figure \ref{FIG:ADNI_Z} shows the PET images for four randomly selected subjects.
    \item Age: the participants' ages, ranging from 55 to 89 years.
    \item Education: the participants' educational status, ranging from 4 to 20 years. 
    \item Gender: the participants' gender, with 169 female and 278 male. We created a dummy variable with value 1 representing female and 0 for male.
    \item Ethnicity: the participants' ethnic categories, with 12 Hispanic/Latino, 429 not Hispanic/Latino, and 6 unknown. We created a dummy variable with value 1 representing Hispanic/Latino and 0 for others.
    \item Race: the participants' racial categories, with 1 Indian/Alaskan, 7 Asians, 24 Blacks, 413 Whites, and 2 more than one category. We created a dummy variable with value 1 representing white and 0 for others.
    \item Marriage: the participants' marital status, with 35 divorced, 344 married, 12 never married, and 56 widowed. We created a dummy variable with value 1 representing married and 0 for others.
    \item Apolipoprotein (APOE) gene: the number of copies of APOE4 gene, the most prevalent genetic risk factor for AD \citep{Ashford:Mortimer:02}, ranging from 0 to 2. We created two dummy variables, APOE1 and APOE2, to denote those with one and two copies of APOE4 gene, respectively. 
    \item Treatment: During the ADNI1 and ADNI GO study periods, the US FDA-approved therapies for AD symptoms included cholinesterase inhibitors and the NMDA-partial receptor antagonist memantine. Cholinesterase inhibitors, including donepezil, galantamine, and rivastigmine, are prescribed for mild-to-moderate-stage AD. Memantine is prescribed for the treatment of AD either as monotherapy or in combination with one of the cholinesterase inhibitors for moderate-to-severe stage AD \citep{Schneider:11}. We denote by $A=1$ those participants taking one or more combinations of Donepezil (Aricept), Galantamine (Razadyne), Rivastigmine (Exelon), and Memantine (Namenda), while we use $A=0$ for those wihout concurrent medical records, or taking some other treatments or supplements. The distribution of participants with $A=1$ is illstruated in Figure \ref{FIG:trt_VD}.
\end{itemize}

\begin{figure}[htbp]
\begin{center}
	\includegraphics[scale=.1]{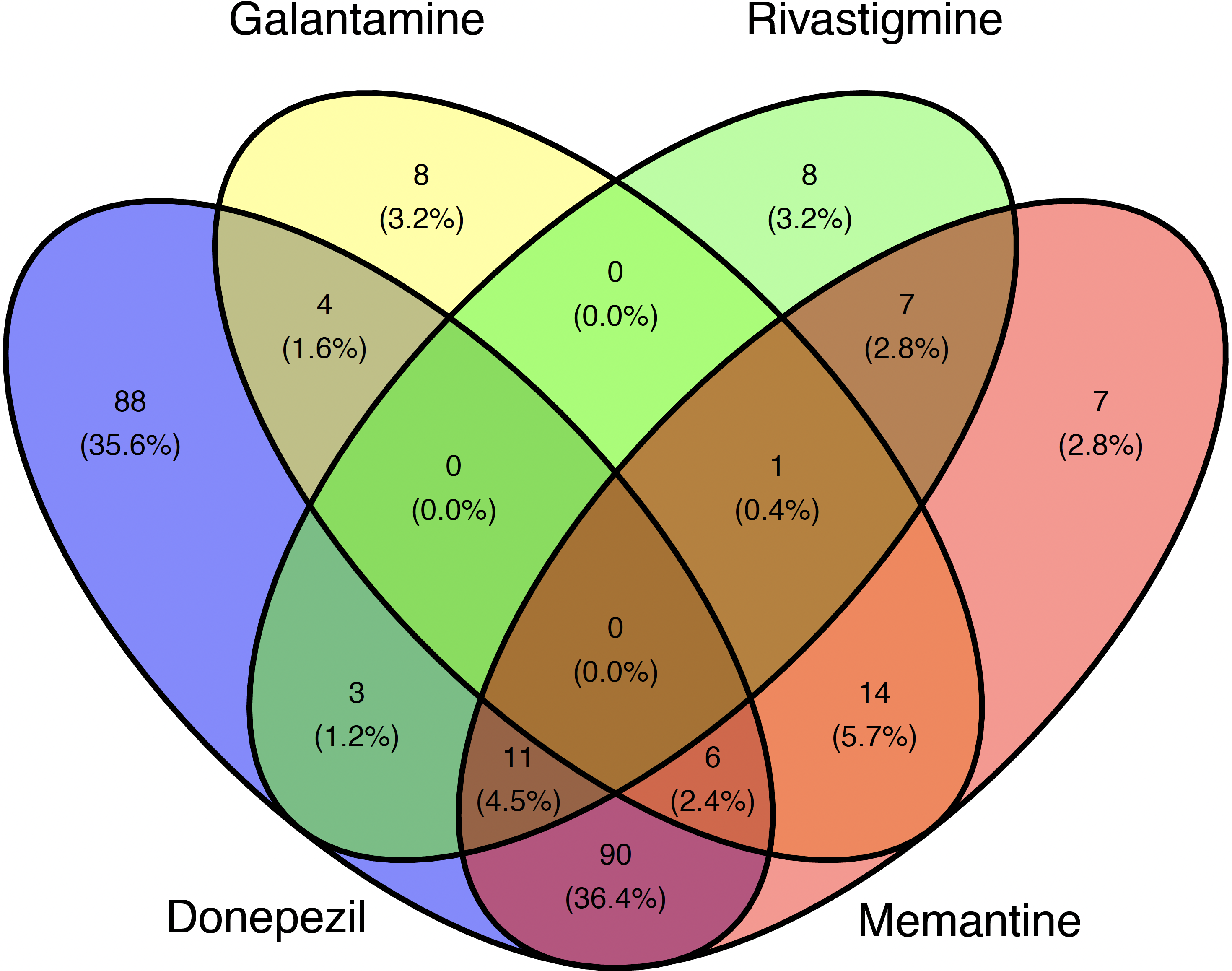} 
\end{center} \vspace*{-.6cm}
\caption{Distribution of patients by treatments (total sample size 441, with 247 $A=1$).}
\label{FIG:trt_VD}
\end{figure}

We apply the proposed method to those data using both PVE and PAVE criteria.
The estimated coefficient map for $\beta_1(\cdot)$ and $\beta_2(\cdot)$ are shown in the top line of Figure \ref{FIG:beta2hat}, in which PVE and PAVE criteria produce similar results, and the estimates of $\beta_2(\cdot)$ illustrate brain structures.
We also display visually the three leading principal component basis maps in the middle and bottom lines of Figure \ref{FIG:beta2hat}. All these estimated basis maps illustrate brain structures.

\begin{figure}[htbp]
\begin{center}
\begin{tabular}{cccc}
    \includegraphics[height=1.25in]{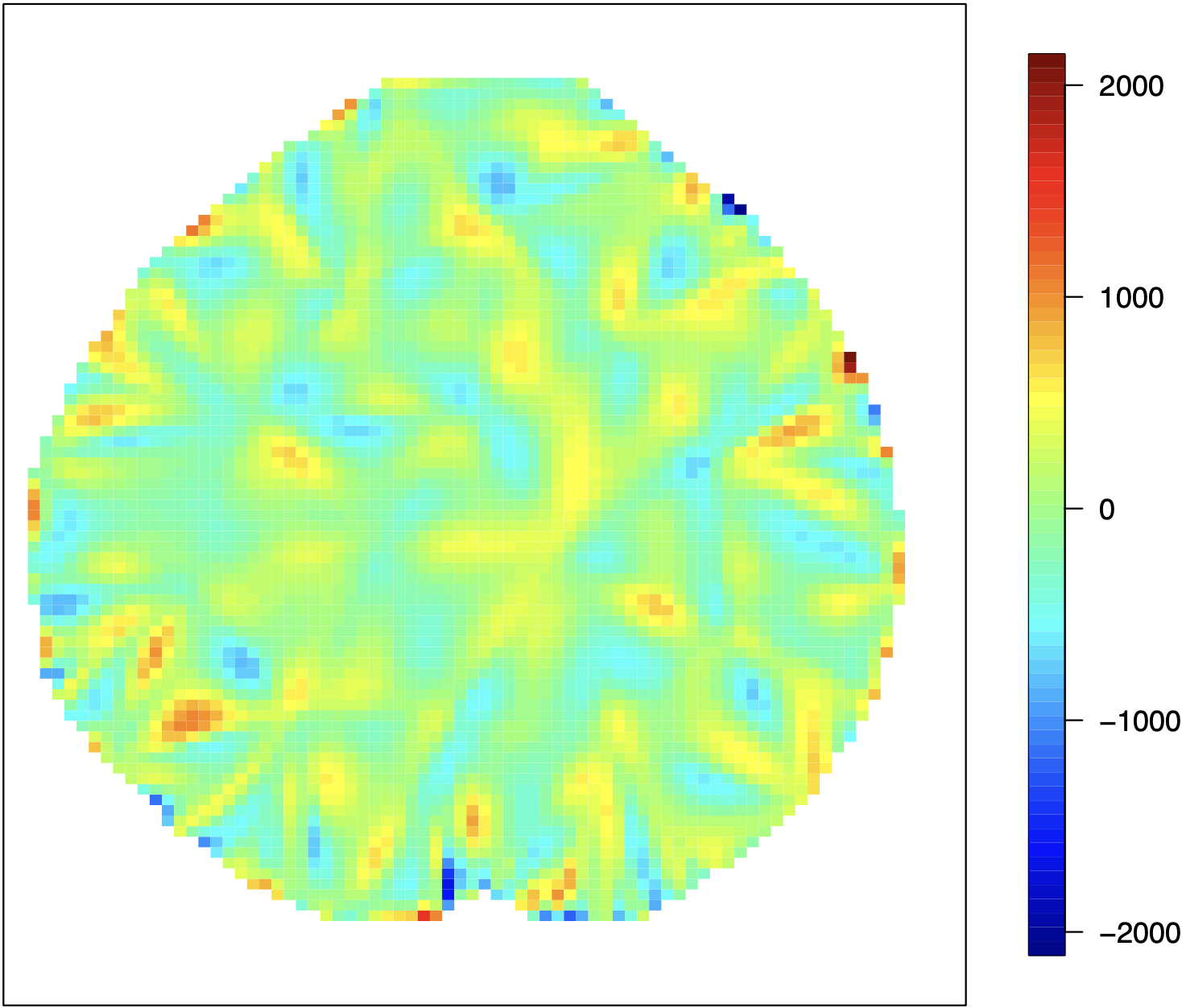}
	& \includegraphics[height=1.25in]{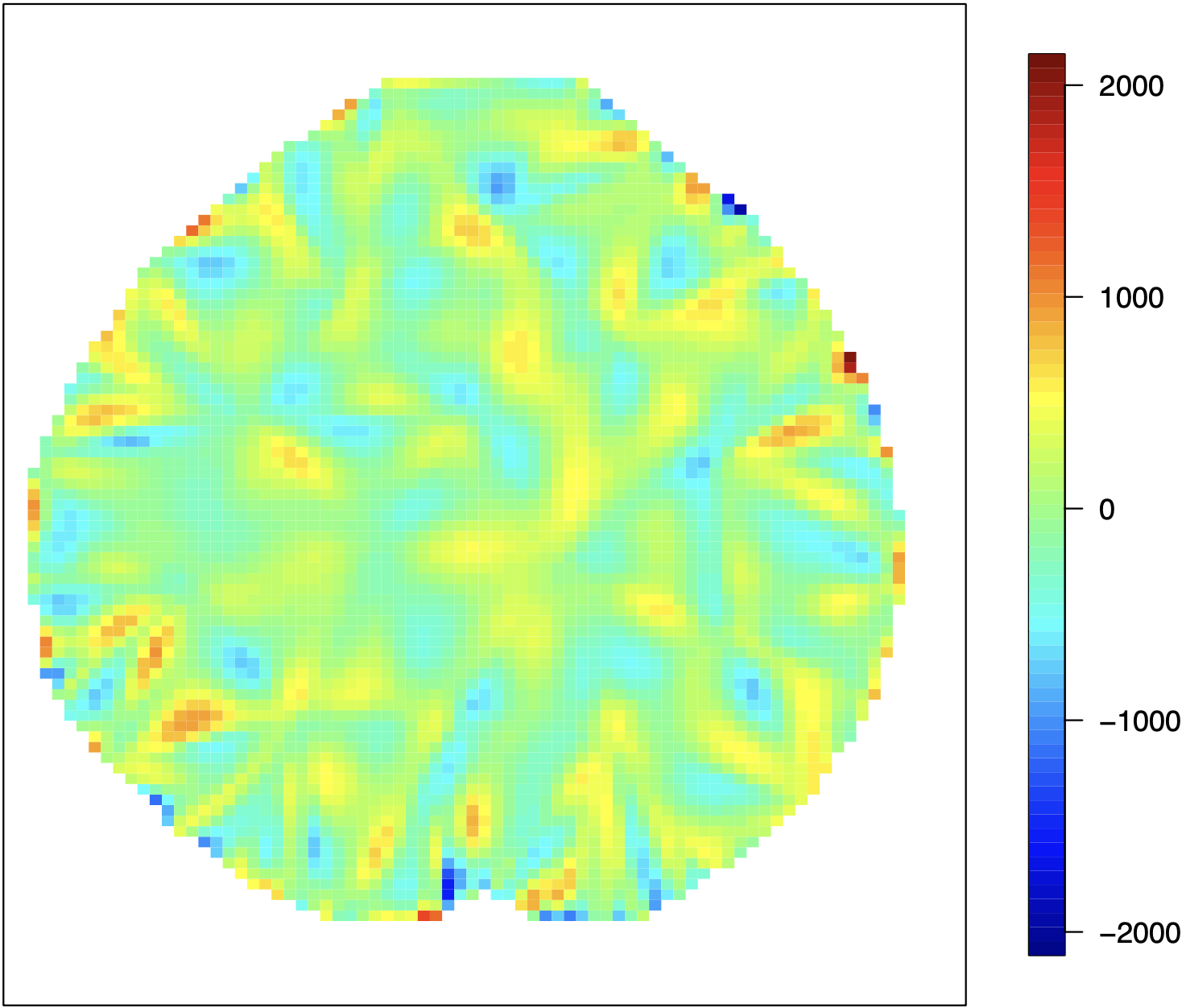}
	& \includegraphics[height=1.25in]{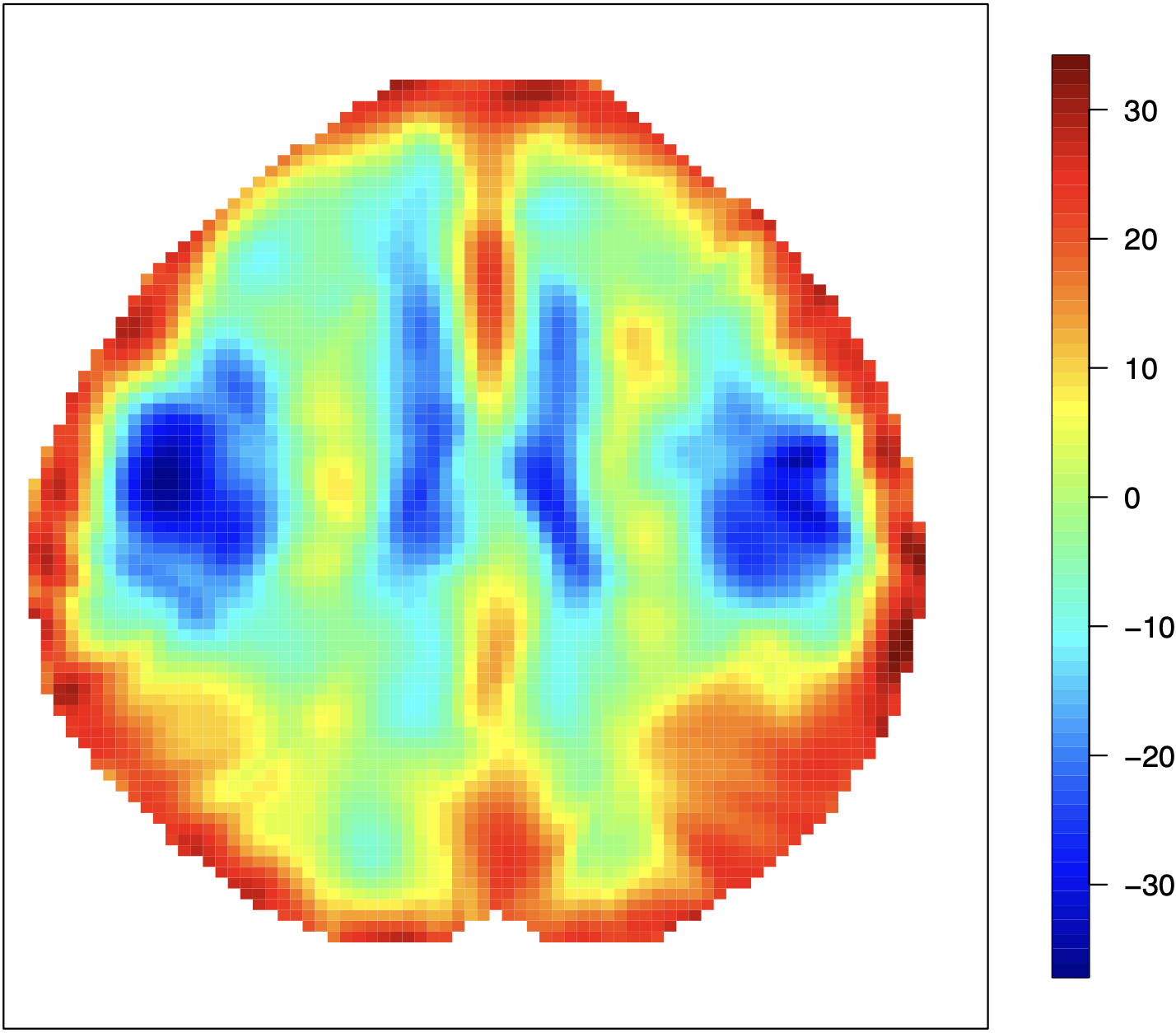}
	& \includegraphics[height=1.25in]{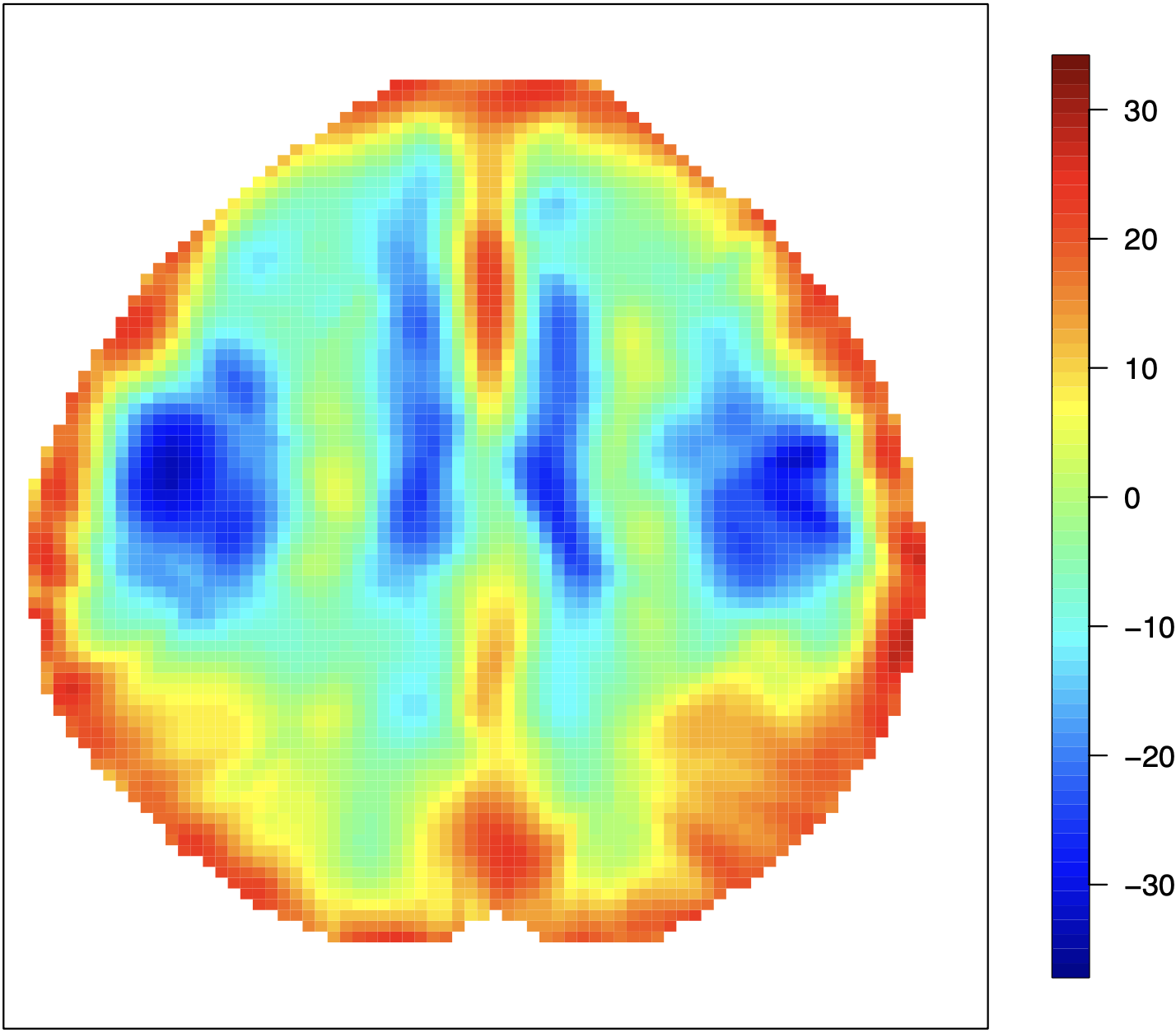} \vspace{-3pt} \\ 
	$\widehat{\beta}_1^{\text{PVE}}(\cdot)~~$
	& $\widehat{\beta}_1^{\text{PAVE}}(\cdot)~~$
	& $\widehat{\beta}_2^{\text{PVE}}(\cdot)~~$
	& $\widehat{\beta}_2^{\text{PAVE}}(\cdot)~~$ \\
    \includegraphics[height=1.25in]{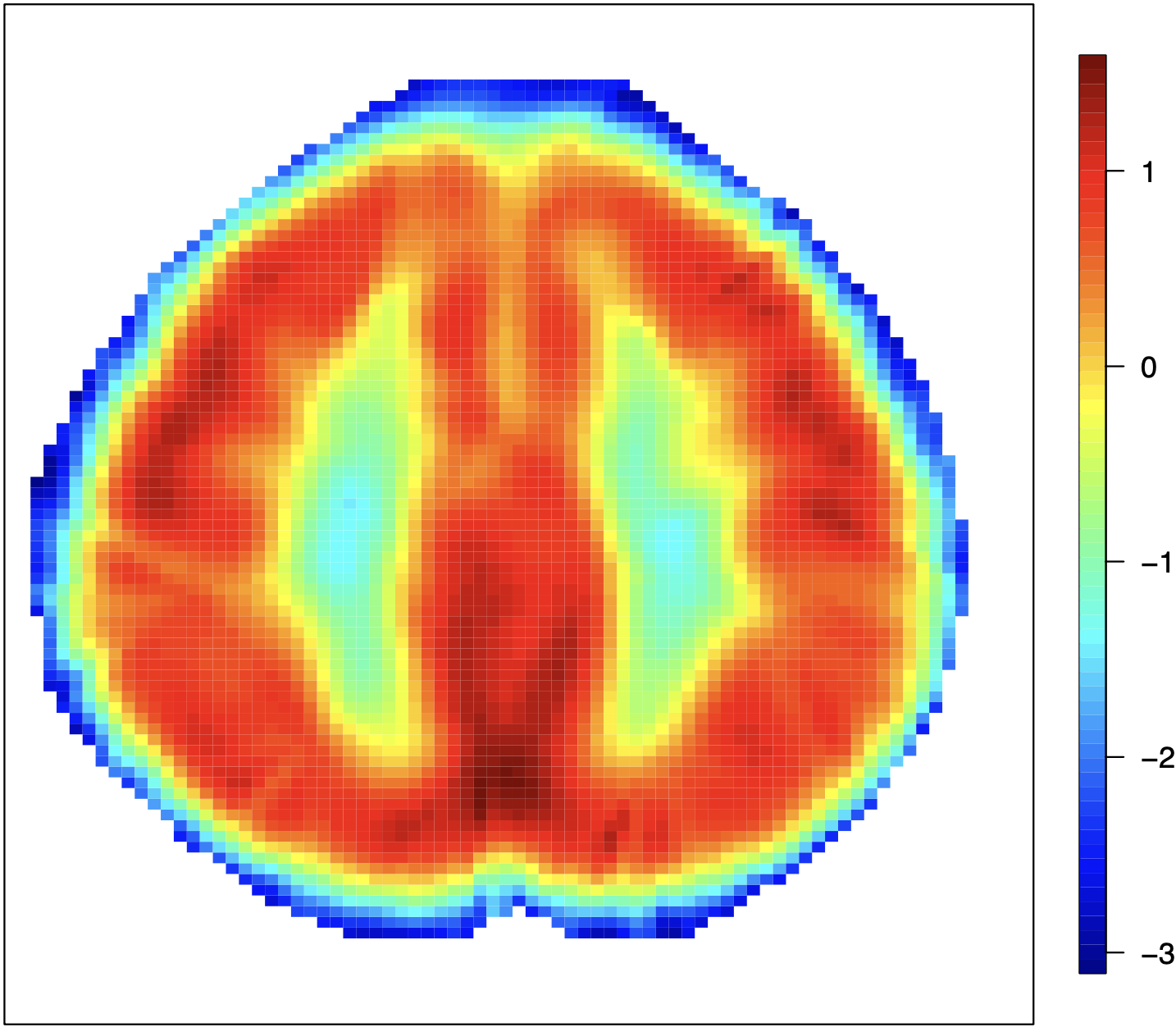}
	& \includegraphics[height=1.25in]{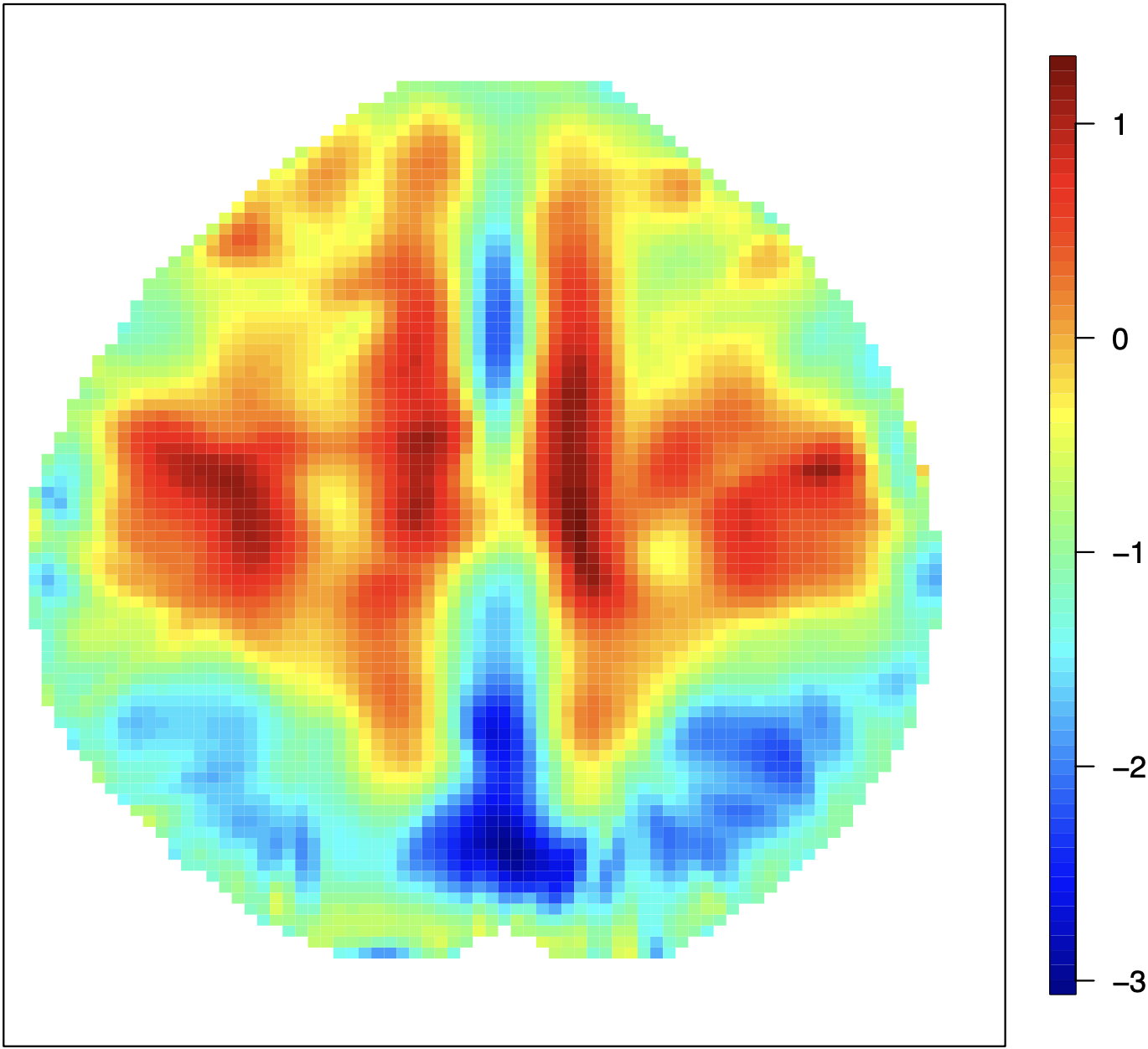} 
	& \includegraphics[height=1.25in]{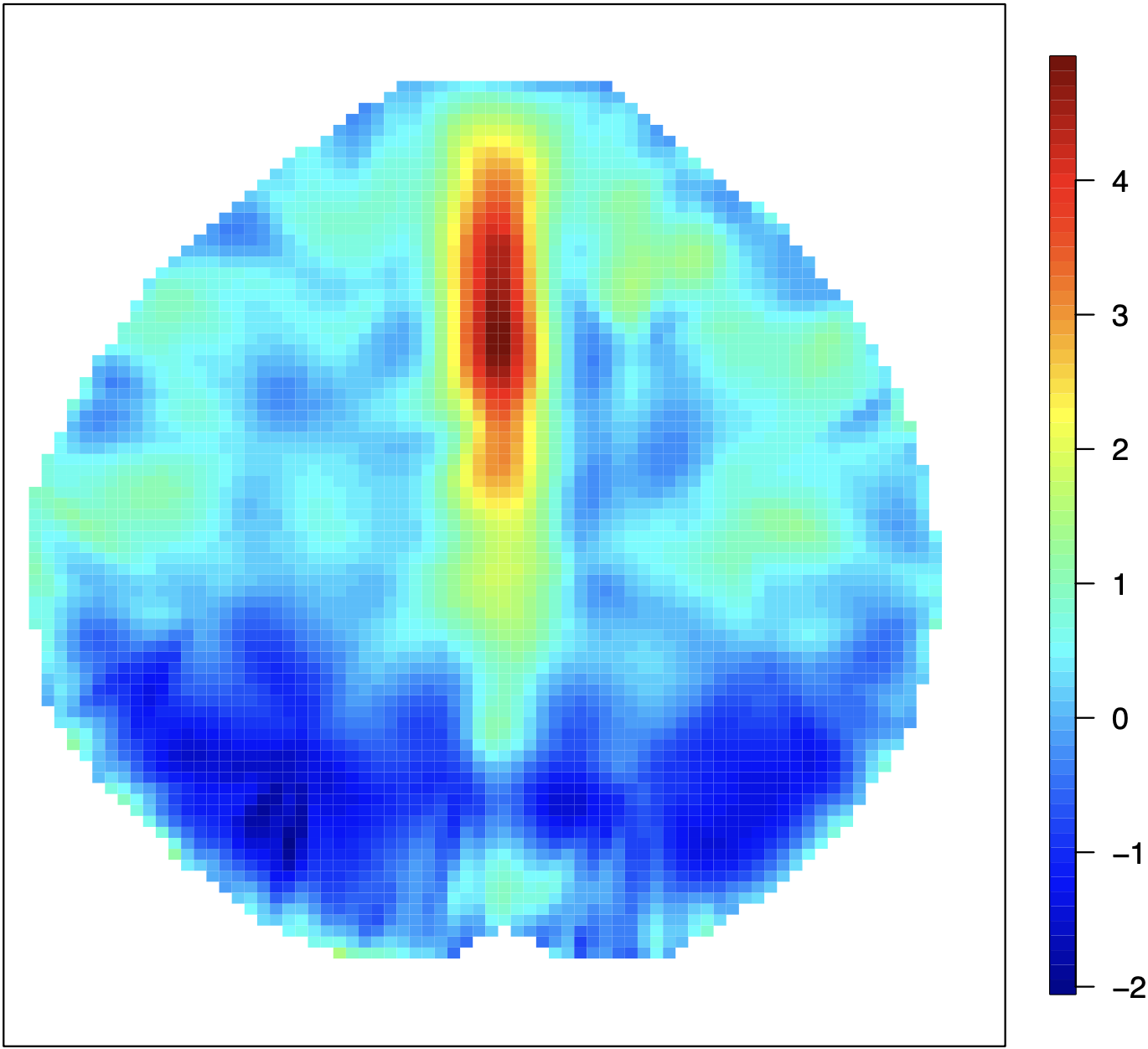} \vspace{-3pt} \\ 
	$\widehat{\phi}_{n1}(\cdot)$ for $Z(\cdot)~~~~~$ 
	& $\widehat{\phi}_{n2}(\cdot)$ for $Z(\cdot)~~~~~$ 
	& $\widehat{\phi}_{n3}(\cdot)$ for $Z(\cdot)~~~~~$ \\
    \includegraphics[height=1.25in]{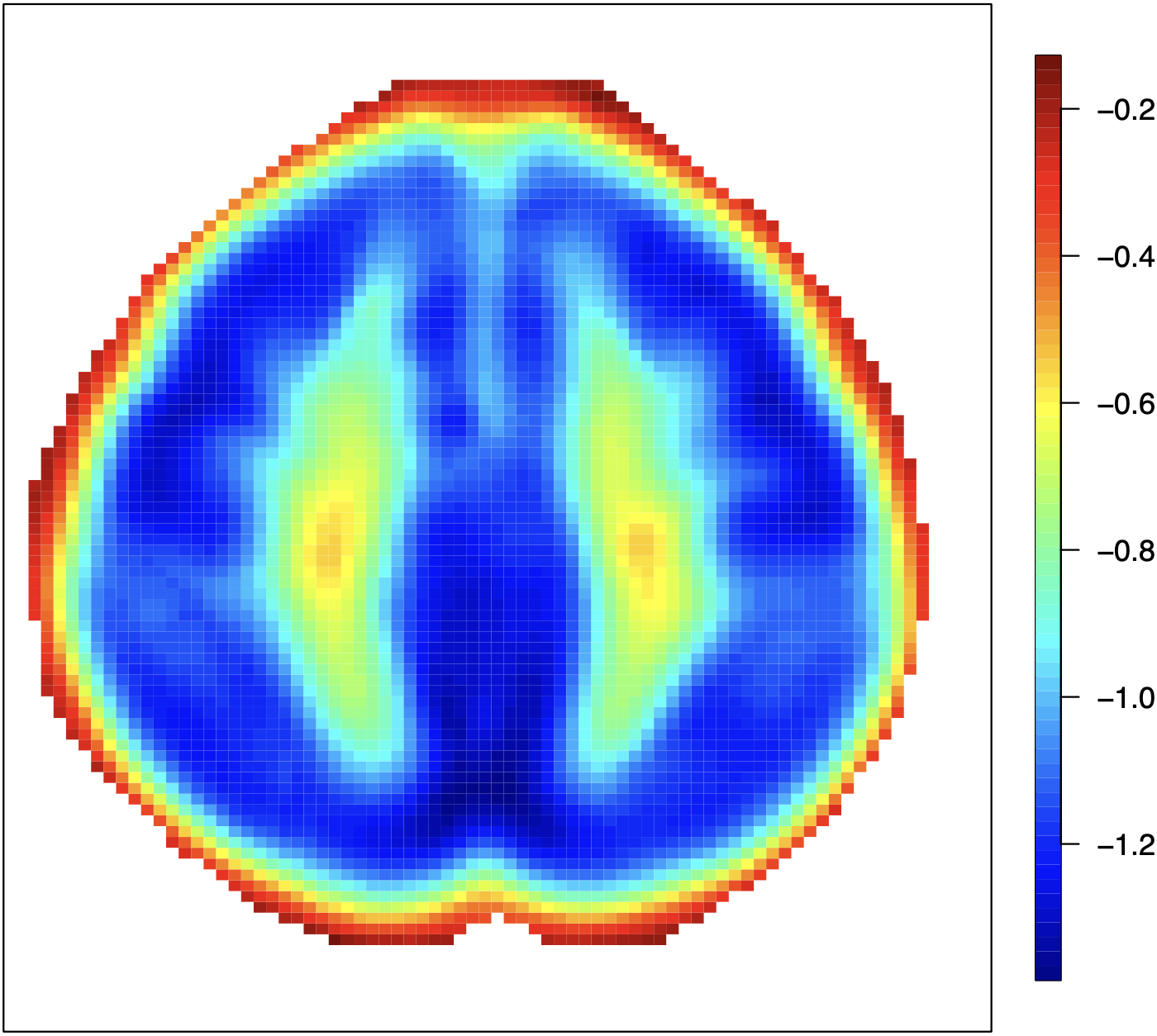}
	& \includegraphics[height=1.25in]{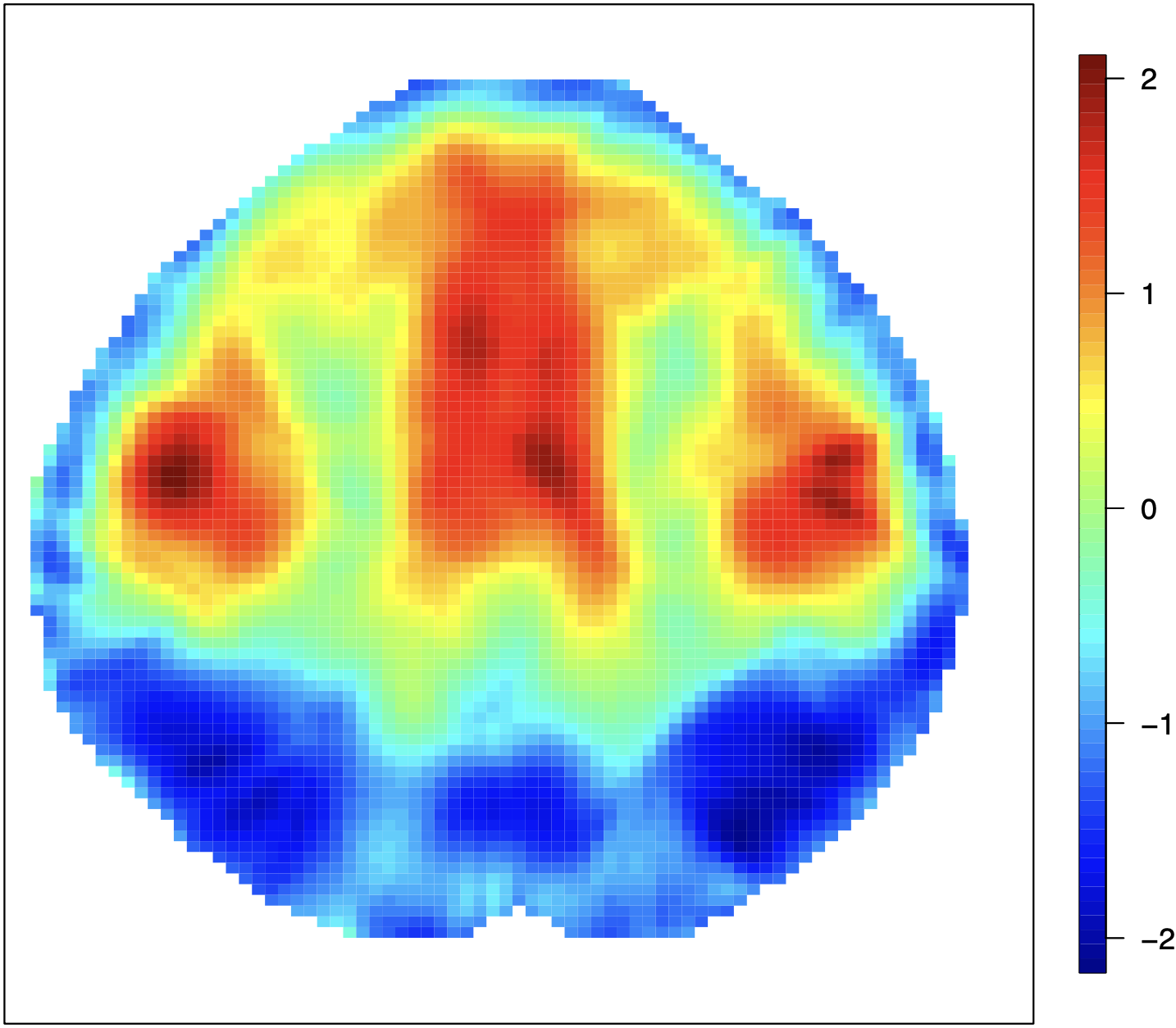} 
	& \includegraphics[height=1.25in]{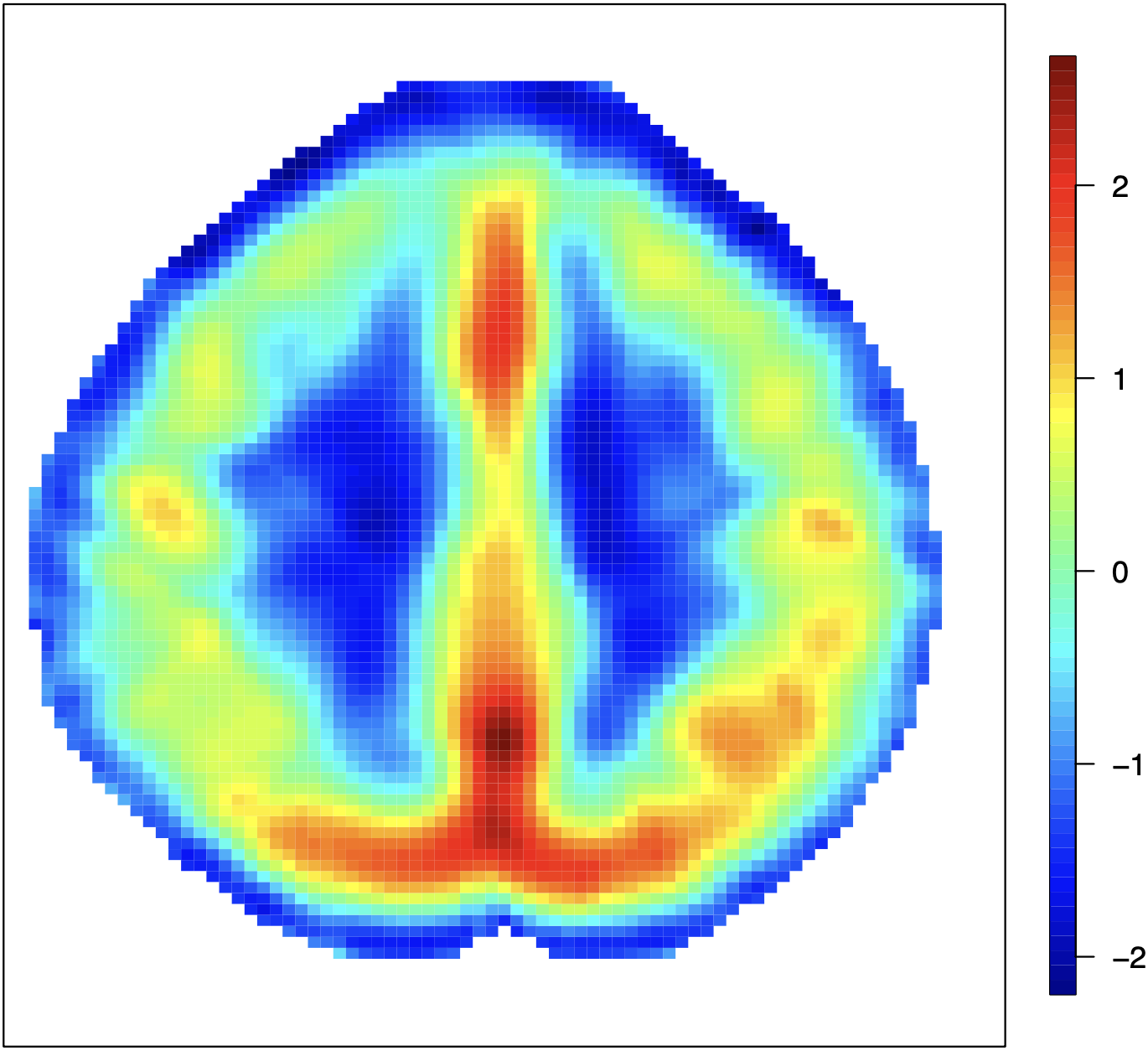} \vspace{-3pt} \\ 
	$\widehat{\phi}_{n1}(\cdot)$ for $AZ(\cdot)~~~~~$ 
	& $\widehat{\phi}_{n2}(\cdot)$ for $AZ(\cdot)~~~~~$ 
	& $\widehat{\phi}_{n3}(\cdot)$ for $AZ(\cdot)~~~~~$ \\
\end{tabular}
\end{center} \vspace*{-.6cm}
\caption{Top: Estimated coefficient maps for $\beta_1(\cdot)$ and $\beta_2(\cdot)$ using PVE and PAVE criteria, respectively. 
Middle: The three leading PC basis maps (top three ranked by both PVE and PAVE) for $Z(\cdot)$. 
Bottom: The three leading PC basis maps (top three ranked by both PVE and PAVE) for $AZ(\cdot)$. }
\label{FIG:beta2hat}
\end{figure}

Tables \ref{TAB:Est_BCI} presents the estimated coefficients for the nonfunctional predictors, along with the corresponding 95\% bootstrap confidence intervals. The main effect of the treatment can improve the performance in MMSE by around $4\sim6$ units on average. With the increase of age, the general MMSE scores will decrease; and with the increase of education level, the MMSE scores also increase in general: both are supported by studies on cognitive reserve in aging and AD  \citep{Fratiglioni:Winblad:vonStrauss:07,Stern:12}. As for the well-known risk genetic factor APOE gene, with more copies of epsilon 4 alleles in the APOE gene, the MMSE scores decrease, which means a higher risk for the onset of AD and agrees with other current studies \citep{Schneider:11}. The difference between females and males is not significant.

\begin{table}[htbp]
\renewcommand{\arraystretch}{0.95}{
\begin{center}
\caption{Estimated coefficients and 95\% bootstrap confidence intervals for the linear covariates using PVE and PAVE criteria, respectively.}
\label{TAB:Est_BCI}%
\begin{tabular}{lrcrc}
    \hline\hline
          & \multicolumn{2}{c}{PVE} & \multicolumn{2}{c}{PAVE} \\
	\cmidrule(lr){2-3}\cmidrule(lr){4-5}
    Term  & \multicolumn{1}{c}{Estimate} & \multicolumn{1}{c}{95\% Bootstrap CI} & \multicolumn{1}{c}{Estimate} & \multicolumn{1}{c}{95\% Bootstrap CI} \\
    \hline
    Intercept & 20.0650 & (8.655, 31.93) & 19.8176 & (8.663, 31.562) \\
    Age   & -0.0560 & (-0.113, -0.004) & -0.0568 & (-0.111, -0.005) \\
    Education & 0.2225 & (0.137, 0.340) & 0.2196 & (0.137, 0.334) \\
    Gender & -0.1998 & (-0.909, 0.707) & -0.1811 & (-0.896, 0.643) \\
    APOE1 & -0.3260 & (-1.079, 0.318) & -0.3201 & (-1.080, 0.272) \\
    APOE2 & -0.9185 & (-2.059, 0.321) & -0.9003 & (-2.022, 0.333) \\
    Ethinicity & 0.2363 & (-1.513, 2.580) & 0.2427 & (-1.536, 2.451) \\
    Race  & 1.2455 & (0.256, 2.664) & 1.3468 & (0.312, 2.527) \\
    Marriage & 0.1760 & (-0.711, 0.837) & 0.1451 & (-0.703, 0.801) \\
    Treatment & 4.1362 & (-6.893, 15.867) & 6.0930 & (-6.759, 12.850) \\
    Treatment$\times$Age & -0.0594 & (-0.107, -0.006) & -0.0643 & (-0.107, -0.014) \\
    Treatment$\times$Education & 0.0001 & (-0.104, 0.076) & 0.0056 & (-0.102, 0.072) \\
    Treatment$\times$Gender & -0.0302 & (-0.711, 0.704) & -0.0279 & (-0.620, 0.710) \\
    Treatment$\times$APOE1 & 0.4237 & (-0.312, 1.089) & 0.4050 & (-0.298, 1.036) \\
    Treatment$\times$APOE2 & -0.7102 & (-1.767, 0.576) & -0.8079 & (-1.764, 0.508) \\
    Treatment$\times$Ethinicity & -0.0115 & (-2.067, 2.216) & 0.0661 & (-1.957, 2.263) \\
    Treatment$\times$Race & 1.3543 & (0.368, 2.536) & 1.3859 & (0.381, 2.477) \\
    Treatment$\times$Marriage & -0.2662 & (-0.998, 0.552) & -0.3088 & (-0.963, 0.524) \\
    \hline\hline
\end{tabular}%
\end{center}}
\end{table}%

When one patient enters into the database with the imaging feature and other features, we can utilize the estimated coefficients to provide the optimal treatment regime based on Equation (\ref{EQN:pihat}). To be specific, for the imaging feature, for Subjects 27, 48, 55, and 160, the process can be visualized as depicted in the following Figure \ref{FIG:ADNI_Z}.
The estimated optimal ITRs for these four subjects are $A=1$, $A=0$, $A=0$, and $A=1$, respectively.
The differences in the images are subtle, note, for example, that the blue and orange parts are a little darker for those 
assigned $A=0$ versus those assigned $A=1$.

\begin{figure}[htbp]
\begin{center}
\begin{tabular}{lcccc}
	& $Z_i(\bs{s})$ & $\times~~~~\widehat{\beta}_2^{\text{PVE}}(\bs{s})$ & $\to Z_i(\bs{s})\widehat{\beta}_2^{\text{PVE}}(\bs{s})$ & {\small $\to \int_{\mathcal{V}} Z_i(\bs{s})\widehat{\beta}_2^{\text{PVE}}(\bs{s})\mathrm{d}\mu(\bs{s})$} \\
	\begin{tabular}{@{}c@{}}Subject \\ 27\end{tabular}& \includegraphics[align=c,height=0.95in]{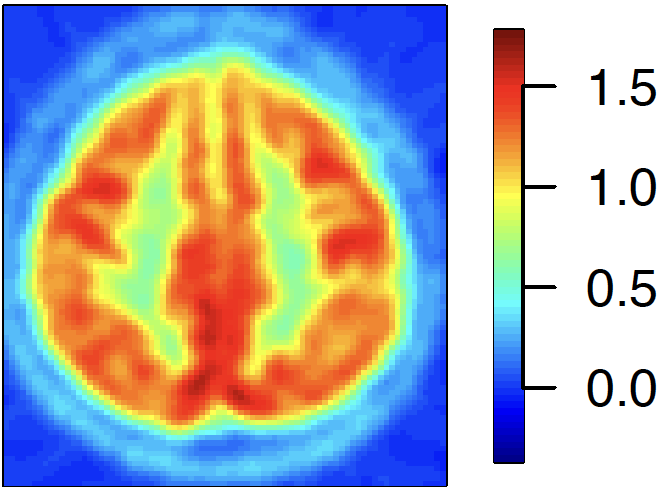}
		& \multirow{14}{*}{\includegraphics[align=c,height=1.01in]{figures/beta2hat_pve.png}}
		& \includegraphics[align=c,height=0.95in]{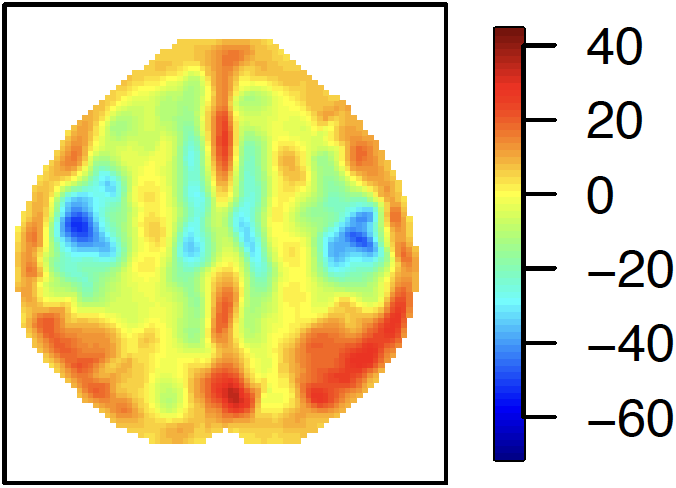}
		& -1.36 \\ 
	\begin{tabular}{@{}c@{}}Subject \\ 48\end{tabular} & \includegraphics[align=c,height=0.95in]{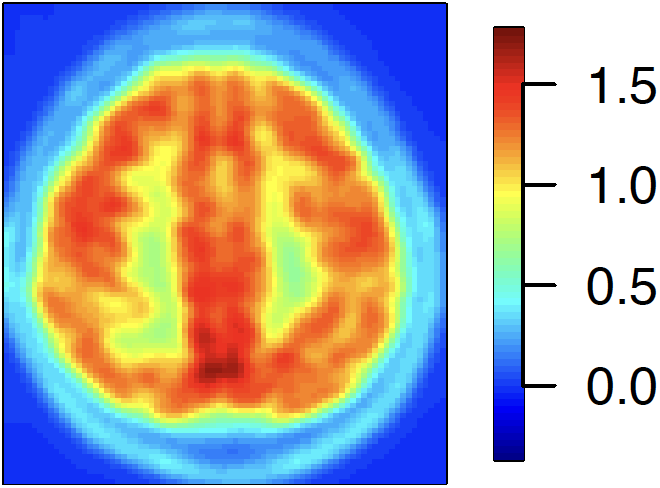}
		& & \includegraphics[align=c,height=0.95in]{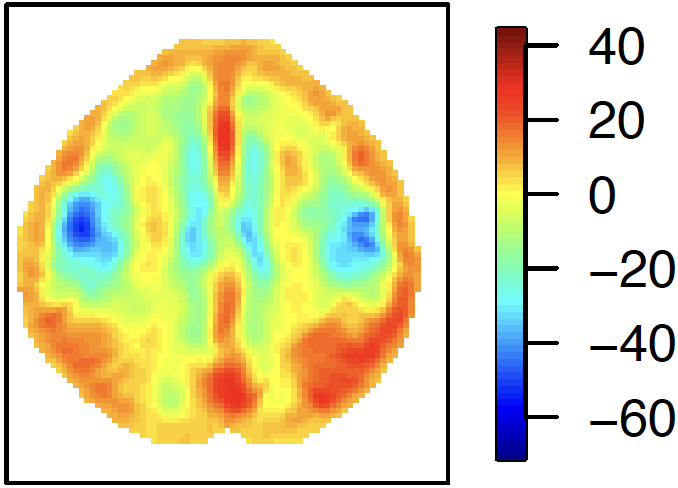}
		& -1.68 \\ 
	\begin{tabular}{@{}c@{}}Subject \\ 160\end{tabular} & \includegraphics[align=c,height=0.95in]{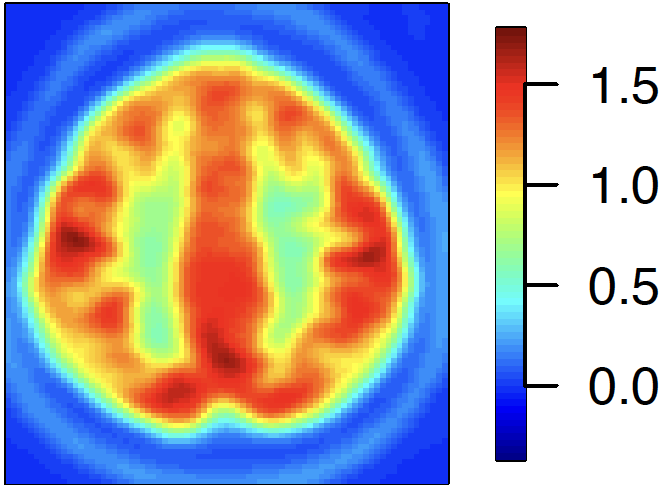}
		& & \includegraphics[align=c,height=0.95in]{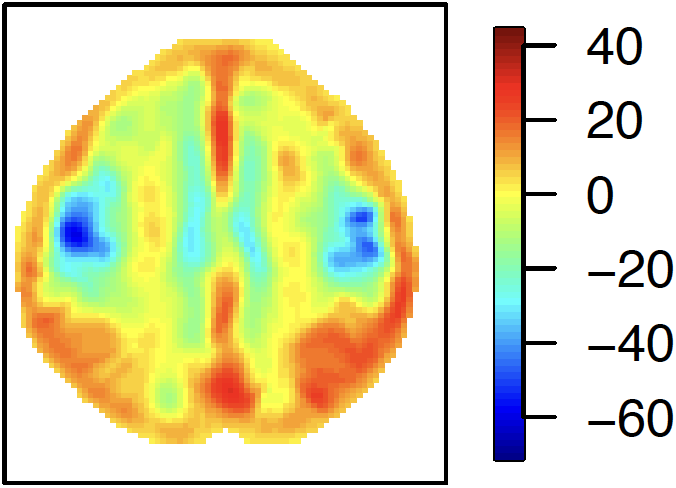}
		& -1.66 \\ 
	\begin{tabular}{@{}c@{}}Subject \\ 433\end{tabular} & \includegraphics[align=c,height=0.95in]{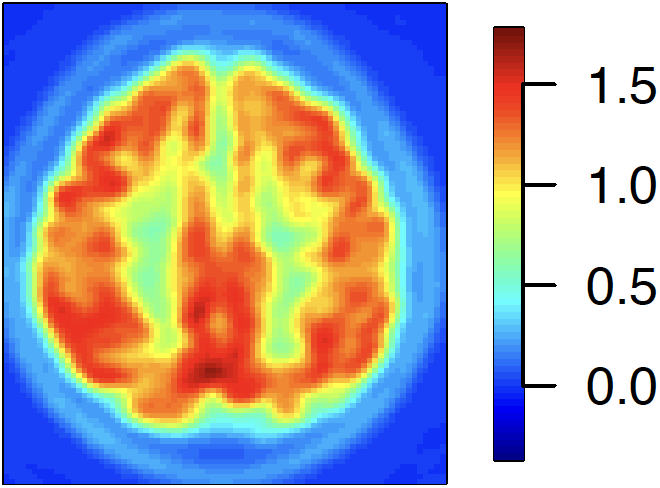}
		& & \includegraphics[align=c,height=0.95in]{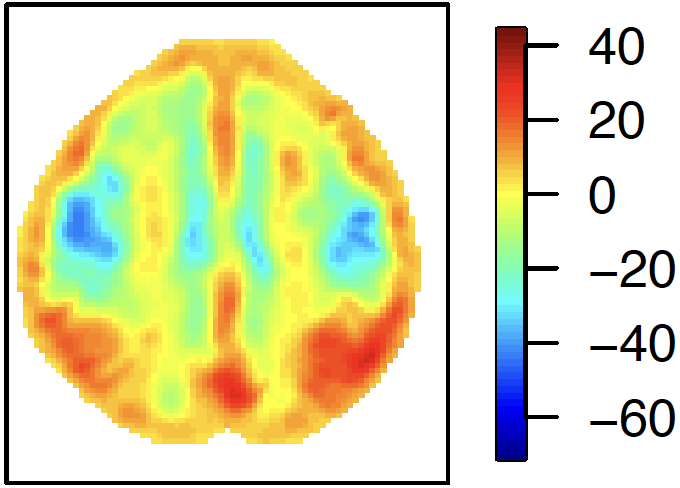}
		& -1.61 \\
\end{tabular}
\end{center} \vspace*{-.6cm}
\caption{An illustration of procedures for dealing with imaging features to obtain ITRs. The left panel represents the imaging features of four randomly selected subjects, Subjects $27$, $48$, $160$, and $433$, with MMSE scores of $30$, $27$, $15$, and $23$, respectively. 
The multiplication of the imaging feature $Z_i(\bs{s})$ and the coefficient map $\widehat{\beta}_2^{\text{PVE}}(\bs{s})$, as shown in the third column, represents the data-driven summary of the imaging features, which differs in the activity level in the blue region and orange region in the bottom.
The integral of the multiplication gives the image contribution of the interaction term; adding $\mathbf{X}\widehat{\alpha}$ and taking the positivity indicator will result in the optimal ITR. 
}
\label{FIG:ADNI_Z}
\end{figure}

\newpage
\vskip 0.1in  \noindent \textbf{8. Discussion} \vskip 0.1in
\label{SEC:discussion}

We proposed semiparametric functional learning with imaging features to estimate the optimal ITRs. 
Overall, the proposed approach can efficiently and precisely estimate the optimal treatment regime with abundant features. It can overcome the ``leakage'' problem in handling the complex domain of imaging data and is also computationally efficient. This approach provides an efficient and powerful tool to estimate optimal treatment regimes by incorporating abundant features into the precision medicine framework.

There are some natural extensions based on the proposed work. One extension is to extend the current regression model to a generalized regression model, in which the response is allowed to be discrete or categorical. For example, for the study of Alzheimer's disease, one could take the disease stage as the response. Another extension could be to generalize the current single-decision setting to the multi-stage decision setting, which is more common in chronic disease settings such as Alzheimer's disease.


\vskip 0.1in  \noindent \textbf{Acknowledgements} \vskip 0.1in

Research reported in this publication was supported by the National Institute Of General Medical Sciences of the National Institutes of Health under Award Number P20GM139769 (Xinyi Li), National Science Foundation awards DMS-2210658 (Xinyi Li) and DMS-2210659 (Michael Kosorok). 
The content is solely the responsibility of the authors and does not necessarily represent the official views of the National Institutes of Health.
The investigators within the ADNI contributed to the design and implementation of ADNI and/or provided data but did not participate in analysis or writing of this report. A complete listing of ADNI investigators can be found at \url{http://adni. loni.usc.edu/wp-content/uploads/how_to_apply/ADNI_Acknowledgement_List.pdf}.

\newpage

\vspace{0.8pc} \centerline{\large Supplemental Materials for \bf ``Functional Individualized Treatment Regimes with Imaging Features''}

\bigskip

\normalsize In this document, we provide the technical details of the conclusions presented in the main paper.
Specifically, we give detailed proofs of Theorems \ref{THM:fHconv} and \ref{THM:Conv} in the paper, and present a number of technical lemmas and other supporting results used in the proofs.

\fontsize{12}{14pt plus.8pt minus .6pt}\selectfont
\vskip 0.1in  \noindent \textbf{A. Proof of 2D-FPC Bases Properties} \vskip 0.1in 
\renewcommand{\thetheorem}{{\sc A.\arabic{theorem}}}
\renewcommand{\thelemma}{{\sc A.\arabic{lemma}}} %
\renewcommand{\thecorollary}{{\sc A.\arabic{corollary}}}
\renewcommand{\theequation}{A.\arabic{equation}} 
\renewcommand{\theproposition}{A.\arabic{proposition}} 
\renewcommand{\thefigure}{A.\arabic{figure}} 
\renewcommand{\thetable}{A.\arabic{table}} 
\renewcommand{\thedefinition}{A.\arabic{definition}} 
\renewcommand{\theremark}{A.\arabic{remark}} 
\renewcommand{\theexample}{A.\arabic{example}} 
\renewcommand{\thesubsection}{A.\arabic{subsection}}
\setcounter{equation}{0}
\setcounter{theorem}{0}
\setcounter{lemma}{0}
\setcounter{figure}{0}
\setcounter{remark}{0}
\setcounter{proposition}{0}
\setcounter{subsection}{0}
\label{SUBSEC:proof_2DFPCA}

We first develop the properties of the proposed 2D-FPC bases. We generally follow the sketch proof of Theorem \ref{THM:fHconv} in Section 4.1 in the paper, and provide a detailed proof for each conclusion, with corresponding supporting lemmas.

\begin{lemma}
\label{LEM:var0}
For any $b(\bs{s})\in\mathcal{H}_0^{\perp}$, $\mathrm{Var}\{\int_{\mathcal{V}}b(\bs{s})Z(\bs{s})\mathrm{d}\mu(\bs{s})\}=0$.
\end{lemma}
\begin{proof}
For any $b(\bs{s})\in\mathcal{H}_0^{\perp}$, 
\begin{align*}
	\mathrm{Var}\left\{\int_{\mathcal{V}}b(\bs{s})Z(\bs{s})\mathrm{d}\mu(\bs{s})\right\}
	&= \mathrm{E}\left[\int_{\mathcal{V}}b(\bs{s})\{Z(\bs{s})-\mathrm{E}Z(\bs{s})\}\mathrm{d}\mu(\bs{s})\right]^2 \\
	&=\mathrm{E}\left\{\sum_{k=1}^{\infty}\xi_k\int_{\mathcal{V}}b(\bs{s})\phi_k(\bs{s})\mathrm{d}\mu(\bs{s})\right\}^2
	=0.
\end{align*}
\end{proof}
\begin{remark}
From Lemma \ref{LEM:var0}, for any linear functional $\int_{\mathcal{V}}b(\bs{s})Z(\bs{s})\mathrm{d}\mu(\bs{s})$, we can assume without loss of generality that $b\in\mathcal{H}_0$. 
\end{remark}

In the following, we denote the theoretical and empirical expectation operators 
as $P$ and $\mathbb{P}_n$, respectively.
The following Theorem \ref{THM:GC} 
shows that both $\mathcal{F}_1$ and $\mathcal{F}_2$ are Glivenko-Cantelli classes, which can be further applied for consistency proofs.
\begin{theorem}
\label{THM:GC}
Suppose Assumption (A3) holds.
Then the following two conclusions hold:
\begin{enumerate}[(i)]
    \item $\mathcal{F}_1$ is a Glivenko-Cantelli class with envelope 
    \begin{align}
    \label{EQN:F1Z}
    	F_1(Z)=\left(\sum_{k=1}^{\infty}\left[\int_{\mathcal{V}}\phi_k(\bs{s})\left\{Z(\bs{s})-\mathrm{E}Z(\bs{s})\right\}\mathrm{d}\mu(\bs{s})\right]^2\right)^{1/2},
    \end{align}
    such that 
    $F_1\leq\|Z-\mathrm{E}Z\|_{\mu,2}$
    and 
    $\mathrm{E}F_1^2<\infty$.
    \item $\mathcal{F}_2$
is also a $P$-Glivenko-Cantelli class with envelope $F_1^2$, 
where $\mathrm{E}F_1^2<\infty$.
\end{enumerate}

\end{theorem}

\begin{proof} 
\begin{enumerate}[(i)]
    \item By Assumption (A3),
fix $\varepsilon>0$, $\exists \, K<\infty$, such that $\sum_{k=K+1}^{\infty}\lambda_k\leq\varepsilon^2$. 
Since $\mathcal{B}_1\subset\mathcal{H}_0$, $b\in\mathcal{B}_1$ has a unique Fourier representation $b(\bs{s})=\sum_{k=1}^{\infty}a_k\phi_k(\bs{s})$, where $\int_{\mathcal{V}}b^2(\bs{s})\mathrm{d}\mu(\bs{s})=\sum_{k=1}^{\infty}a_k^2\leq1$ by definition of $\mathcal{B}_1$ and projection of a basis. Then, letting $\ell_2=\{\bs{a}: \|\bs{a}\|_2\leq \infty\}$, we can write $\mathcal{F}_1$ as 
\[
	\mathcal{F}_1=\left\{\sum_{k=1}^{\infty}a_k\int_{\mathcal{V}}\phi_k(\bs{s})\{Z(\bs{s})-\mathrm{E}Z(\bs{s})\}\mathrm{d}\mu(\bs{s}): \bs{a}\in \ell_2 \text{ and } \|\bs{a}\|_2\leq 1\right\}.
\]
Let $\mathcal{A}_1=\{\bs{a}\in \ell_2 \text{ and } \|\bs{a}\|_2\leq 1\}$ and recall $\bar{Z}_n(\bs{s})=n^{-1}\sum_{i=1}^nZ_i(\bs{s})$. Thus, $\forall f\in\mathcal{F}_1$,
\begin{align*}
	&\sup_{f\in\mathcal{F}_1}\|\mathbb{P}_n f - P f\|
	=\sup_{a\in\mathcal{A}_1}\left|\sum_{k=1}^{\infty}a_k\int_{\mathcal{V}}\phi_k(\bs{s})\{\bar{Z}_n(\bs{s})-\mathrm{E}Z(\bs{s})\}\mathrm{d}\mu(\bs{s})\right| \\
	\leq & \sup_{a\in\mathcal{A}_1}\left|\sum_{k=K+1}^{\infty}a_k\int_{\mathcal{V}}\phi_k(\bs{s})\{\bar{Z}_n(\bs{s})-\mathrm{E}Z(\bs{s})\}\mathrm{d}\mu(\bs{s})\right|  + \sup_{a\in\mathcal{A}_1}\left|\sum_{k=1}^{K}a_k\int_{\mathcal{V}}\phi_k(\bs{s})\{\bar{Z}_n(\bs{s})-\mathrm{E}Z(\bs{s})\}\mathrm{d}\mu(\bs{s})\right| \\
	\equiv & A_n+B_n.
\end{align*}
By the Cauchy-Schwarz inequality,
\begin{align*}
	A_n
	=& \sup_{a\in\mathcal{A}_1}\left|n^{-1}\sum_{i=1}^n\sum_{k=K+1}^{\infty}a_k\int_{\mathcal{V}}\phi_k(\bs{s})\{Z_i(\bs{s})-\mathrm{E}Z(\bs{s})\}\mathrm{d}\mu(\bs{s})\right| \\
	\leq & \sup_{a\in\mathcal{A}_1}\left|n^{-1}\sum_{i=1}^n\left(\sum_{k=K+1}^{\infty}a_k^2\right)^{1/2}\left(\sum_{k=K+1}^{\infty}\left[\int_{\mathcal{V}}\phi_k(\bs{s})\{Z_i(\bs{s})-\mathrm{E}Z(\bs{s})\}\mathrm{d}\mu(\bs{s})\right]^2\right)^{1/2}\right| \\
	\leq & n^{-1}\sum_{i=1}^n\left(\sum_{k=K+1}^{\infty}\left[\int_{\mathcal{V}}\phi_k(\bs{s})\{Z_i(\bs{s})-\mathrm{E}Z(\bs{s})\}\mathrm{d}\mu(\bs{s})\right]^2\right)^{1/2}.
\end{align*}
Let $U_i=(\sum_{k=K+1}^{\infty}[\int_{\mathcal{V}}\phi_k(\bs{s})\{Z_i(\bs{s})-\mathrm{E}Z(\bs{s})\}\mathrm{d}\mu(\bs{s})]^2)^{1/2}$. 
Note that $(\mathrm{E}U_i)^2\leq\mathrm{E}U_i^2=\sum_{k=K+1}^{\infty} \lambda_k$.
Since this is an iid sum, we have that ${\limsup}_{n\to\infty} A_n=\limsup_{n\to\infty}n^{-1}\sum_{i=1}^nU_i \leq \varepsilon$ almost surely. 
Let 
$\mathcal{A}_2=\{\bs{a}\in \mathbb{R}^K\text{ and } \|\bs{a}\|_2\leq 1\}$. 
Then 
\[
	B_n=\sup_{\bs{a}\in\mathcal{A}_2} \left|\sum_{k=1}^K a_k\int_{\mathcal{V}}\phi_k(\bs{s})\{\bar{Z}_n(\bs{s})-\mathrm{E}Z(\bs{s})\}\mathrm{d}\mu(\bs{s}) \right|.
\]
Since $\mathcal{A}_2$ is a compact ball, $\exists$ finite subset $\mathcal{M}_2\subset\mathcal{A}_2$, such that $\sup_{\bs{a}\in\mathcal{A}_2}\inf_{\widetilde{\bs{a}}\in\mathcal{M}_2} \|\bs{a}-\widetilde{\bs{a}}\|\leq \varepsilon$.
Accordingly,
\begin{align*}
	B_n 
	&\leq \max_{\widetilde{\bs{a}}\in\mathcal{M}_2} \left|\sum_{k=1}^K \widetilde{a}_k\int_{\mathcal{V}}\phi_k(\bs{s})\{\bar{Z}_n(\bs{s})-\mathrm{E}Z(\bs{s})\}\mathrm{d}\mu(\bs{s})\right| \\
	&~~~~~~~~~~~~~~~~~~~~~~~~+ \sup_{\bs{a}\in\mathcal{A}_2,\widetilde{\bs{a}}\in\mathcal{M}_2: \|\bs{a}-\widetilde{\bs{a}}\|\leq \varepsilon} \left|\sum_{k=1}^K (a_k-\widetilde{a}_k)\int_{\mathcal{V}}\phi_k(\bs{s})\{\bar{Z}_n(\bs{s})-\mathrm{E}Z(\bs{s})\}\mathrm{d}\mu(\bs{s})\right| \\
	&\equiv B_{1n} + B_{2n}.
\end{align*}
By the standard strong law of large numbers, 
$B_{1n}\xrightarrow{\as}0$ as $n\to\infty$. As for $B_{2n}$, by the Cauchy-Schwarz inequality,
\begin{eqnarray*}
	&&\sup_{\bs{a}\in\mathcal{A}_2,\widetilde{\bs{a}}\in\mathcal{M}_2: \|\bs{a}-\widetilde{\bs{a}}\|\leq \varepsilon} \left|\sum_{k=1}^K (a_k-\widetilde{a}_k)\int_{\mathcal{V}}\phi_k(\bs{s})\{\bar{Z}_n(\bs{s})-\mathrm{E}Z(\bs{s})\}\mathrm{d}\mu(\bs{s})\right| \\
	&=&\sup_{\bs{a}\in\mathcal{A}_2,\widetilde{\bs{a}}\in\mathcal{M}_2: \|\bs{a}-\widetilde{\bs{a}}\|\leq \varepsilon} \left|n^{-1}\sum_{i=1}^n\sum_{k=1}^K (a_k-\widetilde{a}_k)\int_{\mathcal{V}}\phi_k(\bs{s})\{Z_i(\bs{s})-\mathrm{E}Z(\bs{s})\}\mathrm{d}\mu(\bs{s})\right| \\
	&\leq& \varepsilon\left|n^{-1}\sum_{i=1}^n\left(\sum_{k=1}^K\left[\int_{\mathcal{V}}\phi_k(\bs{s})\{Z_i(\bs{s})-\mathrm{E}Z(\bs{s})\}\mathrm{d}\mu(\bs{s})\right]^2\right)^{1/2}\right| \\
	&\leq & \varepsilon\left|n^{-1}\sum_{i=1}^n\sum_{k=1}^{\infty}\left[\int_{\mathcal{V}}\phi_k(\bs{s})\{Z_i(\bs{s})-\mathrm{E}Z(\bs{s})\}\mathrm{d}\mu(\bs{s})\right]^2\right|^{1/2}.
\end{eqnarray*}
Let $V_i=\sum_{k=1}^{\infty}\left[\int_{\mathcal{V}}\phi_k(\bs{s})\{Z_i(\bs{s})-\mathrm{E}Z(\bs{s})\}\mathrm{d}\mu(\bs{s})\right]^2$. Note that $V_i$ does not depend on $\varepsilon$, and 
\[
	\mathrm{E}V_i=\sum_{k=1}^{\infty}\lambda_k<\infty
	\Rightarrow \limsup_{n\to\infty}B_{2n}\leq\varepsilon C_0~~ \as,
\]
where $C_0=|\sum_{k=1}^{\infty}\lambda_k|^{1/2}$ does not depend on $\varepsilon$. Therefore,
$\sup_{f\in\mathcal{F}_1}\|\mathbb{P}_n f - Pf\|\xrightarrow{\as}0$, thus, $\mathcal{F}_1$  
is a $P$-Glivenko-Cantelli class by definition. Also note that $\mathcal{F}_1$ is enveloped by 
\begin{align*}
	F_1(Z)&=
	\left(\sum_{k=1}^{\infty}\left[\int_{\mathcal{V}}\phi_k(\bs{s})\{Z(\bs{s})-\mathrm{E}Z(\bs{s})\}\mathrm{d}\mu(\bs{s})\right]^2\right)^{1/2} \\
	&\leq \left(\sum_{k=1}^{\infty}\left\{\int_{\mathcal{V}}\phi_k^2(\bs{s})\mathrm{d}\mu(\bs{s})\right\}\left[\int_{\mathcal{V}}\{Z(\bs{s})-\mathrm{E}Z(\bs{s})\}^2\mathrm{d}\mu(\bs{s})\right]\right)^{1/2}
	=\|Z-\mathrm{E}Z\|_{\mu,2},
\end{align*}
and
\[
	\mathrm{E}F_1^2(Z) \leq 
	\mathrm{E}\left(\sum_{k=1}^{\infty}\left[\int_{\mathcal{V}}\phi_k(\bs{s})\{Z_i(\bs{s})-\mathrm{E}Z(\bs{s})\}\mathrm{d}\mu(\bs{s})\right]^2\right)
	=\sum_{k=1}^{\infty}\lambda_k
	<\infty.
\]
    \item By (i) and Corollary 9.27 of \cite{Kosorok:08}, we can conclude that $\mathcal{F}_2$ is also a $P$-Glivenko-Cantelli class, with envelope $F_1^2$.
\end{enumerate}
\end{proof}


Given the conclusions in Theorem \ref{THM:GC}, we can therefore develop the convergence results for the bases $\{\widehat{\phi}_{nk}\}_k$ through the following Theorems \ref{THM:fn1hat} -- \ref{THM:fnk+1hat}.

\begin{theorem}
\label{THM:fn1hat}
Let 
$
	\widehat{\phi}_{n1}
	=\arg\max_{f\in\mathcal{B}_1\cap \mathcal{S}}\int_{\mathcal{V}\times\mathcal{V}}f(\bs{s})f(\bs{s}^{\prime})V_n(\bs{s},\bs{s}^{\prime})\mathrm{d}\mu(\bs{s})\mathrm{d}\mu(\bs{s}^{\prime})
$. 
Suppose Assumption (A3) holds. Then $\exists$ a sign sequence $\{S_{n1}: \forall n, S_{n1}\in\{-1,1\}\}$ such that
\[
	\|S_{n1}\widehat{\phi}_{n1}-\phi_1\|_{\mu,2} \xrightarrow{\as}0, ~~~~\text{ as } n\to\infty. 
\]
\end{theorem}

\begin{proof} [Proof of Theorem \ref{THM:fn1hat}]
Let 
\begin{align}
\label{EQN:B1n}
	B_{1n}=\sup_{f\in\mathcal{B}_1}\left|\int_{\mathcal{V}\times\mathcal{V}}f(\bs{s})f(\bs{s}^{\prime})\{V_n(\bs{s},\bs{s}^{\prime})-V_0(\bs{s},\bs{s}^{\prime})\}\mathrm{d}\mu(\bs{s})\mathrm{d}\mu(\bs{s}^{\prime})\right|.
\end{align} 
By Theorem \ref{THM:GC} (ii), $B_{1n}\xrightarrow{\as}0$ as $n\to\infty$. Then
\begin{align*}
	&\int_{\mathcal{V}\times\mathcal{V}}\widehat{\phi}_{n1}(\bs{s})\widehat{\phi}_{n1}(\bs{s}^{\prime})V_n(\bs{s},\bs{s}^{\prime})\mathrm{d}\mu(\bs{s})\mathrm{d}\mu(\bs{s}^{\prime}) 
	\leq \int_{\mathcal{V}\times\mathcal{V}}\widehat{\phi}_{n1}(\bs{s})\widehat{\phi}_{n1}(\bs{s}^{\prime})V_0(\bs{s},\bs{s}^{\prime})\mathrm{d}\mu(\bs{s})\mathrm{d}\mu(\bs{s}^{\prime}) + B_{1n}\\
	& ~~~~~~~~~~~~~~~~~~~~~~~~~~~~~~~~ \leq \sup_{f\in\mathcal{B}_1}\int_{\mathcal{V}\times\mathcal{V}}f(\bs{s})f(\bs{s}^{\prime})V_0(\bs{s},\bs{s}^{\prime})\mathrm{d}\mu(\bs{s})\mathrm{d}\mu(\bs{s}^{\prime}) + B_{1n}
	\xrightarrow{\as}\lambda_1, ~~~~\text{ as } n\to\infty.
\end{align*}
Let $\widetilde{\phi}_{n1}(\bs{s})$ be the projection of $\phi_1(\bs{s})$ onto $\mathcal{B}_1\cap \mathcal{S}$. Then on the other side, as $n\to\infty$,
\begin{align*}
	&\int_{\mathcal{V}\times\mathcal{V}}\widehat{\phi}_{n1}(\bs{s})\widehat{\phi}_{n1}(\bs{s}^{\prime})V_n(\bs{s},\bs{s}^{\prime})\mathrm{d}\mu(\bs{s})\mathrm{d}\mu(\bs{s}^{\prime}) 
	\geq \int_{\mathcal{V}\times\mathcal{V}}\widetilde{\phi}_{n1}(\bs{s})\widetilde{\phi}_{n1}(\bs{s}^{\prime})V_n(\bs{s},\bs{s}^{\prime})\mathrm{d}\mu(\bs{s})\mathrm{d}\mu(\bs{s}^{\prime}) \\
	& \geq \int_{\mathcal{V}\times\mathcal{V}}\widetilde{\phi}_{n1}(\bs{s})\widetilde{\phi}_{n1}(\bs{s}^{\prime})V_0(\bs{s},\bs{s}^{\prime})\mathrm{d}\mu(\bs{s})\mathrm{d}\mu(\bs{s}^{\prime}) - B_{1n} 
	\xrightarrow{\as} \int_{\mathcal{V}\times\mathcal{V}}\phi_1(\bs{s})\phi_1(\bs{s}^{\prime})V_0(\bs{s},\bs{s}^{\prime})\mathrm{d}\mu(\bs{s})\mathrm{d}\mu(\bs{s}^{\prime}) 
	= \lambda_1.
\end{align*}
Therefore, if we define $\widehat{a}_{nk}=\int_{\mathcal{V}}\widehat{\phi}_{n1}(\bs{s})\phi_k(\bs{s})\mathrm{d}\mu(\bs{s})$, we can conclude that as $n\to\infty$,
\begin{align*}
	\int_{\mathcal{V}\times\mathcal{V}}\widehat{\phi}_{n1}(\bs{s})\widehat{\phi}_{n1}(\bs{s}^{\prime})V_0(\bs{s},\bs{s}^{\prime})\mathrm{d}\mu(\bs{s})\mathrm{d}\mu(\bs{s}^{\prime})
	&=\int_{\mathcal{V}\times\mathcal{V}} \sum_{k=1}^{\infty}\widehat{\phi}_{n1}(\bs{s})\widehat{\phi}_{n1}(\bs{s}^{\prime})\lambda_k\phi_k(\bs{s})\phi_k(\bs{s}^{\prime})\mathrm{d}\mu(\bs{s})\mathrm{d}\mu(\bs{s}^{\prime}) \\
	&=\sum_{k=1}^{\infty}\widehat{a}_{nk}^2\lambda_k
	\xrightarrow{\as}\lambda_1.
\end{align*}
Note that $\sum_{k=1}^{\infty}\widehat{a}_{nk}^2\leq 1$ since $\widehat{\phi}_{n1}\in\mathcal{B}_1$, and since $\infty>\lambda_1>\lambda_2>\ldots$: this now leads to $\widehat{a}_{n1}^2\xrightarrow{\as}1$ and $\sum_{j=2}^{\infty}\widehat{a}_{nk}^2\xrightarrow{\as}0$. 
Hence there exists a sign sequence $\{S_{n1}\in\{-1,1\}\}_n$ such that
$
	\|S_{n1}\widehat{\phi}_{n1}-\phi_1\|_{\mu,2} \xrightarrow{\as}0
$ as $n\to\infty$.
\end{proof}

\begin{theorem}
\label{THM:fnkhat}
Suppose for some $1\leq K<\infty$, $\exists$ sign sequences $\{S_{n1}, S_{n2}, \ldots, S_{nK}: S_{nk}\in\{-1,1\}, k=1,\ldots,K\}_n$ such that 
\[
	\max_{1\leq k\leq K} \|S_{nk}\widehat{\phi}_{nk}-\phi_k\|_{\mu,2}  \xrightarrow{\as}0, ~~~~\text{ as } n\to\infty,
\]
where $\widehat{\phi}_{nk}\in\mathcal{B}_1\cap \mathcal{S}$, $\forall 1\leq k\leq K$ and $\{\widehat{\phi}_{n1},\ldots,\widehat{\phi}_{nK}\}$ are orthogonal in 
$\mathcal{H}_0$.
Recall the definitions of $\mathcal{H}_{0K}$ and $\widehat{\mathcal{H}}_{nK}$ in (\ref{DEF:HK}),
and let $\widehat{\mathcal{H}}_{nK}^{\perp}$ and $\mathcal{H}_{0K}^{\perp}$ denote the respective closed orthocomplements in $\mathcal{H}_0$.
Let $\{g_n\}\in\mathcal{H}_0$ be a sequence satisfying $\limsup_{n\to\infty}\|g_n\|_{\mu,2}<\infty$. Suppose Assumption (A3) holds, then both 
\[
	\|\widehat{\mathcal{H}}_{nK}g_n - \mathcal{H}_{0K}g_n\|_{\mu,2}\xrightarrow{\as}0 ~~~~\text{ and }~~~~
	\|\widehat{\mathcal{H}}_{nK}^{\perp}g_n - \mathcal{H}_{0K}^{\perp}g_n\|_{\mu,2}\xrightarrow{\as}0, ~~~~\text{ as } n\to\infty.
\]
\end{theorem}

\begin{proof} [Proof of Theorem \ref{THM:fnkhat}]
\begin{align*}
	\widehat{\mathcal{H}}_{nK}g_n(\bs{s})
	=&\sum_{k=1}^K \widehat{\phi}_{nk}(\bs{s}) \int_{\mathcal{V}}\widehat{\phi}_{nk}(\bs{s}^{\prime})g_n(\bs{s}^{\prime})\mathrm{d}\mu(\bs{s}^{\prime})
	=\sum_{k=1}^K S_{nk}\widehat{\phi}_{nk}(\bs{s}) \int_{\mathcal{V}}S_{nk}\widehat{\phi}_{nk}(\bs{s}^{\prime})g_n(\bs{s}^{\prime})\mathrm{d}\mu(\bs{s}^{\prime}) \\
	=&H_{0K}g_n(\bs{s}) 
		+ \sum_{k=1}^K \{S_{nk}\widehat{\phi}_{nk}(\bs{s})-\phi_k(\bs{s})\} \int_{\mathcal{V}}\phi_k(\bs{s}^{\prime})g_n(\bs{s}^{\prime})\mathrm{d}\mu(\bs{s}^{\prime}) \\
	&\quad\quad\quad\quad\,+ \sum_{k=1}^K S_{nk}\widehat{\phi}_{nk}(\bs{s}) \int_{\mathcal{V}}\{S_{nk}\widehat{\phi}_{nk}(\bs{s}^{\prime})-\phi_k(\bs{s}^{\prime})\}g_n(\bs{s}^{\prime})\mathrm{d}\mu(\bs{s}^{\prime}) \\
	=&H_{0K}g_n(\bs{s}) + F_{1n}(\bs{s}) + F_{2n}(\bs{s}) .
\end{align*}
Note that as $n\to\infty$,
\[
	\|F_{1n}\|_{\mu,2}
	\leq\sum_{k=1}^K \|S_{nk}\widehat{\phi}_{nk}-\phi_k\|_{\mu,2}\|g_n\|_{\mu,2} \xrightarrow{\as} 0, ~~
	\|F_{2n}\|_{\mu,2}
	\leq\sum_{k=1}^K \|S_{nk}\widehat{\phi}_{nk}-\phi_k\|_{\mu,2}\|g_n\|_{\mu,2} \xrightarrow{\as} 0.
\]
Thus, $\|\widehat{\mathcal{H}}_{nK}g_n - \mathcal{H}_{0K}g_n\|_{\mu,2}\xrightarrow{\as}0$, as $n\to\infty$.
Since 
\[
	\|\widehat{\mathcal{H}}_{nK}^{\perp}g_n - \mathcal{H}_{0K}^{\perp}g_n\|_{\mu,2}
	=\|(\mathbb{I}-\widehat{\mathcal{H}}_{nK})g_n - (\mathbb{I}-\mathcal{H}_{0K})g_n\|_{\mu,2}
	=\|\widehat{\mathcal{H}}_{nK}g_n - \mathcal{H}_{0K}g_n\|_{\mu,2},
\]
where $\mathbb{I}$ is the identity operator, we now have the second conclusion of the theorem.
\end{proof}

\begin{theorem}
\label{THM:fnk+1hat}
Suppose Assumption (A3) holds.
Assume for some $1\leq K<\infty$, that $\{\widehat{\phi}_{n1},\ldots,\widehat{\phi}_{nK}\}$ form an orthonormal system on $\mathcal{S}\cap\mathcal{H}_0$ with
\[
	\max_{1\leq k\leq K} \|S_{nk}\widehat{\phi}_{nk}-\phi_k\|_{\mu,2} \xrightarrow{\as}0, ~~~~\text{ as } n\to\infty,
\]
for some sign sequences $\{S_{n1},\ldots,S_{nK}:S_{nk}\in\{-1,1\}, k=1,\ldots,K\}_n$. Let
\[
	\widehat{\phi}_{n(K+1)} = \argmax_{f\in\mathcal{B}_1\cap \mathcal{S}\cap\widehat{\mathcal{H}}_{nK}^{\perp}} 
	\int_{\mathcal{V}\times\mathcal{V}}f(\bs{s})f(\bs{s}^{\prime})V_n(\bs{s},\bs{s}^{\prime})\mathrm{d}\mu(\bs{s})\mathrm{d}\mu(\bs{s}^{\prime}).
\]
Then $\exists$ a sign sequence $\{S_{n(K+1)}\in\{-1,1\}\}_n$ such that 
\[
	\|S_{n(K+1)}\widehat{\phi}_{n(K+1)}-\phi_{(K+1)}\|_{\mu,2} \xrightarrow{\as}0, ~~~~\text{ as } n\to\infty.
\]
\end{theorem}
\begin{proof}[Proof of Theorem \ref{THM:fnk+1hat}]
Let $V_{0K}(\bs{s},\bs{s}^{\prime})=\sum_{k=K+1}^{\infty}\lambda_k\phi_k(\bs{s})\phi_k(\bs{s}^{\prime})$. Recall the definition of $B_{1n}$ given in (\ref{EQN:B1n}) and $B_{1n}\xrightarrow{\as}0$ as shown in Theorem \ref{THM:fn1hat}. Again by Theorem \ref{THM:GC} (ii), we have as $n\to\infty$,
\begin{align}
	&\int_{\mathcal{V}\times\mathcal{V}}\widehat{\phi}_{n(K+1)}(\bs{s})\widehat{\phi}_{n(K+1)}(\bs{s}^{\prime})V_n(\bs{s},\bs{s}^{\prime})\mathrm{d}\mu(\bs{s})\mathrm{d}\mu(\bs{s}^{\prime}) \label{EQN:ffVK}\\
	&~~~~~~~~~~~~~~~~~~~~~~~~~~~~~~~~\leq \int_{\mathcal{V}\times\mathcal{V}}\widehat{\phi}_{n(K+1)}(\bs{s})\widehat{\phi}_{n(K+1)}(\bs{s}^{\prime})V_0(\bs{s},\bs{s}^{\prime})\mathrm{d}\mu(\bs{s})\mathrm{d}\mu(\bs{s}^{\prime}) + B_{1n} \nonumber\\
	&~~~~~~~~~~~~~~~~~~~~~~~~~~~~~~~~\leq \int_{\mathcal{V}\times\mathcal{V}}\widehat{\phi}_{n(K+1)}(\bs{s})\widehat{\phi}_{n(K+1)}(\bs{s}^{\prime})V_{0K}(\bs{s},\bs{s}^{\prime})\mathrm{d}\mu(\bs{s})\mathrm{d}\mu(\bs{s}^{\prime}) + B_{1n} + B_{2n}, \nonumber\\
	&~~~~~~~~~~~~~~~~~~~~~~~~~~~~~~~~\leq \sup_{f\in\mathcal{B}_1}\int_{\mathcal{V}\times\mathcal{V}}f(\bs{s})f(\bs{s}^{\prime})V_{0K}(\bs{s},\bs{s}^{\prime})\mathrm{d}\mu(\bs{s})\mathrm{d}\mu(\bs{s}^{\prime}) + B_{1n} + B_{2n} \xrightarrow{\as}\lambda_{K+1}, \nonumber
\end{align}
where 
$
	B_{2n} = |\sum_{k=1}^K\lambda_k\{\int_{\mathcal{V}}\widehat{\phi}_{n(K+1)}(\bs{s})\phi_k(\bs{s})\mathrm{d}\mu(\bs{s})\}^2 |
	\xrightarrow{\as} 0
$ 
by Theorem \ref{THM:fnkhat}.
Let $\widetilde{\phi}_{n(K+1)}=\mathcal{B}_1\mathcal{S}\phi_{K+1}$, then
\begin{align*}
	(\ref{EQN:ffVK}) 
	&\geq \int_{\mathcal{V}\times\mathcal{V}}\widetilde{\phi}_{n(K+1)}(\bs{s})\widetilde{\phi}_{n(K+1)}(\bs{s}^{\prime})V_n(\bs{s},\bs{s}^{\prime})\mathrm{d}\mu(\bs{s})\mathrm{d}\mu(\bs{s}^{\prime}) \\
	&\geq \int_{\mathcal{V}\times\mathcal{V}}\widetilde{\phi}_{n(K+1)}(\bs{s})\widetilde{\phi}_{n(K+1)}(\bs{s}^{\prime})V_0(\bs{s},\bs{s}^{\prime})\mathrm{d}\mu(\bs{s})\mathrm{d}\mu(\bs{s}^{\prime}) - B_{1n} 
	\xrightarrow{\as}\lambda_{K+1}, ~~~~\text{ as } n\to\infty.
\end{align*}
Thus, 
\begin{align}
\label{EQN:ffV0K}
	\int_{\mathcal{V}\times\mathcal{V}}\widehat{\phi}_{n(K+1)}(\bs{s})\widehat{\phi}_{n(K+1)}(\bs{s}^{\prime})V_{0K}(\bs{s},\bs{s}^{\prime})\mathrm{d}\mu(\bs{s})\mathrm{d}\mu(\bs{s}^{\prime})
	\xrightarrow{\as}\lambda_{K+1}, ~~~~\text{ as } n\to\infty. 
\end{align}
Let $\widehat{a}_{nk}=\int_{\mathcal{V}}\widehat{\phi}_{n(K+1)}(\bs{s})\phi_k(\bs{s})\mathrm{d}\mu(\bs{s})$. Using arguments similar to those in the proof of Theorem \ref{THM:fn1hat}, since $\sum_{k=1}^{\infty}\widehat{a}_{nk}^2\leq1$ and by $(\ref{EQN:ffV0K})$, we have $\sum_{k=K+1}^{\infty}\widehat{a}_{nk}^2\lambda_k\xrightarrow{\as}\lambda_{K+1}$, which forces $\widehat{a}_{n(K+1)}^2\xrightarrow{\as}1$ and $\sum_{k=1}^K\widehat{a}_{nk}^2+\sum_{k=K+2}^{\infty}\widehat{a}_{nk}^2\xrightarrow{\as}0$. Then there exists a sign sequence $\{S_{n(K+1)}\in\{-1,1\}\}_n$ such that $\|S_{n(K+1)}\widehat{\phi}_{n(K+1)}-\phi_{K+1}\|_{\mu,2}\xrightarrow{\as}0$ as $n\to\infty$.
\end{proof}

\begin{proof} [Proof of Theorem \ref{THM:fHconv}]
Based on these theorems and lemmas introduced above, we are ready to give the detailed proof of Theorem \ref{THM:fHconv}.
\begin{enumerate} [(i)]
    \item The conclusion of (i) follows by Theorems \ref{THM:fn1hat} -- \ref{THM:fnk+1hat} directly.
    \item The conclusion of (ii) has been shown in Theorem \ref{THM:fnkhat}.
    \item It is obvious that for $1\leq k\leq K_n\wedge p_n$,
\[
	\int_{\mathcal{V}\times\mathcal{V}} \widehat{\phi}_{nk}(\bs{s})\widehat{\phi}_{nk}(\bs{s}^{\prime})V_n(\bs{s},\bs{s}^{\prime})\mathrm{d}\mu(\bs{s})\mathrm{d}\mu(\bs{s}^{\prime})
	\geq \int_{\mathcal{V}\times\mathcal{V}} \widehat{\phi}^{\S}_{nk}(\bs{s})\widehat{\phi}^{\S}_{nk}(\bs{s}^{\prime})V_n(\bs{s},\bs{s}^{\prime})\mathrm{d}\mu(\bs{s})\mathrm{d}\mu(\bs{s}^{\prime}).
\]
On the other side, for any $f\in\mathcal{B}_1\cap \mathcal{S}$, $\exists \bs{a} \in\mathbb{R}^{J_n}\cap\mathcal{A}_1$, such that $f(\bs{s})=\bs{a}^{\top}\mathbf{H}^{-1/2}\mathbf{B}(\bs{s})$. Thus
\begin{align*}
	&\int_{\mathcal{V}\times\mathcal{V}} f(\bs{s})f(\bs{s}^{\prime})V_n(\bs{s},\bs{s}^{\prime})\mathrm{d}\mu(\bs{s})\mathrm{d}\mu(\bs{s}^{\prime}) \\
    &=\bs{a}^{\top}\mathbf{H}^{-1/2} \int_{\mathcal{V}\times\mathcal{V}} \mathbf{B}(\bs{s})V_n(\bs{s},\bs{s}^{\prime})\mathbf{B}^{\top}(\bs{s}^{\prime})\mathrm{d}\mu(\bs{s})\mathrm{d}\mu(\bs{s}^{\prime})\mathbf{H}^{-1/2}\bs{a} \\
	&=\bs{a}^{\top}\mathbf{K}_n\bs{a} \leq\widehat{\bs{\varphi}}_{n1}^{\top}\mathbf{K}_n\widehat{\bs{\varphi}}_{n1}
	=\int_{\mathcal{V}\times\mathcal{V}} \widehat{\phi}^{\S}_{n1}(\bs{s})\widehat{\phi}^{\S}_{n1}(\bs{s}^{\prime})V_n(\bs{s},\bs{s}^{\prime})\mathrm{d}\mu(\bs{s})\mathrm{d}\mu(\bs{s}^{\prime}).
\end{align*}
We therefore obtain $\widehat{\phi}^{\S}_{n1}(\bs{s})=\argmax_{f\in\mathcal{B}_1\cap \mathcal{S}} \int_{\mathcal{V}\times\mathcal{V}} f(\bs{s})f(\bs{s}^{\prime})V_n(\bs{s},\bs{s}^{\prime})\mathrm{d}\mu(\bs{s})\mathrm{d}\mu(\bs{s}^{\prime})$ by taking the supremum over $\mathcal{B}_1\cap \mathcal{S}$. Thus, $\widehat{\phi}_{n1}(\bs{s})$ is equal to $\widehat{\phi}^{\S}_{n1}(\bs{s})$ up to sign. 
Suppose $\forall k< K_n\wedge p_n$, we have $\widehat{\phi}_{n(k-1)}(\bs{s})=\widehat{\phi}^{\S}_{n(k-1)}(\bs{s})$. 
Similarly, we can obtain 
\[
    \widehat{\phi}^{\S}_{n_k}(\bs{s})
	=\argmax_{f\in\mathcal{B}_1\cap \mathcal{S}\cap \widehat{\mathcal{H}}_{n(k-1)}^{\perp}} \int_{\mathcal{V}\times\mathcal{V}} f(\bs{s})f(\bs{s}^{\prime})V_n(\bs{s},\bs{s}^{\prime})\mathrm{d}\mu(\bs{s})\mathrm{d}\mu(\bs{s}^{\prime})
\]
by noticing that $\widehat{\bs{\varphi}}_{nk}=\argmax_{\bs{a} \in\mathbb{R}^{J_n}\cap\mathcal{A}_1}\bs{a}^{\top}\widehat{\mathbf{K}}_{nk}\bs{a}$, where 
\[
	\widehat{\mathbf{K}}_{nk}=\mathbf{H}^{-1/2}\int_{\mathcal{V}\times\mathcal{V}}\mathbf{B}(\bs{s})\sum_{k^{\prime}\geq k}\widehat{\lambda}_{nk^{\prime}}\widehat{\phi}_{nk}(\bs{s})\widehat{\phi}_{nk}(\bs{s}^{\prime})\mathbf{B}^{\top}(\bs{s}^{\prime})\mathrm{d}\mu(\bs{s})\mathrm{d}\mu(\bs{s}^{\prime}) \mathbf{H}^{-1/2}.
\]
Hence, for $1\leq k\leq K_n\wedge p_n$, $\widehat{\phi}_{nk}(\bs{s})$ and $\widehat{\phi}^{\S}_{nk}(\bs{s})$ are equal up to sign.
    \item By the definitions of $\widehat{\lambda}_{nk}$ and $\lambda_k$, for $1\leq k\leq K_n\wedge p_n$, we have
\begin{align*}
	\left|\widehat{\lambda}_{nk}-\lambda_k\right|
	&=\left|\int_{\mathcal{V}\times\mathcal{V}} 
	\left\{\widehat{\phi}_{nk}(\bs{s})\widehat{\phi}_{nk}(\bs{s}^{\prime})V_n(\bs{s},\bs{s}^{\prime})
	-\phi_k(\bs{s})\phi_k(\bs{s}^{\prime})V_0(\bs{s},\bs{s}^{\prime})\right\}
	\mathrm{d}\mu(\bs{s})\mathrm{d}\mu(\bs{s}^{\prime}) \right| \\
	&\leq \left|\int_{\mathcal{V}\times\mathcal{V}} \widehat{\phi}_{nk}(\bs{s})\widehat{\phi}_{nk}(\bs{s}^{\prime})\left\{V_n(\bs{s},\bs{s}^{\prime})-V_0(\bs{s},\bs{s}^{\prime})\right\}\mathrm{d}\mu(\bs{s})\mathrm{d}\mu(\bs{s}^{\prime})\right|\\
	&~~~~~~~~~~~~~~~~~~ +\left|\int_{\mathcal{V}\times\mathcal{V}} \left\{\widehat{\phi}_{nk}(\bs{s})\widehat{\phi}_{nk}(\bs{s}^{\prime})-\phi_k(\bs{s})\phi_k(\bs{s}^{\prime})\right\} V_0(\bs{s},\bs{s}^{\prime})\mathrm{d}\mu(\bs{s})\mathrm{d}\mu(\bs{s}^{\prime})\right|. 
\end{align*}
From the proofs of Theorem \ref{THM:GC} and Theorem \ref{THM:fnk+1hat}, as $n\to\infty$,
\[
	\left|\int_{\mathcal{V}\times\mathcal{V}} \widehat{\phi}_{nk}(\bs{s})\widehat{\phi}_{nk}(\bs{s}^{\prime})\left\{V_n(\bs{s},\bs{s}^{\prime})-V_0(\bs{s},\bs{s}^{\prime})\right\}\mathrm{d}\mu(\bs{s})\mathrm{d}\mu(\bs{s}^{\prime})\right|\xrightarrow{\as}0.
\]
Since $\sum_{k=1}^{\infty}\lambda_k<\infty$ and $\lambda_1>\lambda_2>\ldots\geq0$, $\int_{\mathcal{V}\times\mathcal{V}} V_0^2(\bs{s},\bs{s}^{\prime})\mathrm{d}\mu(\bs{s})\mathrm{d}\mu(\bs{s}^{\prime})=\sum_{k=1}^{\infty}\lambda_k^2<\infty$. Thus,
\begin{align*}
	&\left|\int_{\mathcal{V}\times\mathcal{V}} \left\{S_{nk}\widehat{\phi}_{nk}(\bs{s})S_{nk}\widehat{\phi}_{nk}(\bs{s}^{\prime})-\phi_k(\bs{s})\phi_k(\bs{s}^{\prime})\right\} V_0(\bs{s},\bs{s}^{\prime})\mathrm{d}\mu(\bs{s})\mathrm{d}\mu(\bs{s}^{\prime})\right| \\
	\leq& \left|\int_{\mathcal{V}\times\mathcal{V}} S_{nk}\widehat{\phi}_{nk}(\bs{s})\left\{S_{nk}\widehat{\phi}_{nk}(\bs{s}^{\prime})-\phi_k(\bs{s}^{\prime})\right\} V_0(\bs{s},\bs{s}^{\prime})\mathrm{d}\mu(\bs{s})\mathrm{d}\mu(\bs{s}^{\prime})\right|  \\
	&~~~~~~~~~~~~~~~~~~~~~~~~~~~~~~~~~~~~+ \left|\int_{\mathcal{V}\times\mathcal{V}} \phi_k(\bs{s}^{\prime})\left\{S_{nk}\widehat{\phi}_{nk}(\bs{s})-\phi_k(\bs{s})\right\} V_0(\bs{s},\bs{s}^{\prime})\mathrm{d}\mu(\bs{s})\mathrm{d}\mu(\bs{s}^{\prime})\right|\\
	\leq & \|S_{nk}\widehat{\phi}_{nk}\|_{\mu,2}\|S_{nk}f_{nj}-\phi_k\|_{\mu,2} \left\{\int_{\mathcal{V}\times\mathcal{V}} V_0^2(\bs{s},\bs{s}^{\prime})\mathrm{d}\mu(\bs{s})\mathrm{d}\mu(\bs{s}^{\prime})\right\}^{1/2} \\
	&~~~~~~~~~~~~~~~~~~~~~~~~~~~~~~~~~~~~+ \|\phi_k\|_{\mu,2}\|S_{nk}f_{nj}-\phi_k\|_{\mu,2} \left\{\int_{\mathcal{V}\times\mathcal{V}} V_0^2(\bs{s},\bs{s}^{\prime})\mathrm{d}\mu(\bs{s})\mathrm{d}\mu(\bs{s}^{\prime})\right\}^{1/2} \\
	&
	\leq  2 \|S_{nk}f_{nj}-\phi_k\|_{\mu,2} \left\{\int_{\mathcal{V}\times\mathcal{V}} V_0^2(\bs{s},\bs{s}^{\prime})\mathrm{d}\mu(\bs{s})\mathrm{d}\mu(\bs{s}^{\prime})\right\}^{1/2}
	\to0
\end{align*}
almost sure as $n\to\infty$. The conclusion follows.
\end{enumerate}
\end{proof}

\fontsize{12}{14pt plus.8pt minus .6pt}\selectfont
\vskip 0.1in  \noindent \textbf{B. Proof of Convergence of Estimators} \vskip 0.1in 
\renewcommand{\theequation}{B.\arabic{equation}}
\renewcommand{\thesubsection}{B.\arabic{subsection}}
\renewcommand{\thetheorem}{B.\arabic{theorem}}
\renewcommand{\thelemma}{{\rm B.\arabic{lemma}}}
\renewcommand{\theproposition}{B.\arabic{proposition}}
\renewcommand{\thecorollary}{B.\arabic{corollary}}
\renewcommand{\thefigure}{B.\arabic{figure}}
\renewcommand{\thetable}{B.\arabic{table}}
\renewcommand{\theremark}{B.\arabic{remark}}
\setcounter{equation}{0}
\setcounter{theorem}{0}
\setcounter{lemma}{0}
\setcounter{figure}{0}
\setcounter{remark}{0}
\setcounter{proposition}{0}
\setcounter{subsection}{0}
\setcounter{subsubsection}{0}

\label{SUBSEC:proof_estimate}

In this section, we develop the convergence results of the coefficient estimator based on the proposed 2D-FPC bases. Following the sketch proof of Theorem \ref{THM:Conv} in Section 4.2 in the paper, we provide the supporting lemmas and the detailed proof.

We first define those coefficients and operators that will be employed in the proof. Recall that $\widetilde{\bs{\theta}}=(\bs{\alpha}_1^{\top},\bs{\alpha}_2^{\top},\bs{\gamma}_1^{\top},\bs{\gamma}_2^{\top})^{\top}\in\widetilde{\bs{\Theta}}_{K_1,K_2}$, where $\bs{\alpha}_{\ell}\in\mathbb{R}^q$, $\ell=1,2$, are the linear coefficients, and $\bs{\gamma}_{\ell}\in\mathbb{R}^{K_{\ell}}$, $\ell=1,2$, are the basis coefficients.
Let 
$\widetilde{\bs{\theta}}_n=\argmin_{\widetilde{\bs{\theta}}\in\widetilde{\bs{\Theta}}_{K_1,K_2}}n^{-1}\sum_{i=1}^n \{Y_i-\widetilde{\bs{\theta}}(\widehat{\mathbf{W}}_i)\}^2$.
We write $\widetilde{\bs{\theta}}_n=(\widetilde{\bs{\alpha}}_{1n},\widetilde{\bs{\alpha}}_{2n},\widetilde{\bs{\gamma}}_1,\widetilde{\bs{\gamma}}_2)$, with
$\widetilde{\bs{\gamma}}_1=(\widetilde{\gamma}_{1,1},\ldots,\widetilde{\gamma}_{1,K_1})^{\top}$ and $\widetilde{\bs{\gamma}}_2=(\widetilde{\gamma}_{2,1},\ldots,\widetilde{\gamma}_{2,K_2})^{\top}$,
then 
\[
    \widehat{\bs{\theta}}_n= \left(\widetilde{\bs{\alpha}}_{1n},\widetilde{\bs{\alpha}}_{2n},\sum_{k=1}^{K_1}\widetilde{\gamma}_{1,k}\widehat{\phi}_{1,nk}(\bs{s}),\sum_{k=1}^{K_2}\widetilde{\gamma}_{2,k}\widehat{\phi}_{2,nk}(\bs{s})\right)^{\top}.
\]
Recall that $\bs{\theta}_0=(\bs{\alpha}_{01},\bs{\alpha}_{02},\beta_{01},\beta_{02})\in\bs{\Theta}$ is the true parameter value. 
Let $\gamma_{0\ell,k}=\int_{\mathcal{V}}\phi_{\ell,k}(\bs{s})\beta_{0\ell}(\bs{s})\mathrm{d}\mu(\bs{s})$, and thus $\beta_{0\ell}(\bs{s})=\sum_{k=1}^{\infty}\gamma_{0\ell,k}\phi_{\ell,k}(\bs{s})$, $\ell=1,2$, where $\{\phi_{1,k}(\bs{s})\}_k$ and $\{\phi_{2,k}(\bs{s})\}_k$ are theoretical orthonormal bases for $Z(\bs{s})$ and $AZ(\bs{s})$, respectively.
 Define 
\[
    \bs{\theta}_{0n}
        =\left(\bs{\alpha}_{01}^{\top}, \bs{\alpha}_{02}^{\top},  \sum_{k=1}^{K_1}\gamma_{01,k}S_{1,nk}\phi_{1,k}(\bs{s}), \sum_{k=1}^{K_2}\gamma_{02,k}S_{2,nk}\phi_{2,k}(\bs{s})\right)^{\top}, 
\]
and
\[
	\widetilde{\bs{\theta}}_{0n}=\left(\bs{\alpha}_{01}^{\top},\bs{\alpha}_{02}^{\top},(\gamma_{01,1},\ldots,\gamma_{01,K_1}),(\gamma_{02,1},\ldots,\gamma_{02,K_2})\right)^{\top}.
\] 

We first explore the properties of the truncated terms. Let
\begin{align}
\label{DEF:deltaW}
	\Delta_{0n}(\mathbf{W}_i)
	=\sum_{k=K_1+1}^{\infty}\gamma_{01,k}\int_{\mathcal{V}}S_{nk}\phi_k(\bs{s})Z_i(\bs{s})\mathrm{d}\mu(\bs{s}) + \sum_{k=K_2+1}^{\infty}\gamma_{02,k}\int_{\mathcal{V}}S_{nk}\phi_k(\bs{s})A_iZ_i(\bs{s})\mathrm{d}\mu(\bs{s}).
\end{align}
We use the following Lemma \ref{LEM:deltanW} to state that the truncated terms are negligible. 

\begin{lemma}
\label{LEM:deltanW}
Under Assumptions (A3) -- (A8), 
as $n\to\infty$, $n^{-1}\sum_{i=1}^n\Delta_{0n}^2(\mathbf{W}_i)\to0$.
\end{lemma}
\begin{proof}
Decompose $\Delta_{0n}(\mathbf{W}_i)=I_1+I_2$, where
\[
 	 I_1 = \sum_{k=K_1+1}^{\infty}\gamma_{01,k}\int_{\mathcal{V}}S_{nk}\phi_k(\bs{s})Z_i(\bs{s})\mathrm{d}\mu(\bs{s}), ~~~~~~
	 I_2 = \sum_{k=K_2+1}^{\infty}\gamma_{02,k}\int_{\mathcal{V}}S_{nk}\phi_k(\bs{s})A_iZ_i(\bs{s})\mathrm{d}\mu(\bs{s}) .
\]
We have
\begin{align*}
    \mathrm{E}I_1^2
     &=\sum_{k=K_1+1}^{\infty}\sum_{k^{\prime}=K_1+1}^{\infty} \gamma_{01,k}\gamma_{01,k^{\prime}} \int_{\mathcal{V}\times\mathcal{V}}S_{nk}S_{nk^{\prime}}\phi_k(\bs{s})\phi_{k^{\prime}}(\bs{s}^{\prime})\left\{\mathrm{E}Z_i(\bs{s})Z_i(\bs{s}^{\prime})\right\}\mathrm{d}\mu(\bs{s})\mathrm{d}\mu(\bs{s}^{\prime}) \\
    &=\sum_{k=K_1+1}^{\infty}\sum_{k^{\prime}=K_1+1}^{\infty} \gamma_{01,k}\gamma_{01,k^{\prime}} \int_{\mathcal{V}\times\mathcal{V}}S_{nk}S_{nk^{\prime}}\phi_k(\bs{s})\phi_{k^{\prime}}(\bs{s}^{\prime})\left\{\sum_{k^{\prime\prime}=1}^{\infty}\lambda_{k^{\prime\prime}}\phi_{k^{\prime\prime}}(\bs{s})\phi_{k^{\prime\prime}}(\bs{s}^{\prime})\right\}\mathrm{d}\mu(\bs{s})\mathrm{d}\mu(\bs{s}^{\prime})\\
	&=\sum_{k=K_1+1}^{\infty}\sum_{k^{\prime}=K_1+1}^{\infty}\sum_{k^{\prime\prime}=1}^{\infty}\gamma_{01,k}\gamma_{01,k^{\prime}}\lambda_{k^{\prime\prime}} S_{nk}S_{nk^{\prime}}
		\int_{\mathcal{V}}\phi_k(\bs{s})\phi_{k^{\prime\prime}}(\bs{s})\mathrm{d}\mu(\bs{s})
		\int_{\mathcal{V}}\phi_{k^{\prime}}(\bs{s}^{\prime})\phi_{k^{\prime\prime}}(\bs{s}^{\prime})\mathrm{d}\mu(\bs{s}^{\prime}) \\
	&=\sum_{k=K_1+1}^{\infty} \gamma_{01,k}^2\lambda_k.
\end{align*}
Similarly, $\mathrm{E}I_2^2\leq\sum_{k=K_2+1}^{\infty} \gamma_{02,k}^2\lambda_k$, and $\mathrm{E}I_1 I_2\leq \sum_{k=K_1\vee K_2+1}^{\infty} \gamma_{01,k}\gamma_{02,k}\lambda_k$.
Then as $n\to\infty$,
\begin{align*}
	n^{-1}\sum_{i=1}^n\Delta_{0n}^2(\mathbf{W}_i)
	&\to \mathrm{E} \Delta_{0n}^2(\mathbf{W})
	\leq 2\left(\sum_{k=K_1+1}^{\infty} \gamma_{01,k}^2\lambda_k + \sum_{k=K_2+1}^{\infty} \gamma_{02,k}^2\lambda_k\right)  \\
	&\leq 2\left(\sum_{k=K_1+1}^{\infty}\gamma_{01,k}^2 + \sum_{k=K_2+1}^{\infty}\gamma_{02,k}^2\right) \left(\sum_{k=K_1\wedge K_2+1}^{\infty}\lambda_k\right)
	\equiv\mathrm{E}_{1n}^{\ast}(K_1,K_2).
\end{align*}
Thus, we have $\infty > \mathrm{E}_{1n}^{\ast}(K_1,K_2) \to 0$ as $K_1\wedge K_2\to\infty$.
\end{proof}

Next, we use the following Lemma \ref{LEM:thetaW} to state that the errors induced by the approximation of $\{\widehat{\phi}_{nk}\}$ to $\{\phi_{k}\}$ are negligible.
\begin{lemma}
\label{LEM:thetaW}
Under Assumptions (A3) -- (A8),
as $n\to\infty$, 
$n^{-1}\sum_{i=1}^n\{\bs{\theta}_{0n}(\mathbf{W}_i) - \widetilde{\bs{\theta}}_{0n}(\widehat{\mathbf{W}}_i)\}^2\xrightarrow{\as}0$.
\end{lemma}
\begin{proof}
By the definitions of $\bs{\theta}_{0n}(\mathbf{W}_i)$ and $\widetilde{\bs{\theta}}_{0n}(\widehat{\mathbf{W}}_i)$,
\begin{align*}
	\bs{\theta}_{0n}(\mathbf{W}_i) - \widetilde{\bs{\theta}}_{0n}(\widehat{\mathbf{W}}_i) 
	=& \sum_{k=1}^{K_1}\gamma_{01,k}\int_{\mathcal{V}}\left\{S_{nk}\phi_k(\bs{s})-\widehat{\phi}_{nk}(\bs{s})\right\}Z_i(\bs{s})\mathrm{d}\mu(\bs{s}) \\
	&~~~~~~~~~~~~~~~~~~~~~~~~ + \sum_{k=1}^{K_2}\gamma_{02,k}\int_{\mathcal{V}}\left\{S_{nk}\phi_k(\bs{s})-\widehat{\phi}_{nk}(\bs{s})\right\}A_iZ_i(\bs{s})\mathrm{d}\mu(\bs{s}).
\end{align*}
So
$
	|\bs{\theta}_{0n}(\mathbf{W}_i) - \widehat{\bs{\theta}}_{0n}(\widehat{\mathbf{W}}_i)| 
	\leq \sum_{k=1}^{K_1}\gamma_{01,k}\|S_{nk}\phi_k-\widehat{\phi}_{nk}\|_{\mu,2}\|Z_i\|_{\mu,2}
		+ \sum_{k=1}^{K_2}\gamma_{02,k}\|S_{nk}\phi_k-\widehat{\phi}_{nk}\|_{\mu,2}\|Z_i\|_{\mu,2}
$
since $|A_i|\leq 1$. 
By Theorem \ref{THM:fHconv}, 
there exist a sequence $\{r_n\}$, such that $r_n\|S_{nk}\phi_k-\widehat{\phi}_{nk}\|_{\mu,2}\to0 $. Take $K_1\vee K_2=r_n^{-1/2}$,
as $n\to\infty$,  
\begin{align*}
	&n^{-1}\sum_{i=1}^n\left\{\bs{\theta}_{0n}(\mathbf{W}_i) - \widetilde{\bs{\theta}}_{0n}(\widehat{\mathbf{W}}_i)\right\}^2 \\
	&~~~~~~~~~~~~~~~~~~~~~~~\leq \left(\sum_{j=1}^{K_1\vee K_2}\gamma_{01,k}^2 + \gamma_{02,k}^2\right) \left(2\sum_{j=1}^{K_1\vee K_2}\|S_{nk}\phi_k-\widehat{\phi}_{nk}\|_{\mu,2}^2\right)n^{-1}\sum_{i=1}^n\|Z_i\|_{\mu,2}^2
	\xrightarrow{\as}{0}.  
\end{align*}
\end{proof}

Combining Lemmas \ref{LEM:deltanW} and \ref{LEM:thetaW}, we use the following Lemma \ref{LEM:ConvOrac} to illuminate the convergence rate of $\widetilde{\bs{\theta}}_{0n}$ to the observed $Y_i$.
\begin{lemma}
\label{LEM:ConvOrac}
Under Assumptions (A3) -- (A8), as $n\to\infty$, $n^{-1}\sum_{i=1}^n\{Y_i-\widetilde{\bs{\theta}}_{0n}(\widehat{\mathbf{W}}_i)\}^2\xrightarrow{\as}C_1$ for some constant $C_1$.
\end{lemma}
\begin{proof}
By definition, 
\begin{eqnarray*}
 	n^{-1}\sum_{i=1}^n\left\{Y_i-\widetilde{\bs{\theta}}_{0n}(\widehat{\mathbf{W}}_i)\right\}^2
	&=&n^{-1}\sum_{i=1}^n\left\{\bs{\theta}_0(\mathbf{W}_i) + \varepsilon_i -\widetilde{\bs{\theta}}_{0n}(\widehat{\mathbf{W}}_i)\right\}^2 \\
	&=&n^{-1}\sum_{i=1}^n \left\{\varepsilon_i +\Delta_{0n}(\mathbf{W}_i) + \bs{\theta}_{0n}(\mathbf{W}_i) -\widetilde{\bs{\theta}}_{0n}(\widehat{\mathbf{W}}_i)\right\}^2~~~\as,
\end{eqnarray*}
where $\Delta_{0n}(\mathbf{W}_i)$ is given in (\ref{DEF:deltaW}).
Since $n^{-1}\sum_{i=1}^n \varepsilon_i^2 \to \mathrm{Var}(\varepsilon)= \sigma^2 < \infty$, combined with Lemmas \ref{LEM:deltanW} and \ref{LEM:thetaW}, we have
\begin{align}
\label{EQN:Y2theta0nWhat}
	n^{-1}\sum_{i=1}^n\left\{Y_i - \widetilde{\bs{\theta}}_{0n}(\widehat{\mathbf{W}}_i)\right\}^2 
	\xrightarrow{\as} C_1 < \infty 
\end{align}
for some constant $C_1$.
\end{proof}

Note that 
$
 	n^{-1}\sum_{i=1}^n\{Y_i-\widetilde{\bs{\theta}}_n(\widehat{\mathbf{W}}_i)\}^2
	\leq n^{-1}\sum_{i=1}^n\{Y_i-\widetilde{\bs{\theta}}_{0n}(\widehat{\mathbf{W}}_i)\}^2. 
$ Consequently,
\begin{align}
\label{EQN:DiffDelta}
	-\frac{2}{n}\sum_{i=1}^n \left\{Y_i-\widetilde{\bs{\theta}}_{0n}(\widehat{\mathbf{W}}_i)\right\} \left\{\widetilde{\bs{\theta}}_n(\widehat{\mathbf{W}}_i)-\widetilde{\bs{\theta}}_{0n}(\widehat{\mathbf{W}}_i)\right\} + \frac{1}{n} \sum_{i=1}^n \left\{\widetilde{\bs{\theta}}_n(\widehat{\mathbf{W}}_i)-\widetilde{\bs{\theta}}_{0n}(\widehat{\mathbf{W}}_i)\right\}^2 
	\leq 0 .
\end{align}
Let $\widetilde{\bs{\delta}}_n=\widetilde{\bs{\theta}}_n-\widetilde{\bs{\theta}}_{0n}$ and $\widetilde{\bs{r}}_n=\widetilde{\bs{\delta}}_n / \|\widetilde{\bs{\delta}}_n\|_2$. 
We use the following Lemma \ref{LEM:OrderDelta} to build the boundedness of $\widetilde{\bs{\delta}}_n$.

\begin{lemma}
\label{LEM:OrderDelta}
Under Assumptions (A3) -- (A10), $\limsup_{n\to\infty} \|\widetilde{\bs{\delta}}_n\|_2 \leq C_2  < \infty$ almost surely for some constant $C_2$.
\end{lemma}
\begin{proof}
\begin{align}
	(\ref{EQN:DiffDelta}) &\Rightarrow {\frac{\|\widetilde{\bs{\delta}}_n\|_2^2}{1+\|\widetilde{\bs{\delta}}_n\|_2}} \frac{1}{n}\sum_{i=1}^n\{\widetilde{\bs{r}}_n(\widehat{\mathbf{W}}_i)\}^2
	\leq \frac{2}{n}\sum_{i=1}^n \left\{Y_i-\widetilde{\bs{\theta}}_{0n}(\widehat{\mathbf{W}}_i)\right\} \widetilde{\bs{r}}_n(\widehat{\mathbf{W}}_i) \nonumber \\
	&~~~~~~~~~~~~~~~~~~~~~~~~~~~~~~~~~~~~~~~~~~\leq 2\left[\frac{1}{n}\sum_{i=1}^n \left\{Y_i-\widetilde{\bs{\theta}}_{0n}(\widehat{\mathbf{W}}_i)\right\}^2\right]^{1/2} 
		\left[\frac{1}{n}\sum_{i=1}^n \left\{\widetilde{\bs{r}}_n(\widehat{\mathbf{W}}_i)\right\}^2\right]^{1/2}, \nonumber \\
	&\Rightarrow {\frac{\|\widetilde{\bs{\delta}}_n\|_2^2}{1+\|\widetilde{\bs{\delta}}_n\|_2}}
	\leq \frac{2\left[\frac{1}{n}\sum_{i=1}^n \left\{Y_i-\widetilde{\bs{\theta}}_{0n}(\widehat{\mathbf{W}}_i)\right\}^2\right]^{1/2} }
		{\left[\frac{1}{n}\sum_{i=1}^n \left\{\widetilde{\bs{r}}_n(\widehat{\mathbf{W}}_i)\right\}^2\right]^{1/2}}. \label{EQN:OrderDelta}
\end{align}
By Lemma \ref{LEM:ConvOrac}, the numerator on the right-hand-side (RHS) of (\ref{EQN:OrderDelta}) $\xrightarrow{\as}2C_1^{1/2}$ as $n\to\infty$. 

As for the denominator on the RHS of $(\ref{EQN:OrderDelta})$, we first define

\[
    \widehat{\mathbf{D}}_n 
	= \frac{1}{n} \sum_{i=1}^n \widehat{\mathbf{W}}_i\widehat{\mathbf{W}}_i^{\top}
	= \begin{pmatrix}
		\widehat{\mathbf{D}}_{11} & \widehat{\mathbf{D}}_{12} & \widehat{\mathbf{D}}_{13} & \widehat{\mathbf{D}}_{14} \\
		\widehat{\mathbf{D}}_{12}^{\top} & \widehat{\mathbf{D}}_{22} & \widehat{\mathbf{D}}_{23} & \widehat{\mathbf{D}}_{24} \\
		\widehat{\mathbf{D}}_{13}^{\top} & \widehat{\mathbf{D}}_{23}^{\top} & \widehat{\mathbf{D}}_{33} & \widehat{\mathbf{D}}_{34} \\
		\widehat{\mathbf{D}}_{14}^{\top} & \widehat{\mathbf{D}}_{24}^{\top} & \widehat{\mathbf{D}}_{34}^{\top} & \widehat{\mathbf{D}}_{44} \\
	\end{pmatrix},
\]
where for each entry of $\widehat{\mathbf{D}}_n$,

\begin{tabular}{lll}
	& $\widehat{\mathbf{D}}_{11}=n^{-1}\sum_{i=1}^n \mathbf{X}_i\mathbf{X}_i^{\top}$, &
	$\widehat{\mathbf{D}}_{12}=\widehat{\mathbf{D}}_{22}=n^{-1}\sum_{i=1}^n A_i\mathbf{X}_i\mathbf{X}_i^{\top}$, \\
	& $\widehat{\mathbf{D}}_{13}=n^{-1}\sum_{i=1}^n \mathbf{X}_i\widehat{\mathbf{U}}_{1,i}^{\top}(K_1)$, &
	$\widehat{\mathbf{D}}_{14}=\widehat{\mathbf{D}}_{24}=n^{-1}\sum_{i=1}^n A_i\mathbf{X}_i\widehat{\mathbf{U}}_{2,i}^{\top}(K_2)$, \\
	& $\widehat{\mathbf{D}}_{23}=n^{-1}\sum_{i=1}^n A_i\mathbf{X}_i\widehat{\mathbf{U}}_{1,i}^{\top}(K_1)$, &
	$\widehat{\mathbf{D}}_{33}=n^{-1}\sum_{i=1}^n \widehat{\mathbf{U}}_{1,i}(K_1)\widehat{\mathbf{U}}_{1,i}^{\top}(K_1)$,\\
	& $\widehat{\mathbf{D}}_{34}=n^{-1}\sum_{i=1}^n A_i\widehat{\mathbf{U}}_{1,i}(K_1)\widehat{\mathbf{U}}_{2,i}^{\top}(K_2)$, &
	$\widehat{\mathbf{D}}_{44}=n^{-1}\sum_{i=1}^n A_i\widehat{\mathbf{U}}_{2,i}(K_2)\widehat{\mathbf{U}}_{2,i}^{\top}(K_2)$. \\
\end{tabular} \\
Then we have 
$n^{-1}\sum_{i=1}^n\{\widetilde{\bs{r}}_n(\widehat{\mathbf{W}}_i)\}^2
	=\widetilde{\bs{r}}_n^{\top}\widehat{\mathbf{D}}_n\widetilde{\bs{r}}_n$.
 
In a parallel fashion to $\widehat{\mathbf{W}}_i$ and $\widehat{\mathbf{D}}_n$,
for $i=1,\ldots,n$,
let $\widetilde{\mathbf{W}}_i=(\mathbf{X}_i^{\top}, A_i\mathbf{X}_i^{\top},\widetilde{\mathbf{U}}_{1,i}^{\top}(K_1), A_i\widetilde{\mathbf{U}}_{2,i}^{\top}(K_2))^{\top}$, where $\widetilde{\mathbf{U}}_{\ell,i}(K_{\ell})=(\widetilde{U}_{\ell,i1},\ldots,\widetilde{U}_{\ell,iK_{\ell}})^{\top}$, $\ell=1,2$, and 
\begin{eqnarray*}
    \widetilde{U}_{1,ik}&=&\int_{\mathcal{V}}S_{1,nk}\phi_{1,k}(\bs{s})Z_i(\bs{s})\mathrm{d}\mu(\bs{s}), ~~1\leq k\leq K_1, \\
    \widetilde{U}_{2,ik}&=&\int_{\mathcal{V}}S_{2,nk}\phi_{2,k}(\bs{s})A_iZ_i(\bs{s})\mathrm{d}\mu(\bs{s}), ~~ 1\leq k\leq K_2.
\end{eqnarray*}
Then define 
\[
    \mathbf{D} = \mathrm{E} \widetilde{\mathbf{W}}_i\widetilde{\mathbf{W}}_i^{\top}
	= \begin{pmatrix}
		\mathbf{D}_{11} & \mathbf{D}_{12} & \mathbf{D}_{13} & \mathbf{D}_{14} \\
		\mathbf{D}_{12}^{\top} & \mathbf{D}_{22} & \mathbf{D}_{23} & \mathbf{D}_{24} \\
		\mathbf{D}_{13}^{\top} & \mathbf{D}_{23}^{\top} & \mathbf{D}_{33} & \mathbf{D}_{34} \\
		\mathbf{D}_{14}^{\top} & \mathbf{D}_{24}^{\top} & \mathbf{D}_{34}^{\top} & \mathbf{D}_{44} \\
	\end{pmatrix},
\]
where for each entry of $\mathbf{D}$,
\begin{eqnarray*}
	\mathbf{D}_{11} &=& \mathrm{E}(\mathbf{X}_i\mathbf{X}_i^{\top}),~~
	\mathbf{D}_{12}=\mathbf{D}_{22}=\mathrm{E}(A_i\mathbf{X}_i\mathbf{X}_i^{\top}), \\
	\mathbf{D}_{13} &=& \mathrm{E}\{\mathbf{X}_i\widetilde{\mathbf{U}}_{1,i}^{\top}(K_1)\}=\int_{\mathcal{V}}\mathrm{E}\left\{\mathbf{X}_i Z_i(\bs{s})\Big(S_{1,n1}\phi_{1,1}(\bs{s}),\ldots,S_{1,nK_1}\phi_{1,K_1}(\bs{s})\Big)\right\}\mathrm{d}\mu(\bs{s}), \\
	\mathbf{D}_{14} &=& \mathbf{D}_{24}=\mathrm{E}\{A_i\mathbf{X}_i\widetilde{\mathbf{U}}_{2,i}^{\top}(K_2)\}\\
    &=&\int_{\mathcal{V}}\mathrm{E}\left[A_i\mathbf{X}_i \left\{A_iZ_i(\bs{s})\right\}\Big(S_{2,n1}\phi_{2,1}(\bs{s}),\ldots,S_{2,nK_2}\phi_{2,K_2}(\bs{s})\Big)\right]\mathrm{d}\mu(\bs{s}), 
\end{eqnarray*}
\begin{eqnarray*}
	\mathbf{D}_{23}&=&\mathrm{E}\{A_i\mathbf{X}_i\widetilde{\mathbf{U}}_{1,i}^{\top}(K_1)\}=\int_{\mathcal{V}}\mathrm{E}\left\{A_i\mathbf{X}_i Z_i(\bs{s})\Big(S_{1,n1}\phi_{1,1}(\bs{s}),\ldots,S_{1,nK_1}\phi_{1,K_1}(\bs{s})\Big)\right\}\mathrm{d}\mu(\bs{s}), 
\end{eqnarray*}
\begin{eqnarray*}
	\mathbf{D}_{33}&=&\mathrm{E}\{\widetilde{\mathbf{U}}_{1,i}(K_1)\widetilde{\mathbf{U}}_{1,i}^{\top}(K_1)\}
		=\int_{\mathcal{V}\times\mathcal{V}}\mathrm{E}\left\{Z_i(\bs{s})Z_i(\bs{s}^{\prime})\times\right. \\
		&&\left.~~~~~~~~~~~~~~~~~~~~~~~~
            \begin{pmatrix}
			S_{1,n1}\phi_{1,1}(\bs{s})\\
			\vdots \\
			S_{1,nK_1}\phi_{1,K_1}(\bs{s}) \\
		\end{pmatrix} \Big(S_{1,n1}\phi_{1,1}(\bs{s}^{\prime}),\ldots,S_{1,nK_1}\phi_{1,K_1}(\bs{s}^{\prime})\Big)\right\}\mathrm{d}\mu(\bs{s})\mathrm{d}\mu(\bs{s}^{\prime}),
\end{eqnarray*}
\begin{eqnarray*}
	\mathbf{D}_{34}&=&\mathrm{E}\{A_i\widetilde{\mathbf{U}}_{1,i}(K_1)\widetilde{\mathbf{U}}_{2,i}^{\top}(K_2)\}
		=\int_{\mathcal{V}\times\mathcal{V}}\mathrm{E}\left[A_i\left\{Z_i(\bs{s})\right\}\left\{A_iZ_i(\bs{s}^{\prime})\right\}\times\right. \\
		&&\left.~~~~~~~~~~~~~~~~~~~~~~~~
            \begin{pmatrix}
			S_{1,n1}\phi_{1,1}(\bs{s})\\
			\vdots \\
			S_{1,nK_1}\phi_{1,K_1}(\bs{s}) \\
		\end{pmatrix}
    \Big(S_{2,n1}\phi_{2,1}(\bs{s}^{\prime}),\ldots,S_{2,nK_2}\phi_{2,K_2}(\bs{s}^{\prime})\Big)\right]\mathrm{d}\mu(\bs{s})\mathrm{d}\mu(\bs{s}^{\prime}),
\end{eqnarray*}
\begin{eqnarray*}
	\mathbf{D}_{44}&=&\mathrm{E}\{A_i\widetilde{\mathbf{U}}_{2,i}(K_2)\widetilde{\mathbf{U}}_{2,i}^{\top}(K_2)\}
		=\int_{\mathcal{V}\times\mathcal{V}}\mathrm{E}\left[A_i\left\{A_iZ_i(\bs{s})\right\}\left\{A_iZ_i(\bs{s}^{\prime})\right\}\times\right. \\
		&&\left.~~~~~~~~~~~~~~~~~~~~~~~~~
            \begin{pmatrix}
			S_{2,n1}\phi_{2,1}(\bs{s})\\
			\vdots \\
			S_{2,nK_2}\phi_{2,K_2}(\bs{s}) \\
		\end{pmatrix}
        \Big(S_{2,n1}\phi_{2,1}(\bs{s}^{\prime}),\ldots,S_{2,nK_2}\phi_{2,K_2}(\bs{s}^{\prime})\Big)\right]\mathrm{d}\mu(\bs{s})\mathrm{d}\mu(\bs{s}^{\prime}).
\end{eqnarray*}
By strong law of large numbers and Theorem \ref{THM:GC}, we can recycle previous arguments to verify that $\|\widehat{\mathbf{D}}_n-\mathbf{D}\|_{\infty}\xrightarrow{\as}0$.
From Assumptions (A3) and (A9), we know that the minimum eigenvalue of $\mathbf{D}\geq c_1$ for some constant $c_1>0$. Hence  
\[
	\liminf_{n\to\infty} \frac{1}{n}\sum_{i=1}^n \left\{\widetilde{\bs{r}}_n(\widehat{\mathbf{W}}_i)\right\}^2 \geq c_1 > 0, ~~~\as
\]
Let $C_2=2C_1^{1/2}c_1^{-1}$, then the $\limsup_{n\to\infty}$ of the RHS of (\ref{EQN:OrderDelta}) $\leq C_2<\infty$.

Now, {if $\frac{a^2}{1+a}\leq k$ for all $a\geq0$ and some $k<\infty$, then $a\leq k+\sqrt{k}$ (follows from algebra).} Thus,
\begin{align*}
	\limsup_{n\to\infty} \|\widetilde{\bs{\delta}}_n\|_2 \leq {C_2 + \sqrt{C_2}} < \infty. 
\end{align*}
\end{proof}

The following Theorem \ref{THM:CompactTheta} gives the tightness of $\widehat{\bs{\theta}}_n$. 
\begin{theorem}
\label{THM:CompactTheta}
Under Assumptions (A5) -- (A10), for fixed $1\leq K_1,K_2<\infty$, $\exists ~C_{\theta}<\infty$, such that
\begin{equation}
\label{EQN:CompactTheta}
	\limsup_{n\to\infty}\|\widehat{\bs{\theta}}_n\|_{\ast} \leq C_{\theta} ~~~\as
\end{equation}
\end{theorem}

\begin{proof}[Proof of Theorem \ref{THM:CompactTheta}]
The conclusion follows directly from Lemmas \ref{LEM:thetaW} and \ref{LEM:OrderDelta}, and the isomorphism between $\widehat{\bs{\Theta}}_{K_1,K_2}$ and $\widetilde{\bs{\Theta}}_{K_1,K_2}$.
\end{proof}

\begin{lemma}
\label{LEM:ConvThetahat}
Under Assumptions (A3) -- (A10), as $n\to\infty$,
\begin{equation}
\label{EQN:ConvThetahat}
	\mathbb{P}_n \left\{Y-\widehat{\bs{\theta}}_n(\mathbf{W})\right\}^2
		- P\left\{Y-\widehat{\bs{\theta}}_n(\mathbf{W})\right\}^2
		\xrightarrow{\as}0.
\end{equation}
\end{lemma}
\begin{proof}
By the isomorphism between $\widehat{\bs{\Theta}}_{K_1,K_2}$ and $\widetilde{\bs{\Theta}}_{K_1,K_2}$, it is equivalent to show 
\[\mathbb{P}_n \left\{Y-\widetilde{\bs{\theta}}_n(\widehat{\mathbf{W}})\right\}^2
		- P\left\{Y-\widetilde{\bs{\theta}}_n(\widehat{\mathbf{W}})\right\}^2
		\xrightarrow{\as}0.
\] 
Note that
\begin{align*}
    \widetilde{\bs{\theta}}_n(\widehat{\mathbf{W}}) 
    &= \widetilde{\bs{\alpha}}_{1n}^{\top}\mathbf{X} + \widetilde{\bs{\alpha}}_{2n}^{\top}A\mathbf{X} + 
    \sum_{k=1}^{K_1}\widetilde{\gamma}_{1,k}\widehat{U}_{1,k}
    + 
    \sum_{k=1}^{K_2}\widetilde{\gamma}_{2,k}\widehat{U}_{2,k} \\
    &= \widetilde{\bs{\alpha}}_{1n}^{\top}\mathbf{X} + \widetilde{\bs{\alpha}}_{2n}^{\top}A\mathbf{X} + 
    \int_{\mathcal{V}} \sum_{k=1}^{K_1}\widetilde{\gamma}_{1,k}\widehat{\phi}_{1,nk}(\bs{s})Z(\bs{s})\mathrm{d}\mu(\bs{s})
    + 
    \int_{\mathcal{V}} \sum_{k=1}^{K_2}\widetilde{\gamma}_{2,k}\widehat{\phi}_{2,nk}(\bs{s})AZ(\bs{s})\mathrm{d}\mu(\bs{s}) \\
    &\triangleq \widehat{f}_n(\mathbf{W}).
\end{align*}
By Theorem \ref{THM:GC} and Lemma \ref{LEM:OrderDelta}, 
$\widehat{f}_n$ is contained in a Glivenko-Cantelli (G-C) class with probability going to one as $n\to\infty$. 
By Glivenko-Cantelli Preservation, $Y-\widetilde{\bs{\theta}}_n(\widehat{\mathbf{W}})$ is also a G-C class, and by Theorem \ref{THM:GC},  $\{Y-\widetilde{\bs{\theta}}_n(\widehat{\mathbf{W}})\}^2$ is also a G-C class, and thus the conclusion follows.
\end{proof}

With all these preparations, we can show the detailed proof of Theorem \ref{THM:Conv}.

\begin{proof}[Proof of Theorem \ref{THM:Conv}]
Let $\bs{\bs{\Omega}}$ be the set of $\bs{\omega}$ for which both (\ref{EQN:CompactTheta}) and (\ref{EQN:ConvThetahat}) hold. Fix $\bs{\omega}\in\bs{\Omega}$. By definition of $\bs{\Omega}$, $\exists$ a subsequence $\{n^{\prime}\}$ such that along this subsequence 
\[
	\widehat{\bs{\theta}}_{n^{\prime}}\to\bs{\theta}_{\ast}\in\bs{\Theta}_{K_1,K_2},
\]
for some limit $\bs{\theta}_{\ast}$, where $\bs{\Theta}_{K_1,K_2}$ is given in (\ref{DEF:ThetaK}). 
Note that by Assumption (A8), $P\{Y-\bs{\theta}(\mathbf{W})\}^2$ has a unique minimizer over $\bs{\Theta}_{K_1,K_2}$, $\bs{\theta}_0^{\ast}$. 
Let 
\begin{eqnarray*}
    r_{1n} &= & \left|\mathbb{P}_n\left\{Y-\bs{\theta}_0^{\ast}(\mathbf{W})\right\}^2 - P\left\{Y-\bs{\theta}_0^{\ast}(\mathbf{W})\right\}^2\right|,  \\
    r_{2n} &= &
	\left|\mathbb{P}_n\left\{Y-\widehat{\bs{\theta}}_n(\mathbf{W})\right\}^2 - P\left\{Y-\widehat{\bs{\theta}}_n(\mathbf{W})\right\}^2\right|.
\end{eqnarray*}
By Lemma \ref{LEM:ConvThetahat}, $r_{2n}\to0$ as $n\to\infty$. Thus,
\begin{align*}
	P\{Y-\bs{\theta}_0^{\ast}(\mathbf{W})\}^2
	&\geq \mathbb{P}_n \{Y_i-\bs{\theta}_0^{\ast}(\mathbf{W})\}^2 - r_{1n}
	\geq \mathbb{P}_n \{Y_i-\widehat{\bs{\theta}}_n(\mathbf{W})\}^2 - r_{1n} \\
	&\geq P\{Y-\widehat{\bs{\theta}}_n(\mathbf{W})\}^2 - r_{1n} - r_{2n} \\
	&\to P\{Y-\bs{\theta}_{\ast}(\mathbf{W})\}^2 - \lim r_{1n}
	\geq P\{Y-\bs{\theta}_0^{\ast}(\mathbf{W})\}^2,
\end{align*}
as $n\to\infty$, where the convergence is attained along the subsequence $n^{\prime}$.
Since the subsequence is arbitrary, and with the help of the strong law, we can conclude that $\widehat{\bs{\theta}}_n\to\bs{\theta}_0^{\ast}$ and
$
	r_{1n} 
	\to 0$ as $n\to\infty$.
Thus, for fixed $1\leq K_1,K_2 <\infty$, integrates $P$  over $\bs{\omega}$ in $\Omega$, we have
\begin{equation}
\label{EQN:YThetaW}
	P\left\{Y-\widehat{\bs{\theta}}_n(\mathbf{W})\right\}^2 - P\left\{Y-\bs{\theta}_0^{\ast}(\mathbf{W})\right\}^2
	\xrightarrow{\as}0.
\end{equation}
However, this is also true if we allow $K_{1n},K_{2n}\to\infty$ sufficiently slowly.
Recycling previous arguments as in the proofs of Lemmas \ref{LEM:deltanW} -- \ref{LEM:OrderDelta}, we can now show that 
\begin{equation}
\label{EQN:Theta0W}
	P\{\bs{\theta}_0(\mathbf{W})-\bs{\theta}_0^{\ast}(\mathbf{W})\}^2\to0, ~~~~\text{ as }~ K_{1n},K_{2n}\to\infty. 
\end{equation}
Then 
\begin{align*}
	&P\{Y-\widehat{\bs{\theta}}_n(\mathbf{W})\}^2 - P\{Y-\bs{\theta}_0^{\ast}(\mathbf{W})\}^2 \\
	=&P\left[\left\{Y-\bs{\theta}_0^{\ast}(\mathbf{W})\right\}-\left\{\widehat{\bs{\theta}}_n(\mathbf{W})-\bs{\theta}_0^{\ast}(\mathbf{W})\right\}\right]^2 - P\{Y-\bs{\theta}_0^{\ast}(\mathbf{W})\}^2  \\
	=&P\left\{\widehat{\bs{\theta}}_n(\mathbf{W})-\bs{\theta}_0^{\ast}(\mathbf{W})\right\}^2 
		-2P\left[\left\{Y-\bs{\theta}_0^{\ast}(\mathbf{W})\right\}\left\{\widehat{\bs{\theta}}_n(\mathbf{W})-\bs{\theta}_0^{\ast}(\mathbf{W})\right\}\right] .
\end{align*}
Let $E_n = 2P[\{Y-\bs{\theta}_0^{\ast}(\mathbf{W})\}\{\widehat{\bs{\theta}}_n(\mathbf{W})-\bs{\theta}_0^{\ast}(\mathbf{W})\}]$. By (\ref{EQN:Theta0W}), we have 
\begin{align*}
	|E_{n}| 
	\leq& 2\left|P\left[\left\{\varepsilon + \bs{\theta}_0(\mathbf{W}) -\bs{\theta}_0^{\ast}(\mathbf{W})\right\}\left\{\widehat{\bs{\theta}}_n(\mathbf{W})-\bs{\theta}_0^{\ast}(\mathbf{W})\right\}\right]\right| \\
	=&2\left|P\left[\left\{\bs{\theta}_0(\mathbf{W})-\bs{\theta}_0^{\ast}(\mathbf{W})\right\}\left\{\widehat{\bs{\theta}}_n(\mathbf{W})-\bs{\theta}_0^{\ast}(\mathbf{W})\right\}\right]\right| \\
	\leq&\left[P\left\{\bs{\theta}_0(\mathbf{W})-\bs{\theta}_0^{\ast}(\mathbf{W})\right\}^2\right]^{1/2} \left[P\left\{\widehat{\bs{\theta}}_n(\mathbf{W})-\bs{\theta}_0^{\ast}(\mathbf{W})\right\}^2\right]^{1/2} 
	\xrightarrow{\as} 0,
\end{align*}
as $n\to\infty$.
Combining with (\ref{EQN:YThetaW}) above, we obtain that $P\{\widehat{\bs{\theta}}_n(\mathbf{W})-\bs{\theta}_0^{\ast}(\mathbf{W})\}^2\xrightarrow{\as}0$. Thus combining with (\ref{EQN:Theta0W}), as $n\to\infty$, the conclusion follows by
\begin{eqnarray*}
	P\left\{\widehat{\bs{\theta}}_n(\mathbf{W})-\bs{\theta}_0(\mathbf{W})\right\}^2
	&=&P\left\{\widehat{\bs{\theta}}_n(\mathbf{W}) -\bs{\theta}_0^{\ast}(\mathbf{W}) + \bs{\theta}_0^{\ast}(\mathbf{W})-\bs{\theta}_0(\mathbf{W})\right\}^2 \\
	&\xrightarrow{\as}& P\left\{\widehat{\bs{\theta}}_n(\mathbf{W}) -\bs{\theta}_0^{\ast}(\mathbf{W})\right\}^2 + P\left\{\bs{\theta}_0^{\ast}(\mathbf{W})-\bs{\theta}_0(\mathbf{W})\right\}^2
	\xrightarrow{\as} 0.
\end{eqnarray*}
\end{proof}

\bibliographystyle{asa}
\bibliography{references}

\end{document}